\documentclass[dvipsnames,letterpaper,11pt]{article}
\usepackage[margin=1in]{geometry}
\usepackage[utf8]{inputenc}
\usepackage{graphicx,fullpage,paralist}
\usepackage{amsmath,amssymb,amsthm}
\usepackage{bm}
\usepackage{array}
\newcolumntype{C}{>{$}c<{$}}
\usepackage{comment,hyperref}
\usepackage{booktabs}
\usepackage{subcaption}
\usepackage[font={footnotesize}]{caption}
\usepackage{mathtools}
\usepackage{thmtools}
\usepackage{tikz}
\usepackage{tikz-network}
\usepackage{pgfplots}
\usepackage{varwidth}
\usepackage{xstring}
\pgfplotsset{compat=1.10}
\usepgfplotslibrary{fillbetween}
\usetikzlibrary{backgrounds}
\usetikzlibrary{patterns}
\usetikzlibrary{arrows}
\usetikzlibrary{calc}
\usetikzlibrary{matrix,decorations.pathreplacing,positioning,fit}
\usepackage{enumitem}
\setlist[itemize]{nosep}
\setlist[enumerate]{nosep}
\usepackage{blkarray}
\usepackage{csquotes}
\usepackage{stmaryrd}
\usepackage{natbib}

\definecolor{color0}{RGB}{230,159,0}
\definecolor{color1}{RGB}{86,180,233}
\definecolor{color2}{RGB}{0,158,115}
\definecolor{color3}{RGB}{240,228,66}
\definecolor{color4}{RGB}{0,114,178}
\definecolor{color5}{RGB}{213,94,0}
\definecolor{color6}{RGB}{204,121,167}

\newcommand{\PP}{\mathbb{P}}

\newcommand{\calG}{\mathcal{G}}

\newcommand{\calP}{\mathcal{P}}

\newcommand{\calS}{\mathcal{S}}

\newcommand{\calA}{\mathcal{A}}

\newcommand{\RR}{\mathbb{R}}
\newcommand{\NN}{\mathbb{N}}

\newcommand{\ie}{i.e.,\xspace}
\newcommand{\eg}{e.g.,\xspace}

\newcommand{\bfb}{\mathbf{b}}
\newcommand{\bfA}{\mathbf{A}}
\newcommand{\bfd}{\mathbf{d}}
\newcommand{\bfpi}{\bm{\pi}}
\newcommand{\bfchi}{\bm{\chi}}
\newcommand{\bfrho}{\bm{\rho}}
\newcommand{\defas}{\ensuremath{\coloneqq}}

\newtheorem{theorem}{Theorem}
\newtheorem{lemma}{Lemma}

\newtheorem{claim}{Claim}

\usepackage[textsize=tiny,textwidth=2cm]{todonotes}

\tikzstyle{component}=[draw opacity=0.4,draw=black,line width=1.0cm,line cap=round,line join=round]

\usepackage{multirow}
\usepackage[ruled]{algorithm2e}
\usepackage{xcolor}

\newlength{\algofontsize}
\hypersetup{
	colorlinks,
	linkcolor={red!50!black},
	citecolor={blue!50!black},
	urlcolor={blue!80!black}
}

\def\equationautorefname~#1\null{(#1)\null}

\usepackage{etoolbox}
\makeatletter
\patchcmd{\hyper@makecurrent}{%
    \ifx\Hy@param\Hy@chapterstring
        \let\Hy@param\Hy@chapapp
    \fi
}{%
    \iftoggle{inappendix}{%
        \@checkappendixparam{chapter}%
        \@checkappendixparam{section}%
        \@checkappendixparam{subsection}%
        \@checkappendixparam{subsubsection}%
        \@checkappendixparam{paragraph}%
        \@checkappendixparam{subparagraph}%
    }{}%
}{}{\errmessage{failed to patch}}

\newcommand*{\@checkappendixparam}[1]{%
    \def\@checkappendixparamtmp{#1}%
    \ifx\Hy@param\@checkappendixparamtmp
        \let\Hy@param\Hy@appendixstring
    \fi
}
\makeatletter

\newtoggle{inappendix}
\togglefalse{inappendix}

\apptocmd{\appendix}{\toggletrue{inappendix}}{}{\errmessage{failed to patch}}

\begin{document}
	
\title{Impartial Rank Aggregation}

\author{Javier Cembrano
	\thanks{Institut für Mathematik, Technische Universität Berlin, Germany}
	\and Felix Fischer			
	\thanks{School of Mathematical Sciences, Queen Mary University of London, UK}
	\and Max Klimm
	\thanks{Institut für Mathematik, Technische Universität Berlin, Germany}
	}

\date{}
\maketitle

\begin{abstract}
We study functions that produce a ranking of~$n$ individuals from~$n$ such rankings and are impartial in the sense that the position of an individual in the output ranking does not depend on the input ranking submitted by that individual. When~$n\geq 4$, two properties concerning the quality of the output in relation to the input can be achieved in addition to impartiality: individual full rank, which requires that each individual can appear in any position of the output ranking; and monotonicity, which requires that an individual cannot move down in the output ranking if it moves up in an input ranking. When~$n\geq 5$, monotonicity can be dropped to strengthen individual full rank to weak unanimity, requiring that a ranking submitted by every individual must be chosen as the output ranking. Mechanisms achieving these results can be implemented in polynomial time. Both results are best possible in terms of their dependence on~$n$. The second result cannot be strengthened further to a notion of unanimity that requires agreement on pairwise comparisons to be preserved.
\end{abstract}

\section{Introduction}

Decisions in modern democracies are made through a process of voting on sets of alternatives or individuals. The mathematical treatment of voting processes dates back almost to the beginning of modern democracies themselves, notably to work of \citet{borda1784memoire} and \citet{condorcet1785essai}.
Formally, voting can be viewed as the task of aggregating a set of preference rankings, one for each voter, into a single \emph{social} preference ranking.
Here and in the following, a ranking is a permutation $\pi\colon \{0,1,\dots,n-1\} \to \{0,1,\dots,n-1\}$ of a set of alternatives numbered $0,1,\dots,n-1$. The understanding is that alternative~$i$ is preferred over alternative~$j$ if and only if $\pi^{-1}(i)<\pi^{-1}(j)$, where $\pi^{-1}(k)$ denotes the position of alternative~$k$ in ranking~$\pi$. Writing $\calP_n$ for the set of permutations of~$n$ alternatives, a rank aggregation rule for~$n$ alternatives and~$m$ voters is a function $f\colon \calP_n^m \to \calP_n$ mapping a profile $\bfpi = (\pi_0,\dots,\pi_{m-1})$ of rankings to a social ranking $f(\bfpi)$.
A classic result of \citet{arrow1963social} established that rank aggregation rules are severely limited, in that every rule must violate one of three natural axioms: unanimity, non-dictatorship, or independence of irrelevant alternatives. \emph{Unanimity} requires that the aggregation retains unanimous pairwise orderings, \ie that $(f(\bfpi))^{-1}(j)<(f(\bfpi))^{-1}(k)$ whenever $\pi_{i}^{-1}(j)<\pi_i^{-1}(k)$ for all $i=0,\dots,m-1$. \emph{Non-dictatorship} prevents the existence of a voter whose preference always prevails, \ie a voter~$i$ such that $f(\bfpi)=\pi_i$ for all $\bfpi\in\calP_n^m$.
\emph{Independence of irrelevant alternatives} finally requires that the relative order of alternatives~$j$ and~$k$ in the social ranking only depends on their relative order in all individual rankings.
Most voting systems used in democracies around the globe today satisfy unanimity and non-dictatorship, and resolve \citeauthor{arrow1950difficulty}'s impossibility result at the expense of independence of irrelevant alternatives. Indeed, independence of irrelevant alternatives is deemed ``by far the most controversial'' among the three axioms \citep[cf.][]{maskin2020modified}.

\citeauthor{arrow1950difficulty}'s axioms are particularly compelling when voters are balloted on abstract alternatives such as different policy options. When the alternatives are the voters themselves, \ie when the sets of voters and alternatives coincide, additional considerations may come into play. For example, when balloting the members of an organization on a priority list for promotions, it is reasonable to assume that each member would primarily be concerned about their own position in the output ranking. A voter may be incentivized to misrepresent their input ranking in order to improve their own position in the output, thus compromising the entire process.
The axiom of \emph{impartiality} makes such manipulation impossible by requiring that a voter cannot influence their position in the output ranking, \ie that, for an aggregation rule $f\colon\calP_n^n\to\calP_n$, the equation $(f(\bfpi))^{-1}(i) = (f(\bfpi'))^{-1}(i)$ holds whenever $\pi_j=\pi_j'$ for all $j \neq i$.
Similar impartiality axioms were previously proposed for the allocation of a divisible resource and the selection of one or more of the voters.

It is easy to see that impartiality is stronger than non-dictatorship, in the sense that every aggregation rule failing non-dictatorship also fails impartiality. It thus follows directly from \citeauthor{arrow1950difficulty}'s impossibility that there is no aggregation rule satisfying impartiality, unanimity, and independence of irrelevant alternatives. Compared to \citeauthor{arrow1950difficulty}'s setting relatively little is known regarding impartial ranking, but we will see that it holds significant challenge both conceptually and mathematically. Like impartial division and selection it is a foundational problem with a large degree of generality and numerous applications also outside democratic decision making, for example in peer review and peer appraisal, apportionment of credit, and collaborative filtering.

\paragraph{Our Contribution}

Impartiality turns out to be a relatively demanding axiom that on its own is incompatible with unanimity; we prove this as \autoref{thm:imp-unanimity} 
at the end of the paper. Compared with \citeauthor{arrow1950difficulty}'s impossibility, strengthening non-dictatorship to impartiality thus renders independence of irrelevant alternatives redundant. Our impossibility result holds for all $n\geq 2$, whereas that of \citeauthor{arrow1950difficulty} holds for $n\geq 2$ and $m\geq 3$.

Motivated by the impossibility result we relax the axioms. In doing so we draw inspiration from an earlier version of Arrow's result~\citep{arrow1950difficulty}, which used two separate axioms in place of unanimity: \emph{monotonicity}, which requires that an alternative that moves up in input rankings while everything else stays the same cannot move down in the output ranking, and \emph{non-imposition}, which requires every ranking to be obtained as the output for some profile of input rankings.\footnote{\citeauthor{arrow1950difficulty} initially proved that monotonicity, non-imposition, non-dictatorship, and independence of irrelevant alternatives are incompatible. However, monotonicity, non-imposition, and independence of irrelevant alternatives together imply unanimity, so this is a weaker version of the result commonly known as Arrow's impossibility.}

We show as \autoref{thm:monotonicity} that impartiality is compatible with monotonicity and a weaker version of non-imposition we call \emph{individual full rank}, which requires that every voter can appear in every position of the output ranking. 
The result is constructive and holds for all $n \geq 4$, which is tight: by \autoref{thm:imp-ifr-n23} impartiality and individual full rank are incompatible for $n\in\{2,3\}$.

We then show as \autoref{thm:weak-unanimity} that monotonicity can be dropped to strengthen individual full rank to non-imposition. The mechanism we construct in fact satisfies a stronger version of non-imposition we call \emph{weak unanimity}, which requires that a ranking on which all agents agree must be returned as the output, \ie that $f(\bfpi)=\hat{\pi}$ whenever $\pi_i = \hat{\pi}$ for all $i=0,\dots,n-1$. 
The result holds when $n\geq 5$, which is again tight: for $n\in \{2,3\}$ impartiality and weak unanimity are incompatible by \autoref{thm:imp-ifr-n23}, for $n=4$ we can show their incompatibility computationally.

We do not know whether there exists a mechanism satisfying impartiality, monotonicity, and non-imposition. This question is interesting, and likely difficult. Indeed the large number of distinct outcomes required for non-imposition makes monotonicity much more challenging to work with, both for our techniques and in general.

\paragraph{Our Techniques}

Both of our positive results are obtained via reductions to combinatorial problems related to graph or hypergraph coloring.

To obtain a mechanism satisfying impartiality, monotonicity, and individual full rank, and thus prove \autoref{thm:monotonicity}, we study a class of mechanisms in which the input ranking of each agent~$i$ is interpreted as a message in $\{0,1\}$. For each choice of such a message, and every other agent $j$, $i$ blocks $j$ from occupying a subset of positions in the output ranking. Agents are then assigned to a unique unblocked position, or to a default position in case all other positions are blocked.
By further requiring the structure of blocked positions to be symmetric, we can define the mechanism completely by a multigraph with vertex set $\{0,1,\dots,n-1\}$. The mechanism is well-defined and satisfies individual full rank whenever the graphs avoid a certain condition on triples of vertices; monotonicity can be expressed as a condition on the existence or non-existence of certain edges. We prove the existence of multigraphs satisfying both types of conditions via a probabilistic argument and the Lov\'asz local lemma, along with specific constructions for small values of~$n$. The argument can be derandomized and implemented efficiently using known techniques, and doing so yields a polynomial-time mechanism.

To obtain a mechanism satisfying impartiality and weak unanimity, and thus prove \autoref{thm:weak-unanimity}, we consider mechanisms that determine the output ranking based on the input rankings of three decisive agents. The position of each decisive agent is determined by the other two; the position of the remaining agents is decided by the decisive ones.
The constraint that only a single agent can be assigned to each position can be reduced to the existence of a proper coloring of a tripartite graph. Each vertex of the graph corresponds to a ranking profile for two decisive agents, its color to the position of the respective third decisive agent in the output ranking. Weak unanimity thus fixes the colors of vertices corresponding to profiles in which two decisive agents cast the same ranking. The existence of a proper coloring is shown constructively, so the corresponding mechanisms can again be implemented in polynomial time.

\paragraph{Related Work}

Impartiality as a formal property in cooperative decision making was first considered by \citet{de2008impartial}, for mechanisms that allocate a divisible resource among the members of a group based on each member's opinion regarding the others' relative entitlement. Impartial selection, proposed by \citet{alon2011sum} and \citet{holzman2013impartial}, is the subject of a growing literature in economics and computer science~\citep[\eg][]{tamura2014impartial,mackenzie2015symmetry,tamura2016characterizing,bousquet2014near,fischer2015optimal,bjelde2017impartial,caragiannis2022impartial,cembrano2022impartial-ec,cembrano2023single}. 

In rank aggregation, impartiality can be defined in more than one way. \citet{berga2014impartial} and \citet{alcalde2022impartial} propose a strict notion of impartiality by which an agent cannot influence its comparison with any other agent. This notion of impartiality is incompatible with individual full rank, but allows some positive results to be achieved for all agents except one by placing that agent in a fixed position. %
Even the strong notion of impartiality is compatible with monotonicity and neutrality~\citep{berga2014impartial}, which shows that impartiality and monotonicity do not entirely prevent the transfer of information about all agents from input to output.
For the notion of impartiality with regard to rank we consider here, \citet{kahng2018ranking} proposed a randomized method to approximate a given target rule, with high probability and for a particular measure of approximation. An axiomatic characterization of one candidate target rule, that of Borda, was given by \citet{ohseto2007characterization}. In computer science, an extensive literature exists on efficient algorithms for rank aggregation and related problems~\citep[\eg][]{dwork01,alon06,kenyon07,ailon08}.
Methods that produce a ranking from inputs other than rankings have been studied for example by \citet{altman2007incentive,altman2008axiomatic} and \citet{ng2003exclusion}. \citeauthor{altman2008axiomatic} considered methods that use dichotomous input preferences and satisfy incentive constraints, monotonicity, and axioms relating to mutual agreement among agents. \citeauthor{ng2003exclusion} showed that for methods based on numerical evaluations, and counter to intuition, self-evaluations must be allowed in order to achieve desirable properties including monotonicity and unanimity.

\section{Preliminaries}
\label{sec:preliminaries}

Let $\NN=\{0,1,2,\ldots\}$. For $n\in\NN$, let $[n]=\{0,1,\ldots,n-1\}$,\footnote{This notation is somewhat non-standard. We use it, and index $n$-dimensional vectors by numbers in $[n]$, for convenience of notation involving modular arithmetic. For $r\in \NN$ and $a,~b\in \RR$, we also denote $a+_r b=a+b \pmod r$ and $a-_r b=a-b \pmod  r$, respectively.} and let $\calP_n$ be the set of permutations of $[n]$.
We use $[n]$ to represent a set of $n$ agents, and permutation $\pi\colon[n]\to [n]$ to represent a linear order or \textit{ranking} of the set of agents. Thus, for each position $k\in[n],~\pi(k)$ is the agent in the $k$th position of ranking $\pi$, and for each agent $i\in[n],~ \pi^{-1}(i)$ is the position of agent $i$ in ranking $\pi$.
We often use the one-line notation of a permutation and write $\pi = (\pi(0)\ \ \pi(1)\ \ \cdots\ \ \pi(n-1))$ for $\pi\in \calP_n$.
An \textit{$n$-ranking mechanism} is a function $f\colon\calP_n^n \to \calP_n$ that takes a \emph{ranking profile} $\bfpi=(\pi_0,\pi_1,\ldots,\pi_{n-1})\in \calP_n^{n}$ and maps it to a single \emph{social ranking} $f(\bfpi)$.
For $n\in\NN$, $\bfpi\in\calP_n^n$, and $i\in [n]$,
we denote by $\bfpi_{-i} = (\pi_0, \pi_1, \ldots, \pi_{i-1}, \pi_{i+1}, \ldots, \pi_{n-1}) \in \calP_n^{n-1}$ the profile of rankings for all agents except~$i$, and for any $\tilde{\pi}\in\calP_n$ write $f(\tilde{\pi},\bfpi_{-i})=f(\pi_0, \pi_1, \ldots, \pi_{i-1},\tilde{\pi},\pi_{i+1}, \ldots, \pi_{n-1})$. 
We say that an $n$-ranking mechanism $f$ satisfies
\begin{itemize}
    \item \emph{impartiality} if the position of an agent in the output ranking does not depend on its input ranking, \ie if for all $i\in [n],~\bfpi\in\calP_n^n$, and $\tilde{\pi}\in\calP_n$, $(f(\bfpi))^{-1}(i)=(f(\tilde{\pi}, \bfpi_{-i}))^{-1}(i)$;
    \item \emph{individual full rank} if every agent can appear in any position of the output ranking, \ie if for every $j\in [n]$ and $k\in [n]$, there exists $\bfpi\in\calP_n^{n}$ such that $(f(\bfpi))(k)=j$;
    \item \emph{weak unanimity} if it outputs the agreed ranking whenever all agents agree on a ranking, \ie if whenever there exists $\hat{\pi}\in \calP_n$ such that $\pi_i=\hat{\pi}$ for all $i\in [n]$, $f(\bfpi)=\hat{\pi}$;
    \item \emph{unanimity} if it preserves agreement by all agents on the relative order between two agents, \ie if whenever there exist $j_0, j_1\in [n]$ such that $\pi_i^{-1}(j_0) < \pi_i^{-1}(j_1)$ for all $i\in [n]$, then $(f(\bfpi))^{-1}(j_0) < (f(\bfpi))^{-1}(j_1)$; and
    \item \emph{monotonicity} if an agent that moves up in a single input ranking, while everything else stays the same, cannot move down in the output ranking, \ie if for all $\bfpi\in\calP_n^n$, $i, j_0\in [n]$, and $\tilde{\pi}\in\calP_n$ 
    such that for all $j_1\in[n]$ and $j_2\in[n]\setminus\{j_0\}$,
    \vspace*{-1ex}
    \begin{equation*}
        \tilde{\pi}^{-1}(j_1)<\tilde{\pi}^{-1}(j_2) \quad \text{whenever}\quad \pi_i^{-1}(j_1)<\pi_i^{-1}(j_2),  \vspace*{-1ex}
    \end{equation*}
    $(f(\tilde{\pi},\bfpi_{-i}))^{-1}(j_0)\leq (f(\bfpi))^{-1}(j_0)$.\footnote{\citeauthor{arrow1950difficulty} uses a notion of monotonicity that binds also when an agent moves up in multiple input rankings at the same time. It is easy to see that the two notions are equivalent.}
\end{itemize}

Unanimity implies weak unanimity, which in turn implies individual full rank. We state this formally as the following lemma, the straightforward proof can be found in \autoref{app:pf-lem-relation-axioms}.
\begin{lemma}  \label{lem:relation-axioms}
Let $n\in\NN$, $n\geq 2$. Let $f$ be an $n$-ranking mechanism. If $f$ satisfies unanimity, then it satisfies weak unanimity. If $f$ satisfies weak unanimity, then it satisfies individual full rank.
\end{lemma}

\section{Monotonicity and Individual Full Rank}
\label{sec:monotonicity}

We begin by studying monotone ranking mechanisms. Impartiality and monotonicity are trivial to achieve, and hold for example for a mechanism that produces the same output for every input. We show that they do not preclude mechanisms that produce a much richer class of outcomes: for $n\geq 4$, impartiality and monotonicity are compatible with individual full rank. The condition that $n\geq 4$ cannot be removed: as we will see later, impartiality and individual full rank are incompatible for $n\leq 3$.
\begin{theorem}  \label{thm:monotonicity}
Let $n\geq 4$. Then there exists an $n$-ranking mechanism that satisfies impartiality,  monotonicity, and individual full rank.
\end{theorem}

We explicitly construct mechanisms that satisfy the theorem. For each~$i$ these mechanisms use a single bit of information about the input ranking submitted by agent~$i$, namely whether it ranks a specific other agent $\rho_i$ above or below~$i$ itself.
For $n\in\NN$, let $R_n\subseteq[n]^n$ be the set of vectors with the $i$th component different to $i$, for every $i\in [n]$. Formally,
\[
    R_n = \{\bfrho\in [n]^n\colon \rho_i\not= i \text{ for every } i\in [n] \}.
\]
Let $\chi$ denote the indicator function for logical propositions. For $i\in [n]$, $\pi\in\calP_n$, and $\bfrho\in R_n$, let $\chi_i(\pi,j)=\chi(\pi^{-1}(j)<\pi^{-1}(i))$. For $\bfpi\in\calP_n^n$ and $\bfrho\in R_n$, let $\bfchi(\bfpi,\bfrho) = (\chi_0(\pi_0,\rho_0),\chi_1(\pi_1,\rho_1),\ldots, \chi_{n-1}(\pi_{n-1},\rho_{n-1}))$. 
For $g\colon \{0,1\}^n \to \calP_n$ and $\bfrho\in R_n$, let $f_{g,\bfrho}\colon \calP_n^n \to \calP_n$ be the $n$-ranking mechanism such that $f_{g,\bfrho}(\bfpi)=g(\bfchi(\bfpi,\bfrho))$ for every $\bfpi \in \calP^n_n$. For $i\in[n]$ and $\bfb\in \{0,1\}^n$, we will refer to $b_i\in \{0,1\}$ as a message sent to the mechanism by agent~$i$, and analogously to ranking profiles write $\bfb_{-i} = (b_0,b_1,\ldots,b_{i-1},b_{i+1},\ldots,b_{n-1}) \in \{0,1\}^{n-1}$ for the profile of messages for all agents except~$i$ and denote $g(b',\bfb_{-i})=g(b_0,b_1,\ldots,b_{i-1},b',b_{i+1},\ldots,b_{n-1})$.

Since a mechanism $f_{g,\bfrho}$ depends on the ranking submitted by agent $i$ only insofar as that ranking ranks $\rho_i$ above or below $i$, a simple condition is sufficient for monotonicity: for any fixed profile of the other agents, whenever agent $i$ changes from a ranking where~$\rho_i$ is below~$i$ to a ranking where~$\rho_i$ is above~$i$, $\rho_i$ cannot descend in the output ranking. 
The following lemma establishes this condition formally, along with conditions that guarantee impartiality and individual full rank. We prove the lemma in \autoref{app:pf-lem-suf-conds-g}.
\begin{lemma}
\label{lem:suf-conds-g}
    Let $n\geq 4$. Let $g\colon \{0,1\}^n \to \calP_n$ and $\bfrho\in R_n$ such that
    \begin{enumerate}[label=(\roman*)]
        \item for every $i\in [n]$ and $\bfb\in \{0,1\}^n$, $(g(0,\bfb_{-i}))^{-1}(i) = (g(1,\bfb_{-i}))^{-1}(i)$;  \label{item:g-impartial}
        \item for every $j,k\in[n]$, there exists $\bfb\in \{0,1\}^n$ such that $(g(\bfb))(k) = j$;  and \label{item:g-ifr}
        \item for every $i\in [n]$ and $\bfb\in \{0,1\}^{n}$, $(g(1,\bfb_{-i}))^{-1}(\rho_i) \leq (g(0,\bfb_{-i}))^{-1}(\rho_i)$.  \label{item:g-monotone}
    \end{enumerate}
    Then, $f_{g,\bfrho}$ satisfies impartiality, monotonicity, and individual full rank.
\end{lemma}

In what follows, we construct functions $g$ so that for some $\bfrho\in R_n$, $g$ has codomain $\calP_n$ and~$g$ and~$\bfrho$ satisfy the conditions of \autoref{lem:suf-conds-g}. The case where $n=4$ requires a dedicated construction and is covered by the following lemma, which we prove in \autoref{app:pf-lem-existence-g-n4}.
\begin{lemma}  \label{lem:existence-g-n4}
There exist $g\colon\{0,1\}^4\to\calP_4$ and $\bfrho\in R_4$ satisfying Conditions~\ref{item:g-impartial}-\ref{item:g-monotone} of \autoref{lem:suf-conds-g}.
\end{lemma}

For $n\geq 5$, $g$ is constructed in a uniform way. Each agent $j\in[n]$ has a \textit{default position}, which is equal to $j$. For each pair of distinct agents $i,j$ and each message $b\in\{0,1\}$, we define a \textit{blocking set} $S^b_{ij}\subset[n]\setminus\{i,j\}$ of positions that $i$ blocks for $j$ by sending $b$ to the mechanism. For an agent $j$ and a message profile $\bfb$, a position is unblocked if it is \emph{not} contained in any of the blocking sets $S^b_{ij}$ for $i\in[n]\setminus\{j\}$. We note that for each agent~$j$ the default position~$j$ is always unblocked, and will define the blocking sets in such a way that at most one additional position is unblocked for~$j$ for any message profile. We then assign each agent to an unblocked position, giving preference to positions other than the default ones: for a given message profile, if an agent has an unblocked position other than its default position, we assign it to that position; otherwise, we assign it to its default position. 
To obtain a well-defined mechanism we must guarantee for every message profile that each position is assigned to only one agent. Individual full rank requires that for every pair of an agent and a position, the position is unblocked for the agent for some message profile of the other agents. Monotonicity finally requires that the positions agent~$i$ blocks from agent~$\rho_i$ by sending message~$0$ have smaller indices than those it blocks by sending message~$1$.

An example of blocking sets satisfying these conditions for $n=6$ and $\bfrho=(1,2,3,4,5,0)$ is shown in Part~(a) of \autoref{fig:blocking-sets}.
    \begin{figure}[t]

        \tikzset{toprule/.style={%
            execute at end cell={%
                \draw [line cap=rect,#1] (\tikzmatrixname-\the\pgfmatrixcurrentrow-\the\pgfmatrixcurrentcolumn.north west) -- (\tikzmatrixname-\the\pgfmatrixcurrentrow-\the\pgfmatrixcurrentcolumn.north east);%
                }
            },
            bottomrule/.style={%
            execute at end cell={%
                \draw [line cap=rect,#1] (\tikzmatrixname-\the\pgfmatrixcurrentrow-\the\pgfmatrixcurrentcolumn.south west) -- (\tikzmatrixname-\the\pgfmatrixcurrentrow-\the\pgfmatrixcurrentcolumn.south east);%
                }
            }
        }   
        \hspace{1em}
\begin{subfigure}[c]{0.5\textwidth}
\footnotesize
\renewcommand{\arraystretch}{0.85}
\begin{tabular}{CC@{\hspace{1em}}*{6}{C}}
\toprule
  && \multicolumn{6}{C}{j} \\[-.5ex]
\cmidrule{3-8}
   i & b & 0 & 1 & 2 & 3 & 4 & 5 \\[-.5ex]
\midrule
\multirow{2}{*}{0}
        & 0 & \phantom{0123} & & 35 & 24 & 3 & 2 \\   
        & 1 & & \makebox[0pt][c]{$2345$} & 14 & 15 & 125 & 134 \\[-.5ex]
\midrule
\multirow{2}{*}{1}
        & 0 & 24 & \phantom{0123} & 0 & 5 & 0 & 3 \\   
        & 1 & 35 & & 345 & 024 & 235 & 024 \\[-.5ex]
    \midrule
    \multirow{2}{*}{2}
        & 0 & 3 & 35 & \phantom{0123} & 01 & & 1 \\   
        & 1 & 145 & 04 & & 45 & 0135 & 034 \\[-.5ex]
        \midrule
        \multirow{2}{*}{3}
        & 0 & 45 & 4 & 4 & \phantom{0123} & 012 & 0 \\   
        & 1 & 12 & 025 & 015 & & 5 & 124 \\ [-.5ex]
        \midrule
        \multirow{2}{*}{4}
        & 0 & 135 & 05 & 5 & 05 & \phantom{0123} & 0123 \\   
        & 1 & 2 & 23 & 013 & 12 & & \\[-.5ex] 
        \midrule
        \multirow{2}{*}{5}
        & 0 & & 24 & 13 & 2 & 1 & \phantom{0123} \\   
        & 1 & 1234 & 03 & 04 & 014 & 023 & \\[-.5ex]
\bottomrule \\[-2ex]
\end{tabular}
\caption{}
\end{subfigure}
\hfill 
\begin{subfigure}[c]{0.4\textwidth}
\centering
\begin{tikzpicture}

        \foreach \i/\c in {0/color0,1/color1,2/color2,3/color3,4/color4,5/color5}{\draw (\i*360/6:2.5cm) node[circle,fill=\c](\i){\i};}
\draw[-,draw,ultra thick,color=color4,bend left=0] (0) edge (1);
\draw[-,draw,ultra thick,color=color1,bend left=0] (0) edge (2);
\draw[-,draw,ultra thick,color=color2,bend left=-5] (0) edge (3);
\draw[-,draw,ultra thick,color=color4,bend left=5] (0) edge (3);
\draw[-,draw,ultra thick,color=color1,bend left=-5] (0) edge (4);
\draw[-,draw,ultra thick,color=color3,bend left=5] (0) edge (4);
\draw[-,draw,ultra thick,color=color3,bend left=-5] (0) edge (5);
\draw[-,draw,ultra thick,color=color4,bend left=5] (0) edge (5);
\draw[-,draw,ultra thick,color=color5,bend left=0] (1) edge (2);
\draw[-,draw,ultra thick,color=color2,bend left=0] (1) edge (3);
\draw[-,draw,ultra thick,color=color3,bend left=-5] (1) edge (4);
\draw[-,draw,ultra thick,color=color5,bend left=5] (1) edge (4);
\draw[-,draw,ultra thick,color=color2,bend left=-5] (1) edge (5);
\draw[-,draw,ultra thick,color=color4,bend left=5] (1) edge (5);
\draw[-,draw,ultra thick,color=color0,bend left=-5] (2) edge (3);
\draw[-,draw,ultra thick,color=color5,bend left=5] (2) edge (3);
\draw[-,draw,ultra thick,color=color3,bend left=0] (2) edge (4);
\draw[-,draw,ultra thick,color=color0,bend left=-5] (2) edge (5);
\draw[-,draw,ultra thick,color=color4,bend left=5] (2) edge (5);
\draw[-,draw,ultra thick,color=color0,bend left=0] (3) edge (4);
\draw[-,draw,ultra thick,color=color1,bend left=-5] (3) edge (5);
\draw[-,draw,ultra thick,color=color4,bend left=5] (3) edge (5);
\end{tikzpicture}
\caption{}
\end{subfigure}
\caption{(a)~Blocking sets satisfying \autoref{lem:existence-blocking-sets} for $n=6$ and $\bfrho=(1,2,3,4,5,0)$, shown by concatenating their elements for compactness, and (b)~the corresponding multigraph. Vertex~$i$ in the graph is labeled~$i$ and drawn in color~$i$. For $i,j,k\in[n]$, an edge $\{j,k\}$ of color~$i$ exists if and only if $k\in S^0_{ij}$ (and thus also $j\in S^0_{ik}$). %
}
\label{fig:blocking-sets}
\end{figure}
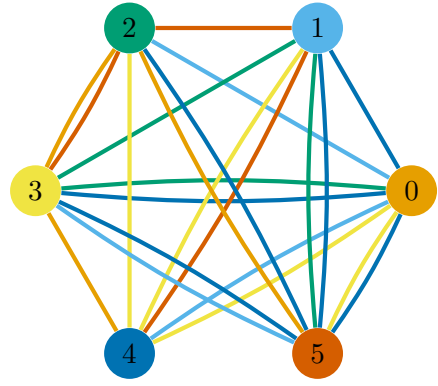
For instance, for profile $\bfb=(0,0,1,1,1,0)$, $S^0_{1,0}=\{2,4\}$, $S^1_{2,0}=\{1,4,5\}$, $S^1_{3,0}=\{1,2\}$, $S^1_{4,0}=\{2\}$, and $S^0_{5,0}=\emptyset$, so the set of positions blocked for agent~$0$ is $S^0_{1,0}\cup S^1_{2,0}\cup S^1_{3,0}\cup S^1_{4,0}\cup S^0_{5,0}=\{1,2,4,5\}$. This leaves position $3$ unblocked, in addition to agent $0$'s default position~$0$, and agent~$0$ is assigned position~$3$. Analogously, for the other agents,
\begin{gather*}
    \bigcup_{i\in [6]\setminus \{1\}} S^{b_i}_{i1} =\{0,2,3,4,5\},\quad \bigcup_{i\in [6]\setminus \{2\}} S^{b_i}_{i2} =\{0,1,3,5\},\quad \bigcup_{i\in [6]\setminus \{3\}} S^{b_i}_{i3} =\{1,2,4,5\},\\
    \bigcup_{i\in [6]\setminus \{4\}} S^{b_i}_{i4} =\{0,1,3,5\},\quad \bigcup_{i\in [6]\setminus \{5\}} S^{b_i}_{i5} =\{0,1,2,3,4\},
\end{gather*}
so $g(\bfb)=(3\ \ 1\ \ 4\ \ 0\ \ 2\ \ 5)$. If $b'=1$, then $g(b',\bfb_{-1})=(0\ \ 1\ \ 2\ \ 3\ \ 4\ \ 5)$, which does not move agent~$2$ down and is thus compatible with monotonicity.

In the example, a unique choice was available for the position of each agent, and no position was assigned more than once. In addition, we observed a particular change of message compatible with monotonicity. The following lemma establishes that blocking sets as in the example exist more generally, and provides sufficient conditions for the resulting mechanisms to be well-defined and satisfy impartiality, individual full rank, and monotonicity. 
\begin{lemma}  \label{lem:existence-blocking-sets}
    Let $n\geq 5$. Then there exist $\bfrho\in R_n$ and sets $S^b_{ij}\subset [n]\setminus \{i,j\}$ for all $i,j\in [n]$ with $i\neq j$ and $b\in\{0,1\}$ such that the following hold:
    \begin{enumerate}[label=(\roman*)]
        \item for every $i,j\in [n]$ with $i\not=j$, $S^0_{ij} \cap S^1_{ij} = \emptyset$ and $S^0_{ij} \cup S^1_{ij} = [n]\setminus \{i,j\}$;  \label{item:ifr}
        \item for every $j\in [n]$ and $\bfb\in \{0,1\}^n$, $\big| \bigcup_{i\in [n]\setminus \{j\}} S^{b_i}_{ij} \big| \in \{n-2,n-1\}$;  \label{item:one-pos-per-agent}
        \item for every $j\in [n]$ and $\bfb\in \{0,1\}^n$, if $\bigcup_{i\in [n]\setminus \{j\}} S^{b_i}_{ij} = [n]\setminus \{j\}$, then for every $j'\in [n]\setminus \{j\}$ it holds that $j\in \bigcup_{i\in [n]\setminus \{j'\}} S^{b_{i}}_{ij'}$; \label{item:one-agent-per-pos-def}
        \item for every $j, k\in [n]$ with $j\not= k$ and $\bfb\in \{0,1\}^n$, if $\bigcup_{i\in [n]\setminus \{j\}} S^{b_i}_{ij} = [n]\setminus \{j,k\}$, then for every $j'\in [n]\setminus \{j,k\}$ it holds that $k\in \bigcup_{i\in [n]\setminus \{j'\}} S^{b_{i}}_{ij'}$; and  \label{item:one-agent-per-pos-non-def}
        \item for every $i\in [n]$, $S^0_{i\rho_i} = \{ k\in [n]\setminus \{i\}\colon k<\rho_i\}$ and $S^1_{i\rho_i} = \{ k\in [n]\setminus \{i\}\colon k>\rho_i\}$.\label{item:mon}
    \end{enumerate}
\end{lemma}

We defer the proof of the lemma to \autoref{app:pf-lem-existence-blocking-sets} but briefly explain the main ideas here. Condition~\ref{item:ifr} states that an agent never blocks its own default position for any of the other agents, and for each of the other agents blocks each remaining position by exactly one of its two messages. We will in fact achieve this in a symmetric way, by guaranteeing that agent~$i$ blocks position~$k$ from agent~$j$ if and only if it blocks position~$j$ from agent~$k$. This immediately implies Condition~\ref{item:one-agent-per-pos-def}, since whenever every position $k\in [n]\setminus \{j\}$ is blocked for $j$, position $j$ is then blocked for every agent other than~$j$. It also makes agents and positions interchangeable, so that Conditions~\ref{item:one-pos-per-agent} and~\ref{item:one-agent-per-pos-non-def} become equivalent: 
no position is unblocked for more than one agent other than the agent whose default position it is if and only if no agent has more than one unblocked position other than its default position. Condition~\ref{item:mon} ensures that each agent $i$ blocks, for agent $\rho_i$, all positions with a higher index than the default position of that agent by sending message $0$, and all positions with a lower index than the default position by sending message $1$; this turns out to be sufficient for Condition~\ref{item:g-monotone} of \autoref{lem:suf-conds-g}.

To prove the lemma we establish a correspondence between the blocking sets of each agent~$i$ when it sends message~$0$ and edges of color~$i$ in a multigraph in which vertices correspond to agents. Specifically, an edge of color $i$ between vertices $j$ and $k$ corresponds to agent $i$ blocking position $k$ from agent $j$ (and position $j$ from agent $k$) when sending message $0$.
We then show that the existence of blocking sets satisfying Conditions~\ref{item:one-pos-per-agent} and~\ref{item:one-agent-per-pos-non-def} is equivalent to the existence of a multigraph avoiding certain subgraphs, namely triples of vertices $(i,j,k)$ such that for every color different from $i$, $j$, and $k$, we have either both edges of the path $(i,j,k)$ or none of them. 
Condition~\ref{item:mon} finally becomes a condition on the existence or non-existence of certain edges: for $j\neq i$, there is an edge $\{j,\rho_i\}$ of color $i$ if and only if $j<\rho_i$.
Existence of the graphs in question is shown using a probabilistic argument and the Lov\'asz local lemma, along with explicit constructions for small values of $n$. By a result of \citet{chandrasekaran2013deterministic}, the graphs can be found efficiently.
Part~(b) of \autoref{fig:blocking-sets} shows the multigraph corresponding to the blocking sets in Part~(a) of that figure.

\autoref{lem:existence-blocking-sets} in hand, we are ready to prove \autoref{thm:monotonicity}. This proof is deferred to \autoref{app:pf-thm-monotonicity} and follows from showing that the mechanism $g$ described above, given blocking sets satisfying the conditions of \autoref{lem:existence-blocking-sets}, is well-defined and guaranteed to fulfill the conditions of \autoref{lem:suf-conds-g}. On an intuitive level, Condition~\ref{item:one-pos-per-agent} of \autoref{lem:existence-blocking-sets} ensures that at most one position is unblocked for each agent in addition to its default position; Conditions~\ref{item:one-agent-per-pos-def} and~\ref{item:one-agent-per-pos-non-def} of this lemma ensure to assign all agents to different positions. These facts guarantee that the mechanism is well-defined. Condition~\ref{item:g-impartial} of \autoref{lem:suf-conds-g}, guaranteeing impartiality of the obtained mechanism, follows directly from its definition, since the position of each agent is fully determined by the positions blocked by other agents and its own ranking plays no role in this matter, while Condition~\ref{item:g-monotone} of this lemma is proven by a careful analysis of the unblocked positions of $\rho_i$ when each agent $i$ switches its own message, using Condition~\ref{item:mon} of \autoref{lem:existence-blocking-sets}. Finally, to prove Condition~\ref{item:g-ifr} of \autoref{lem:suf-conds-g} we use the structure of the blocking sets given by Condition~\ref{item:ifr} of \autoref{lem:existence-blocking-sets} and construct, for each agent $j$ and position $k$, a profile for which agent $j$ ends up in position $k$ of the output ranking. When $k\not=j$, this is achieved by simply taking, for each agent other than~$j$ and~$k$, the message for which that agent does not block position~$k$. When $k=j$, we construct the profile inductively by picking, for each agent other than~$j$ and~$k$, a message that strictly increases the size of the set of positions that are blocked for agent~$j$; in the end all positions except its default position are blocked for agent~$j$, and it is assigned the default position.

We do not know whether there exists a mechanism satisfying impartiality, monotonicity, and non-imposition, but the construction of such a mechanism would require new techniques. Non-imposition in particular means that there are $n!$ distinct output rankings. This in turn requires $\Omega(n)$ messages per agent, making monotonicity much more difficult to control than for the mechanisms with two messages per agent we have constructed above.

\section{Weak Unanimity}
\label{sec:weak-unanimity}

We have shown the existence of an impartial and monotone mechanism that allows every agent to appear in every position of the output ranking. By dropping monotonicity it is possible to increase expressiveness further, and guarantee that an arbitrary ranking is produced as the output if all agents agree on that ranking.
\begin{theorem}  \label{thm:weak-unanimity}
  Let $n\geq 5$. Then there exists an $n$-ranking mechanism that satisfies impartiality and weak unanimity.
\end{theorem}

To prove this result we construct a mechanism that only takes into account the input rankings of three \textit{decisive agents}, which for simplicity we associate with indices $0$, $1$, and $2$. The position of each decisive agent in the output ranking is determined by the other two decisive agents, and the position of all other agents is determined by the decisive agents. Impartiality thus follows by construction. Weak unanimity leads to the following requirements: whenever two decisive agents agree on a ranking, the third decisive agent is ranked in the position in which it appears in that ranking; and whenever all three decisive agents agree on a ranking, all other agents are ranked in the positions in which they appear in that ranking. The difficulty, of course, is to meet these requirements while ensuring that each position is assigned to only one agent.

The requirements can be expressed more conveniently as a coloring problem. The goal will be to construct, for some $m\in\NN$, three $m\times m$ matrices with entries in $[n]$. Values on the diagonals are given, and the values in column~$p\in[m]$ of each matrix may not intersect with those in row~$p$ of the next matrix.
Formally, for $n,m\in\NN$ and vectors $\bfd^0,\bfd^1,\bfd^2 \in[n]^m$, let $\calA(n,m,\bfd)$ be the set of tuples $(\bfA^0, \bfA^1, \bfA^2)$, where $\bfA^0, \bfA^1, \bfA^2 \in [n]^{m\times m}$ are such that
\begin{align}
        A^i_{pp} & = d^i_p \text{ for every } i\in \{0,1,2\}  \text{ and } p\in [m], \label{eq:vertex-coloring-1}\\
        A^i_{pq} & \not= A^{i+_3 1}_{qr} \text{ for every }i\in \{0,1,2\} \text{ and } p,q,r\in [m].\label{eq:vertex-coloring-2}
\end{align}
An example for $n=m=5$, and diagonals $d^i_p=p+_5 i$ for $i\in \{0,1,2\}$ and $p\in [5]$,
is shown in Part~(a) of \autoref{fig:matrix-coloring-n5}.
\begin{figure}[t]
\newcommand{\Size}{.7cm}%

\def\NumOfColumns{5}%
\def\Sequence{1/A, 2/B, 3/C, 4/D, 5/E}%

\tikzset{Square/.style={
    inner sep=0pt,
    text width=\Size, 
    minimum size=\Size,
    draw=black,
    align=center
    }
}

\newcommand{\NodepAA}{0}%
\newcommand{\NodepAB}{3}%
\newcommand{\NodepAC}{4}%
\newcommand{\NodepAD}{3}%
\newcommand{\NodepAE}{4}%
\newcommand{\NodepBA}{4}%
\newcommand{\NodepBB}{1}%
\newcommand{\NodepBC}{4}%
\newcommand{\NodepBD}{1}%
\newcommand{\NodepBE}{4}%
\newcommand{\NodepCA}{3}%
\newcommand{\NodepCB}{3}%
\newcommand{\NodepCC}{2}%
\newcommand{\NodepCD}{3}%
\newcommand{\NodepCE}{2}%
\newcommand{\NodepDA}{3}%
\newcommand{\NodepDB}{3}%
\newcommand{\NodepDC}{2}%
\newcommand{\NodepDD}{3}%
\newcommand{\NodepDE}{2}%
\newcommand{\NodepEA}{0}%
\newcommand{\NodepEB}{3}%
\newcommand{\NodepEC}{4}%
\newcommand{\NodepED}{3}%
\newcommand{\NodepEE}{4}%
\begin{subfigure}[c]{\textwidth}
\hfill
\begin{tikzpicture}[draw=black, ultra thick, x=\Size,y=\Size]
\useasboundingbox (-0.5,-7.14) rectangle (6.64,0);

    \foreach \col/\colLetter in \Sequence {%
        \foreach \row/\rowLetter in \Sequence{%
            \pgfmathtruncatemacro{\value}{\col+\NumOfColumns*(\row-1)}
            \def\NodeText{\expandafter\csname Nodep\rowLetter\colLetter\endcsname}
            \node [Square,fill=color\NodeText] at ($(\col,-\row)-(0.5,0.5)$) {$\NodeText$};
        }
    }
    \foreach \col/\colSet in {1/$S^0_1$,2/$S^0_2$,3/$S^0_3$,4/$S^0_2$,5/$S^0_3$} {
        \node[opacity=0.4] at ($(\col,-\NumOfColumns-1)-(0.5,0.5)$) {\colSet}; 
    }
    \foreach \row/\rowSet in {1/$(S^2_2)^C$,2/$(S^2_1)^C$,3/$(S^2_0)^C$,4/$(S^2_0)^C$,5/$(S^2_2)^C$} {
        \node[opacity=0.4] at ($(\NumOfColumns+1,-\row)-(0.2,0.5)$) {\rowSet}; 
    }
\end{tikzpicture}
\newcommand{\NodeqAA}{1}%
\newcommand{\NodeqAB}{2}%
\newcommand{\NodeqAC}{1}%
\newcommand{\NodeqAD}{1}%
\newcommand{\NodeqAE}{2}%
\newcommand{\NodeqBA}{0}%
\newcommand{\NodeqBB}{2}%
\newcommand{\NodeqBC}{0}%
\newcommand{\NodeqBD}{4}%
\newcommand{\NodeqBE}{0}%
\newcommand{\NodeqCA}{0}%
\newcommand{\NodeqCB}{0}%
\newcommand{\NodeqCC}{3}%
\newcommand{\NodeqCD}{1}%
\newcommand{\NodeqCE}{0}%
\newcommand{\NodeqDA}{0}%
\newcommand{\NodeqDB}{0}%
\newcommand{\NodeqDC}{0}%
\newcommand{\NodeqDD}{4}%
\newcommand{\NodeqDE}{0}%
\newcommand{\NodeqEA}{0}%
\newcommand{\NodeqEB}{0}%
\newcommand{\NodeqEC}{0}%
\newcommand{\NodeqED}{1}%
\newcommand{\NodeqEE}{0}%
\hfill
\begin{tikzpicture}[draw=black, ultra thick, x=\Size,y=\Size]
\useasboundingbox (-0.5,-7.14) rectangle (6.64,0);
    \foreach \col/\colLetter in \Sequence {%
        \foreach \row/\rowLetter in \Sequence{%
            \pgfmathtruncatemacro{\value}{\col+\NumOfColumns*(\row-1)}
            \def\NodeText{\expandafter\csname Nodeq\rowLetter\colLetter\endcsname}
            \node [Square,fill=color\NodeText] at ($(\col,-\row)-(0.5,0.5)$) {$\NodeText$};
        }
    }
    \foreach \col/\colSet in {1/$S^1_0$,2/$S^1_1$,3/$S^1_0$,4/$S^1_2$,5/$S^1_1$} {
        \node[opacity=0.4] at ($(\col,-\NumOfColumns-1)-(0.5,0.5)$) {\colSet}; 
    }
    \foreach \row/\rowSet in {1/$(S^0_1)^C$,2/$(S^0_2)^C$,3/$(S^0_3)^C$,4/$(S^0_2)^C$,5/$(S^0_3)^C$} {
        \node[opacity=0.4] at ($(\NumOfColumns+1,-\row)-(0.2,0.5)$) {\rowSet}; 
    }
\end{tikzpicture}
\newcommand{\NoderAA}{2}%
\newcommand{\NoderAB}{2}%
\newcommand{\NoderAC}{4}%
\newcommand{\NoderAD}{4}%
\newcommand{\NoderAE}{2}%
\newcommand{\NoderBA}{1}%
\newcommand{\NoderBB}{3}%
\newcommand{\NoderBC}{1}%
\newcommand{\NoderBD}{1}%
\newcommand{\NoderBE}{1}%
\newcommand{\NoderCA}{2}%
\newcommand{\NoderCB}{2}%
\newcommand{\NoderCC}{4}%
\newcommand{\NoderCD}{4}%
\newcommand{\NoderCE}{2}%
\newcommand{\NoderDA}{2}%
\newcommand{\NoderDB}{0}%
\newcommand{\NoderDC}{0}%
\newcommand{\NoderDD}{0}%
\newcommand{\NoderDE}{2}%
\newcommand{\NoderEA}{1}%
\newcommand{\NoderEB}{3}%
\newcommand{\NoderEC}{1}%
\newcommand{\NoderED}{1}%
\newcommand{\NoderEE}{1}%
\hfill
\begin{tikzpicture}[draw=black, ultra thick, x=\Size,y=\Size]
\useasboundingbox (-0.5,-7.14) rectangle (6.64,0);
    \foreach \col/\colLetter in \Sequence {%
        \foreach \row/\rowLetter in \Sequence{%
            \pgfmathtruncatemacro{\value}{\col+\NumOfColumns*(\row-1)}
            \def\NodeText{\expandafter\csname Noder\rowLetter\colLetter\endcsname}
            \node [Square,fill=color\NodeText] at ($(\col,-\row)-(0.5,0.5)$) {$\NodeText$};
        }
    }
    \foreach \col/\colSet in {1/$S^2_2$,2/$S^2_1$,3/$S^2_0$,4/$S^2_0$,5/$S^2_2$} {
        \node[opacity=0.4] at ($(\col,-\NumOfColumns-1)-(0.5,0.5)$) {\colSet}; 
    }
    \foreach \row/\rowSet in {1/$(S^1_0)^C$,2/$(S^1_1)^C$,3/$(S^1_0)^C$,4/$(S^1_2)^C$,5/$(S^1_1)^C$} {
        \node[opacity=0.4] at ($(\NumOfColumns+1,-\row)-(0.2,0.5)$) {\rowSet}; 
    }
\end{tikzpicture}
\hspace*{\fill}
\caption{}
\end{subfigure}
\par\vspace*{3ex}
\begin{subfigure}[c]{\textwidth}
\hfill
\begin{tikzpicture}[scale=0.9]
\useasboundingbox (-2.5,-2.5) rectangle (2.5,2.5);
\tikzset{
    contour/.style={
    black, 
      double=lightgray,
      opacity=0.4,
      double distance=#1,
      cap=round,
      rounded corners=0.5mm,
    }
  }
\pgfdeclarelayer{bg}
\pgfsetlayers{bg,main}
    \foreach \i/\c in {0/color0,1/color1,2/color2,3/color3,4/color4}{\draw (\i*360/5:1.8cm) node[circle,fill=\c](\i){\i};}
    \begin{pgfonlayer}{bg}
    \draw[component,draw=Violet,line width=1.2cm] (0.center) node[below=10mm] {\textcolor{Violet}{$S^0_1$}} to[bend right=15] (4.center) to[bend right=15] (3.center) to[bend right=15] cycle;
    \draw[component,draw=OliveGreen](2.center) node[left =4mm] {\textcolor{OliveGreen}{$S^0_3$}} to (4.center) to cycle;
    \draw[component,draw=MidnightBlue](1.center) to  (3.center) node[left=4mm] {\textcolor{MidnightBlue}{$S^0_2$}} to  cycle;
    \draw[component,draw=Brown,line width=1.4cm] (0.center) node[above=10mm] {\textcolor{Brown}{$S^0_0$}} to[bend left=15] (1.center) to[bend left=15](2.center) to[bend left=15] cycle;
    \end{pgfonlayer}

\end{tikzpicture}
\hfill
\begin{tikzpicture}[scale=0.9]
\useasboundingbox (-2.5,-2.5) rectangle (2.5,2.5);

\tikzset{
    contour/.style={
      black, 
      double=lightgray,
      opacity=0.4,
      double distance=#1,
      cap=round,
      rounded corners=0.5mm,
    }
  }
\pgfdeclarelayer{bg}
\pgfsetlayers{bg,main}

    \foreach \i/\c in {0/color0,1/color1,2/color2,3/color3,4/color4}{\draw (\i*360/5:1.8cm) node[circle,fill=\c](\i){\i};}
    \begin{pgfonlayer}{bg}
    \draw[component,draw=Brown,line width=1.4cm](0.center) node[above=10mm] {\textcolor{Brown}{$S^1_0$}} to[bend left=15] (1.center) to[bend left=15] (3.center)  to[bend left=15] cycle;
    \draw[component=OliveGreen](2.center) to node[left=0mm] {\textcolor{OliveGreen}{$S^1_3$}} (3.center);
    \draw[component,draw=Violet,line width=1.2cm](2.center) to[bend right=15] (0.center) node[below=10mm] {\textcolor{Violet}{$S^1_1$}} to[bend right=15](4.center) to[bend right=15] cycle;
    \draw[component,draw=MidnightBlue] (1.center)-- node[left=-8mm,pos=0.5] {\textcolor{White}{$S^1_2$}} (4.center);
    \end{pgfonlayer}

\end{tikzpicture}
\hfill
\begin{tikzpicture}[scale=0.9]
\useasboundingbox (-2.5,-2.5) rectangle (2.5,2.5);
\tikzset{
    contour/.style={
      black, 
      double=lightgray,
      opacity=0.4,
      double distance=#1,
      cap=round,
      rounded corners=0.5mm,
    }
  }
\pgfdeclarelayer{bg}
\pgfsetlayers{bg,main}
    \foreach \i/\c in {0/color0,1/color1,2/color2,3/color3,4/color4}{\draw (\i*360/5:1.8cm) node[circle,fill=\c](\i){\i};}
    \begin{pgfonlayer}{bg}
    \draw[component,draw=Brown,line width=1.4cm](1.center) to[bend right=15] (0.center) node[above=10mm] {\textcolor{Brown}{$S^2_0$}} to[bend right=15](4.center) to[bend right=15] cycle;
    \draw[component,draw=OliveGreen](4.center) to node[below,pos=0.6] {\textcolor{OliveGreen}{$S^2_3$}} (3.center);
    \draw[component,draw=Violet,line width=1.2cm](0.center) to[bend left=15] (2.center) node[above=0mm] {\textcolor{Violet}{$S^2_1$}} to[bend left=15]  (3.center) to[bend left=15] cycle;
    \draw[component,draw=MidnightBlue](2.center)-- node[above,pos=0.4]{\textcolor{MidnightBlue}{$S^2_2$}} (1.center);
    \end{pgfonlayer}
\end{tikzpicture}
\hspace*{\fill}
\caption{}
\end{subfigure}
        \caption{Part (a) illustrates a triple $(\bfA^0, \bfA^1, \bfA^2) \in \calA(5,5,\bfd)$ with $d^i_p=p+_5 i$ for $i\in \{0,1,2\}$ and $p\in [5]$, and colors representing the values $\ell\in \{0,1,2,3,4\}$. Observe that the colors in the $p$th column of each matrix $\bfA^i$ do not appear in the $p$th row of the matrix $\bfA^{i+_3 1}$. Part (b) illustrates a way of coloring these matrices that is used in the proof of \autoref{lem:existence-vertex-coloring}. For each matrix $\bfA^i$, each column is assigned a set of colors given by a hyperedge of the corresponding hypergraph below the matrix, and the $p$th row is assigned the colors that do not appear in the $p$th column of the matrix $\bfA^{i-_3 1}$. Each cell is then colored with any color assigned to both its row and its column.} 
        \label{fig:matrix-coloring-n5}
    \end{figure}
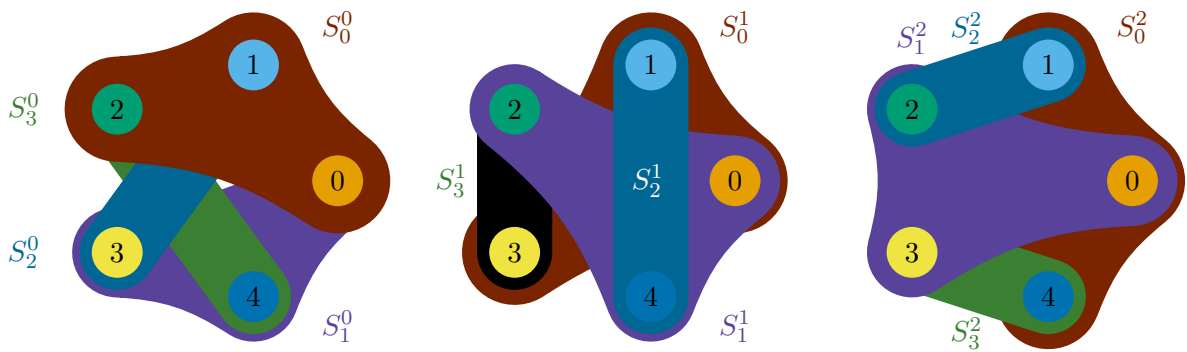

It turns out that $\calA(n,m,\bfd)$ is nonempty whenever $n\geq 5$ and the respective $p$th entries of the vectors $\bfd^0$, $\bfd^1$, and $\bfd^2$ do not coincide for any $p\in[m]$. %
\begin{lemma}  \label{lem:existence-vertex-coloring}
    Let $n\in \NN$ with $n\geq 5$, $m\in \NN$. Let $\bfd^0,\bfd^1,\bfd^2\in[n]^m$ such that $d^i_p\not=d^j_p$ for every $i,j\in \{0,1,2\}$ with $i\not=j$ and every $p\in [m]$. Then $\calA(n,m,\bfd)\not= \emptyset$.
\end{lemma}

The proof of the lemma is given in \autoref{app:pf-lem-existence-vertex-coloring}. It proceeds by defining, for each row and column of each matrix, a feasible set of values for that row or column. The non-intersection constraint~\eqref{eq:vertex-coloring-2} is achieved by defining the feasible set for the $p$th row of a matrix $\bfA^i$ as the complement of the feasible set for the $p$th column of the previous matrix, $\bfA^{i-_3 1}$. For each entry of each matrix, a value is then chosen from the intersection of the feasible set for its row and the feasible set for its column. Nonemptiness of the intersection of the feasible sets for each row and column of the same matrix turns out to be equivalent to the condition that none of the feasible sets for columns of matrix $\bfA^i$ is a subset of any of the feasible sets for columns of matrix $\bfA^{i-_3 1}$. An illustration of this condition for the example in Part~(a) of \autoref{fig:matrix-coloring-n5} is shown in Part~(b) of that figure.

The proof of \autoref{thm:weak-unanimity}, which can be found in \autoref{app:pf-thm-weak-unanimity}, defines a mechanism that places each agent $i\in\{0,1,2\}$ in a position of the output ranking given by an entry of matrix $\bfA^i$. We take $m=n!$, and associate each row of $\bfA^i$ with an input ranking of agent $i+_3 1$ and each column of $\bfA^i$ with an input ranking of agent $i+_3 2$. Rankings are associated with rows and columns in a symmetric way across agents, so that the position of agent~$i$ is given by a value on the diagonal of $\bfA^i$ when the other two decisive agents agree on a ranking. To enable weak unanimity, the value on the diagonal is set to the position of agent~$i$ in the agreed ranking. The position of a non-decisive agent $i\in[n]\setminus\{0,1,2\}$ is determined by a function $g_i\colon[n!]^3\to[n]$, such that the agent is ranked in position $g_i(p,q,r)$ when the input rankings of the three decisive agents are given by $p$, $q$, and $r$. This is done in such a way that whenever the decisive agents agree on a ranking all non-decisive agents, and indeed all agents, are placed in the same position as in that ranking. Weak unanimity thus holds. Impartiality for each decisive agent holds because its position is determined by the other two decisive agents, impartiality for non-decisive agents because they have no influence on the outcome.

The fact that the mechanism determines the output ranking from the input rankings of only three agents makes it more convenient to describe and analyze. It means, however, that the mechanism violates a property \citet{holzman2013impartial} call \emph{no dummy}, which requires every agent to be able to influence the output ranking for some profile of input rankings of the other agents. When $n\geq 6$, no dummy can easily be achieved without sacrificing any of the other properties: for any profile of input rankings, assign positions to the three decisive agents as before; let~$i_0$ be the agent assigned the highest position among non-decisive agents, and let~$i_1$ and~$i_2$ be the next two non-decisive agents according to some fixed ordering; now assign positions to non-decisive agents as before, but make the relative order of~$i_1$ and~$i_2$ the same as their relative order in the input ranking of~$i_0$.

\section{Impossibility Results}
\label{sec:impossibilities}

We conclude by showing that our positive results do not leave much room for improvement.
We have seen that for $n\geq 4$, impartiality is compatible with monotonicity and individual full rank. The requirement that $n\geq 4$ is necessary, as for all non-trivial $n<4$ individual full rank alone is incompatible with impartiality.
\begin{theorem}  \label{thm:imp-ifr-n23}
For $n\in \{2,3\}$, there does not exist an $n$-ranking mechanism that satisfies impartiality and individual full rank.
\end{theorem}

The proof of this theorem can be found in \autoref{app:pf-thm-imp-ifr-n23}. It is straightforward for $n=2$. For $n=3$, we show that an impartial mechanism that produces a particular ranking~$\pi$ for some ranking profile cannot for some other ranking profile produce a cyclic shift of~$\pi$; it follows that an impartial mechanism can produce at most two distinct rankings, which stands in sharp contrast to individual full rank.

When $n\geq 5$ impartiality is compatible with weak unanimity, which is stronger than individual full rank. Weak unanimity cannot be strengthened further to unanimity, since the latter is incompatible with impartiality for all non-trivial $n$.
\begin{theorem}  \label{thm:imp-unanimity}
For $n\in\NN$ with $n\geq 2$, there does not exist an $n$-ranking mechanism that satisfies impartiality and unanimity.
\end{theorem}

We defer the proof to \autoref{app:pf-thm-imp-unanimity} but briefly explain the underlying ideas. 
The proof proceeds in a similar way as some proofs of Arrow's impossibility, for example one by \citet{geanakoplos2005three}.
It starts from a profile in which every agent casts the ranking $(1\ \ 2 \ \ \cdots\ \ n-1\ \ 0)$, which by unanimity must also be the output ranking. Agents from $n-1$ to $1$ then change, one by one, to the identity ranking $(0\ \ 1\ \ \cdots \ \ n-2\ \ n-1)$ by moving agent $0$ from the last position of their ranking to the first. By impartiality and unanimity, when agents $k,k+1,\dots,n-1$ have changed their input ranking, the positions of agents $1,2,\dots,k$ in the output ranking remain unchanged. Thus agent~$1$ remains in the first position when all agents except agent~$0$ have changed their input ranking, and by impartiality agent~$0$ is not in the first position after all agents have changed their input ranking. This is a contradiction to unanimity.

\autoref{thm:imp-ifr-n23} leaves open the possibility of a $4$-ranking mechanism satisfying impartiality and weak unanimity, but we can show computationally that such a mechanism does not exist. This is done by writing the axioms for a set of rank profiles as linear constraints, and finding a small set of rank profiles for which these constraints are infeasible. The counterexample is large, so we do not reproduce it here.

\newpage

\appendix

\section{Deferred Proofs from \autoref{sec:preliminaries}}

\subsection{Proof of \autoref{lem:relation-axioms}}
\label{app:pf-lem-relation-axioms}

Let $n\in \NN$ with $n\geq 2$, and let first $f$ be an $n$-ranking mechanism satisfying unanimity. Let $\bfpi\in \calP_n^n$ be such that there is some $\hat{\pi}\in \calP_n$ with $\pi_i=\hat{\pi}$ for every $i\in [n]$. Denoting $\hat{\pi}=(j_0\ \ j_1\ \ \ldots\ \ j_{n-1})$, we have that for every $k, \ell\in [n]$ with $k<\ell,  \pi_i^{-1}(j_k) < \pi_i^{-1}(j_{\ell})$ for every $i\in [n]$. Therefore, the fact of $f$ satisfying unanimity implies $(f(\bfpi))^{-1}(j_k) < (f(\bfpi))^{-1}(j_{\ell})$ for every $k,\ \ell\in [n]$ with $k<\ell$. We conclude that $f(\bfpi) = (j_0\ \ j_1\ \ \ldots\ \ j_{n-1})=\hat{\pi}$, thus $f$ satisfies weak unanimity.

Let now $f$ be an $n$-ranking mechanism satisfying weak unanimity and let $j, k\in [n]$ be arbitrary. Defining $\hat{\pi}\in \calP_n$ as a ranking with $\hat{\pi}(k) = j$ and $\bfpi \in \calP_n^n$ as the ranking profile with $\pi_i = \hat{\pi}$ for every $i\in [n]$, we have from the fact of $f$ satisfying weak unanimity that $f(\bfpi) = \hat{\pi}$ so, in particular, $(f(\bfpi))(k) = j$. We conclude that $f$ satisfies individual full rank.

\section{Deferred Proofs from \autoref{sec:monotonicity}}

\subsection{Proof of \autoref{lem:suf-conds-g}}
\label{app:pf-lem-suf-conds-g}

Let $n\geq 4$ and consider $g$ and $\bfrho$ as defined in the statement of the lemma. 
In a slight abuse of notation, for $i\in [n],~\bfpi\in \calP_n^n$, and $\tilde{\pi}\in \calP_n$, we use $g(\chi_i(\tilde{\pi},\rho_i), \bfchi_{-i}(\bfpi_{-i}, \bfrho_{-i}))$ to denote 
\[
    g(\chi_0(\pi_0,\rho_0),\chi_1(\pi_1,\rho_1),\dots, \chi_{i-1}(\pi_{i-1},\rho_{i-1}), \chi_{i}(\tilde{\pi},\rho_{i}), \chi_{i+1}(\pi_{i+1},\rho_{i+1}),\dots, \chi_{n-1}(\pi_{n-1},\rho_{n-1})).
\]

To see that $f_{g,\bfrho}$ is impartial, fix $i\in [n],~\bfpi\in \calP_n^n$, and $\tilde{\pi} \in \calP_n$ arbitrarily. We have that 
\begin{align*}
    (f_{g,\bfrho}(\bfpi))^{-1}(i)
    & = (g(\chi_i(\pi_i,\rho_i), \bfchi_{-i}(\bfpi_{-i}, \bfrho_{-i})))^{-1}(i)\\
    & = (g(\chi_i(\tilde{\pi},\rho_i), \bfchi_{-i}(\bfpi_{-i}, \bfrho_{-i})))^{-1}(i)\\
    & = (f_{g,\bfrho}(\tilde{\pi}, \bfpi_{-i}))^{-1}(i),
\end{align*}
where the first and the last equality follow from the definition of $f_{g,\bfrho}$ and the second equality follows from Condition~\ref{item:g-impartial} in the statement of the lemma.

To see that $f_{g,\bfrho}$ satisfies individual full rank, consider $j, k\in [n]$. We know from Condition~\ref{item:g-ifr} in the statement of the lemma that there exists $\bfb\in \{0,1\}^n$ such that $(g(\bfb))(k) = j$. Define, for each $i\in [n]$, $\pi_i\in \calP_n$ such that $\pi_i^{-1}(\rho_i) < \pi_i^{-1}(i)$ if $b_i = 1$ and $\pi_i^{-1}(\rho_i) > \pi_i^{-1}(i)$ otherwise.
Then, for every $i\in [n]$ we have that $\chi_i(\pi_i,\rho_i)=b_i$ and thus $f_{g,\bfrho}(\bfpi) = g(\bfchi(\bfpi,\bfrho)) = g(\bfb)$. In particular, $(f_{g,\bfrho}(\bfpi))(k) = (g(\bfb))(k) = j$.

Finally, to see that $f_{g,\bfrho}$ is monotone, fix $\bfpi\in\calP_n^n$, $i,j_0\in [n]$, and $\tilde{\pi}\in\calP_n$ such that for all $j_1\in[n]$ and $j_2\in[n]\setminus\{j_0\}$,
    \vspace*{-1ex}
    \begin{align}\label{eq:mon}
        \tilde{\pi}^{-1}(j_1)< \tilde{\pi}^{-1}(j_2) \quad \text{whenever}\quad \pi_i^{-1}(j_1)<\pi_i^{-1}(j_2).  \vspace*{-1ex}
    \end{align}
If $\rho_i \not= j_0$, then~\eqref{eq:mon} for both $(j_1,j_2)=(i,\rho_i)$ and $(j_1,j_2)=(\rho_i,i)$ yields $\chi_i(\tilde{\pi},\rho_i)=\chi_i(\pi_i,\rho_i)$, thus
\begin{align*}
    (f_{g,\bfrho}(\tilde{\pi}, \bfpi_{-i}))^{-1}(j_0) & = (g(\chi_i(\tilde{\pi},\rho_i),\bfchi_{-i}(\bfpi_{-i},\bfrho_{-i})))^{-1}(j_0)\\
    & = (g(\chi_i(\pi_i,\rho_i),\bfchi_{-i}(\bfpi_{-i},\bfrho_{-i})))^{-1}(j_0)\\
    & = (f_{g,\bfrho}(\bfpi))^{-1}(j_0),
\end{align*}
and monotonicity holds trivially in this case.
If $\rho_i = j_0$, we know that, if $\pi_i^{-1}(\rho_i) < \pi_i^{-1}(i)$, then $\tilde{\pi}^{-1}(\rho_i) < \tilde{\pi}^{-1}(i)$, thus $\chi_i(\tilde{\pi},\rho_i) \geq \chi_i(\pi_i,\rho_i)$. This implies that
\begin{align*}
    (f_{g,\bfrho}(\tilde{\pi}, \bfpi_{-i}))^{-1}(j_0) & = (g(\chi_i(\tilde{\pi},\rho_i),\bfchi_{-i}(\bfpi_{-i},\bfrho_{-i})))^{-1}(j_0)\\
    & \leq (g(\chi_i(\pi_i,\rho_i),\bfchi_{-i}(\bfpi_{-i},\bfrho_{-i})))^{-1}(j_0)\\
    & = (f_{g,\bfrho}(\bfpi))^{-1}(j_0),
\end{align*}
where the inequality follows from Condition~\ref{item:g-monotone} in the statement of the lemma. This concludes the proof of monotonicity and the proof of the lemma.

\subsection{Proof of \autoref{lem:existence-g-n4}}
\label{app:pf-lem-existence-g-n4}

Let $\bfrho=(1,0,1,0)\in R_4$, and let $g\colon\{0,1\}^4\to\calP_4$ be defined as shown in \autoref{fig:g-01-4}.
\begin{figure}[tb]
    \centering
    \begin{tikzpicture}
        \matrix [matrix of math nodes,left delimiter=.,right delimiter=.](A){ 
        (2\ \ 3\ \ 1\ \ 0) & (2\ \ 3\ \ 0\ \ 1) & \qquad & (0\ \ 3\ \ 1\ \ 2) & (0\ \ 3\ \ 2\ \ 1) \\[0.3cm]
        (2\ \ 1\ \ 3\ \ 0) & (2\ \ 0\ \ 3\ \ 1) & \qquad & (3\ \ 1\ \ 0\ \ 2) & (3\ \ 0\ \ 2\ \ 1) \\[0.8cm]        
        (3\ \ 2\ \ 1\ \ 0) & (3\ \ 1\ \ 0\ \ 2) & \qquad & (0\ \ 2\ \ 1\ \ 3) & (0\ \ 1\ \ 2\ \ 3) \\[0.3cm]
        (1\ \ 2\ \ 3\ \ 0) & (1\ \ 0\ \ 3\ \ 2) & \qquad & (1\ \ 2\ \ 0\ \ 3) & (1\ \ 0\ \ 2\ \ 3) \\
        };
        \draw[<->, thick] (-3.4,2.3) node[left] {$0$} -- (-1.8,2.3) node[midway,above] {$3$} node[right] {$1$};
        \draw[<->, thick] (1.8,2.3) node[left] {$0$} -- (3.4,2.3) node[midway,above] {$3$} node[right] {$1$};
        \draw[<->, thick] (-5,1.6) node[above] {$0$} -- (-5,0.75) node[midway,left] {$2$} node[below] {$1$};
        \draw[<->, thick] (-5,-0.75) node[above] {$0$} -- (-5,-1.6) node[midway,left] {$2$} node[below] {$1$};
        \draw[<->, thick] (-2.35,-2.3) node[left] {$0$} -- (2.35,-2.3) node[midway,below] {$1$} node[right] {$1$};
        \draw[<->, thick] (5,1.175) node[above] {$0$} -- (5,-1.175) node[midway,right] {$0$} node[below] {$1$};
    \end{tikzpicture}
    \caption{Matrix of values of the function $g\colon\{0,1\}^4\to\calP_4$ used in the proof of \autoref{lem:existence-g-n4}. Each entry of the matrix is equal to $g(\bfb)$ for some $\bfb\in\{0,1\}^n$, where $b_i$ is indicated by an arrow with label~$i$ at the center and values~$0$ and~$1$ at either end.
Thus, for example, $g(0,1,0,0)=(0~~3~~1~~2)$ and $g(1,0,0,1)=(3~~1~~0~~2)$.}
    \label{fig:g-01-4}
\end{figure}
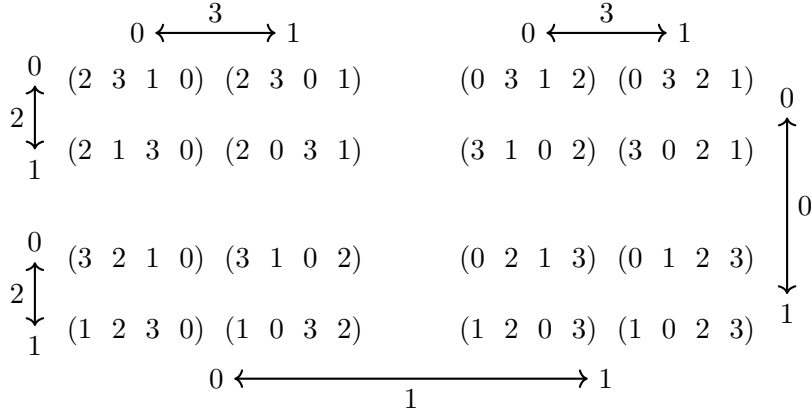
The fact that $g$ satisfies Conditions~\ref{item:g-impartial} and~\ref{item:g-monotone} in the statement of \autoref{lem:suf-conds-g} follows from the fact that, for every $i\in \{0,1,2,3\}$ and $\bfb\in \{0,1\}^4$, we have $(g(0,\bfb_{-i}))^{-1}(i) = (g(1,\bfb_{-i}))^{-1}(i)$ and $(g(1,\bfb_{-i}))^{-1}(\rho_i) \leq (g(0,\bfb_{-i}))^{-1}(\rho_i)$. These facts are easy to check from the matrix; \autoref{fig:table-imp-mon-n4} shows them explicitly for $i=0$.
\begin{table}[t]
\centering
\begin{tabular}{ccccccc}
\toprule
$b_1$ & $b_2$ & $b_3$ & $(g(0,\bfb_{-0}))^{-1}(0)$ & $(g(1,\bfb_{-0}))^{-1}(0)$ & $(g(0,\bfb_{-0}))^{-1}(1)$ & $(g(1,\bfb_{-0}))^{-1}(1)$ \\
\midrule
$0$   & $0$   & $0$   & $3$                         & $3$                         & $2$                         & $2$                         \\
$0$   & $0$   & $1$   & $2$                         & $2$                         & $3$                         & $1$                         \\
$0$   & $1$   & $0$   & $3$                         & $3$                         & $1$                         & $0$                         \\
$0$   & $1$   & $1$   & $1$                         & $1$                         & $3$                         & $0$                         \\
$1$   & $0$   & $0$   & $0$                         & $0$                         & $2$                         & $2$                         \\
$1$   & $0$   & $1$   & $0$                         & $0$                         & $3$                         & $1$                         \\
$1$   & $1$   & $0$   & $2$                         & $2$                         & $1$                         & $0$                         \\
$1$   & $1$   & $1$   & $1$                         & $1$                         & $3$                         & $0$                         \\
\bottomrule
\end{tabular}
\caption{Position of agents $0$ and $\rho_0=1$ given by $g(\bfb)$ for each profile $\bfb\in \{0,1\}^4$, where $g$ is defined as in the proof of \autoref{lem:existence-g-n4}. The fourth and fifth columns show that $g$ satisfies Condition~\ref{item:g-impartial} in the statement of \autoref{lem:suf-conds-g} for $i=0$ since the position of $0$ does not change for a given vector $\bfb_{-0}$; the last two columns show that $g$ satisfies Condition~\ref{item:g-monotone} in the statement of \autoref{lem:suf-conds-g} for $i=0$ since, for a given vector $\bfb_{-0}$, $\rho_0$ does not move below in the output ranking when switching from $(0,\bfb_{-0})$ to $(1,\bfb_{-0})$.}
\label{fig:table-imp-mon-n4}
\end{table}

Finally, Condition~\ref{item:g-ifr} follows from the fact for each $j,k\in \{0,1,2,3\}$, there exists $\bfb\in \{0,1\}^4$ such that $(g(\bfb))(k)=j$. For example, for $j=0$ we have $(g(0,1,0,0))(0)=0$, $(g(0,0,1,1))(1)=0$, $(g(0,0,0,1))(2)=0$, and $(g(0,0,0,0))(3)=0$. This is easily checked in an analogous way for the other agents, which concludes the proof of the lemma.

\subsection{Proof of \autoref{lem:existence-blocking-sets}}
\label{app:pf-lem-existence-blocking-sets}

In order to prove \autoref{lem:existence-blocking-sets}, we start by showing the existence of vectors $\bfrho\in R_n$ and colored multigraphs satisfying certain joint properties, and then conclude the lemma by appropriately defining the blocking sets given such vector and multigraph. To do so, we introduce the family of undirected multigraphs 
\[
    \calG_n = \left\{G=(N,E)\colon N=[n],\ E=\bigcup_{i\in N}E_i,\ E_i\subseteq \left\{\{j,k\}\in 2^{[n]\setminus \{i\}}\colon j\not= k\right\} \text{ for each } i\in N\right\},
\]
where each vertex $i\in N=[n]$ is associated to a set of undirected edges $E_i$ not including $i$; we say that such edges have \textit{color $i$}. For $G=(N,E)\in \calG_n$ and $i,j\in N$, we let $N_i(j,G)=\{k\in N\colon \{j,k\}\in E_i\}$ denote the neighbors of $j$ in $G$ with edges of color $i$. We obtain the following lemma.

\begin{lemma}
\label{lem:edge-coloring}
    Let $n\in \NN$ with $n\geq 5$. Then, there exists $\bfrho\in R_n$ and $G=(N,E)\in \calG_n$ such that 
    \begin{enumerate}[label=(\roman*)]
        \item for every $i\in N$, $N_i(\rho_i,G)=\{j\in N\setminus \{i\}\colon j<\rho_i\}$;\label{item:multigraph-monotonicity}
        \item for every $k\in N$ and $j,\ell \in N\setminus \{k\}$ with $j\not=\ell$, there exists $i\in N\setminus \{j,k,\ell\}$ such that $|N_i(k,G)\cap \{j,\ell\}| = 1$.\label{item:multigraph-paths}
    \end{enumerate}
\end{lemma}

\begin{proof}
    Let $n\in \NN$ with $n\geq 11$, $N=[n]$, and $\bfrho\in R_n$ defined as $\rho_i = i+_n 1$ for every $i\in [n]$. We show the existence of $G=(N,E)\in \calG_n$ satisfying Conditions~\ref{item:multigraph-monotonicity}-\ref{item:multigraph-paths} in the statement of the lemma via a probabilistic argument. More specifically, we will randomly define the additional edges of a multigraph satisfying Condition~\ref{item:multigraph-monotonicity} and show that with strictly positive probability the multigraph satisfies Condition~\ref{item:multigraph-paths} as well.

    We construct a multigraph with an algorithm that first fixes, for each $i$, the edges of color $i$ incident to the vertex $\rho_i$ as stated in Condition~\ref{item:multigraph-monotonicity}, and then draws every other edge with probability $1/2$, independently of any other event. This algorithm, which we call $\mathsf{RandomMultigraph}$ and takes $n$ and $\bfrho$ as input, is formally described as \autoref{alg:random-edges}.
    \begin{algorithm}[t]
    \SetAlgoNoLine
    \KwIn{$n\in \NN$ with $n\geq 11$, $\bfrho\in R_n$}
    \KwOut{multigraph $(N,E)\in \calG_n$}
    Let $N=[n]$\;
    \For{$i\in N$}{
        $E_i \xleftarrow{} \{\{\rho_i,j\}\colon j\in N\setminus \{i\},\ j<\rho_i\}$\;
        \For{$e \in 2^{N\setminus \{i,\rho_i\}}$ with $|e|=2$}{
            sample $\xi(i,e)\sim \text{Bernoulli}(1/2)$ independently of other samples\;
            \If{$\xi(i,e)=1$}{
                $E_i \xleftarrow{} E_i \cup \{e\}$
            }
        }
    }
    {\bf return} $(N,E)$
    \caption{$\mathsf{RandomMultigraph}(n,\bfrho)$}
    \label{alg:random-edges}
    \end{algorithm}
    We now let $G=(N,E)=\mathsf{RandomMultigraph}(n,\bfrho)$ and claim that $\bfrho$ and $G$ satisfy Conditions~\ref{item:multigraph-monotonicity}-\ref{item:multigraph-paths} in the statement of the lemma with strictly positive probability.
    
    The fact that $G\in \calG_n$ follows by construction, since for every $i\in N$ and $e\in E_i$ we have that $e\cap \{i\} =\emptyset$. Condition~\ref{item:multigraph-monotonicity} is satisfied by construction as well, since for every $i\in N$ the algorithm sets $\{(\rho_i,j)\colon j\in N\setminus \{i\},\ j<\rho_i\}$ as a subset of $E_i$ and no other edges containing $\rho_i$ are added to $E_i$.
    To see that Condition~\ref{item:multigraph-paths} is satisfied with positive probability, for each $k\in N$ and $j,\ell\in N\setminus \{k\}$ with $j<\ell$ we define the event
    \[
        A_{jk\ell} = \bigcap_{i\in N\setminus \{j,k,\ell\}} (|E_i \cap \{\{j,k\}, \{k,\ell\}\}| \not= 1),
    \]
    corresponding to the fact that Condition~\ref{item:multigraph-paths} fails for these values of $j$, $k$, and $\ell$. We observe that 
    for every $k\in N$ and every $j,\ell \in N\setminus \{k\}$ with $j< \ell$, we have that
    \begin{align}
        \PP[A_{jk\ell}] &  \leq \PP\left[ \bigcap_{i\in N\setminus \{j,k,\ell, k-_n 1\}} (|E_i \cap \{\{j,k\}, \{k,\ell\}\}| \not= 1) \right]\nonumber \\
        & = \prod_{i\in N\setminus \{j,k,\ell, k-_n 1\}} \PP\left[ |E_i \cap \{\{j,k\}, \{k,\ell\}\}| \not= 1 \right]\nonumber\\
        & = \frac{1}{2^{n-4}}.\label{eq:prob-intersection}
    \end{align}
    Indeed, the inequality holds since the event whose probability is computed on the left-hand side is a subset of that whose probability is computed on the right-hand side, and the first equality follows from the fact that $E_i$ is independently sampled for each $i\in N$. In order to show the last equality, we fix $k\in N$ and $j,\ell \in N\setminus \{k\}$ with $j< \ell$ arbitrarily, and first note that for a fixed $i\in N\setminus \{j,k,\ell, j-_n 1,k-_n 1, \ell-_n 1\}$ each of the sets $\{j,k\}$ and $\{k,\ell\}$ belong to $E_i$ if and only if the corresponding sample of $\xi$ is $1$, and $|E_i \cap \{\{j,k\}, \{k,\ell\}\}| \not= 1$ if and only if either $\{j,k\}\in E_i$ and $\{k,\ell\}\in E_i$, or $\{j,k\}\not\in E_i$ and $\{k,\ell\}\not\in E_i$. Therefore, for $i\in N\setminus \{j,k,\ell, j-_n 1,k-_n 1, \ell-_n 1\}$, the probability $\PP\left[ |E_i \cap \{\{j,k\}, \{k,\ell\}\}| \not= 1 \right]$ is simply
    \begin{align*}
        & \PP[ (\xi(i, \{j,k\}) = 1 \text{ and } \xi(i,\{k,\ell\}) = 1) \text{ or } (\xi(i,\{j,k\}) = 0 \text{ and } \xi(i, \{k,\ell\}) = 0)] \\
        &  = \PP[\xi(i, \{j,k\}) = 1] \PP[\xi(i, \{k,\ell\}) = 1] + \PP[\xi(i, \{j,k\}) = 0] \PP[ \xi(i, \{k,\ell\}) = 0] \\
        & = \frac{1}{2}\cdot \frac{1}{2} + \frac{1}{2}\cdot \frac{1}{2} = \frac{1}{2}.
    \end{align*}
    On the other hand, when $i=j-_n 1$ we have that $\{j,k\}\in E_{j-_n 1}$ if and only if $k<j$. Therefore, if $k<j$ we have that $|E_{j-_n 1} \cap \{\{j,k\}, \{k,\ell\}\}| \not= 1$ if and only if $\{k,\ell\} \in E_{j-_n 1}$ and thus
    \[
        \PP\left[ |E_{j-_n 1} \cap \{\{j,k\}, \{k,\ell\}\}| \not= 1 \right] = \PP[\xi(j-_n 1, \{k,\ell\})=1] = \frac{1}{2},
    \]
    whereas if $k>j$ we have that $|E_{j-_n 1} \cap \{\{j,k\}, \{k,\ell\}\}| \not= 1$ if and only if $\{k,\ell\}\not\in E_{j-_n 1}$, so 
    \[
        \PP\left[ |E_{j-_n 1} \cap \{\{j,k\}, \{k,\ell\}\}| \not= 1 \right] = \PP[\xi(j-_n 1, \{k,\ell\})=0] = \frac{1}{2}.
    \]
    The argument for $i=\ell-_n 1$ is completely analogous and~\eqref{eq:prob-intersection} follows.
    
    We now fix an arbitrary vertex $k\in[n]$ and vertices $j,\ell\in [n]\setminus [k]$ with $j<\ell$, and observe that the event $A_{jk\ell}$ only depends on the realization of $\xi(i,\{j,k\})$ and $\xi(i,\{k,\ell\})$ for every $i\in [n]\setminus \{j,k,\ell\}$. Therefore, for another arbitrary vertex $k'\in [n]$ and vertices $j',\ell'\in [n]\setminus \{k'\}$ with $j'<\ell'$ and such that $(j,k,\ell)\not=(j',k',\ell')$, the events $A_{jk\ell}$ and $A_{j'k'\ell'}$ are independent unless
    \[
        \{\{j,k\},\{k,\ell\}\} \cap \{\{j',k'\}, \{k',\ell'\}\} \not= \emptyset.
    \]
    A tuple $(j',k',\ell')\not= (j,k,\ell)$ with distinct elements and $j'<\ell'$ fulfills this inequality if either (a) $k'=k$, $j\in \{j',\ell'\}$, and $\ell\not\in \{j',\ell'\}$; (b) $k'=k$, $\ell \in \{j',\ell'\}$, and $j\not\in \{j',\ell'\}$; (c) $k'=j$ and $k\in \{j',\ell'\}$; or (d) $k'=\ell$ and $k\in \{j',\ell'\}$. 
    (a) and (b) are fulfilled for $n-3$ such tuples each; (c) and (d) are fulfilled for $n-2$ tuples each. We conclude that the number of tuples $(j',k',\ell')\not= (j,k,\ell)$ with distinct elements and $j'<\ell'$ such that the events $A_{j'k'\ell'}$ and $A_{jk\ell}$ are potentially correlated is $4n-10$.

    Condition~\ref{item:multigraph-paths} is satisfied with strictly positive probability if
    \[
        \PP\left[ \bigcap_{k\in [n]} \bigcap_{j,\ell \in [n]\setminus \{k\}\colon j<\ell} (\neg A_{jk\ell}) \right] > 0.
    \]
    By the Lov\'asz local lemma~\citep{erdos1975problems,spencer1977asymptotic}, along with~\eqref{eq:prob-intersection} and the aforementioned bound on the number of correlated events, this holds as long as 
    \begin{equation}
        h(n)\defas \frac{1}{2^{n-4}}e(4n-9) \leq 1.\label{eq:lll-condition}
    \end{equation}
    Observe that $h(11)=35e/128 \approx 0.743$ and that
    \[
        h'(n) = 16e\frac{4-(4n-9)\ln 2}{2^n},
    \]
    which is negative if $n\geq 1/\ln 2 + 9/4 \approx 3.69$. We conclude that~\eqref{eq:lll-condition}, and thus the 
    lemma, hold for all $n\geq 11$.

Multigraphs satisfying Conditions~\ref{item:multigraph-monotonicity}-\ref{item:multigraph-paths} for $n\in\NN$ with $5\leq n\leq 10$ are given in \autoref{fig:multigraphs-n-5-to-7} and \autoref{fig:multigraphs-n-8-to-10}, which completes the proof.
\end{proof}
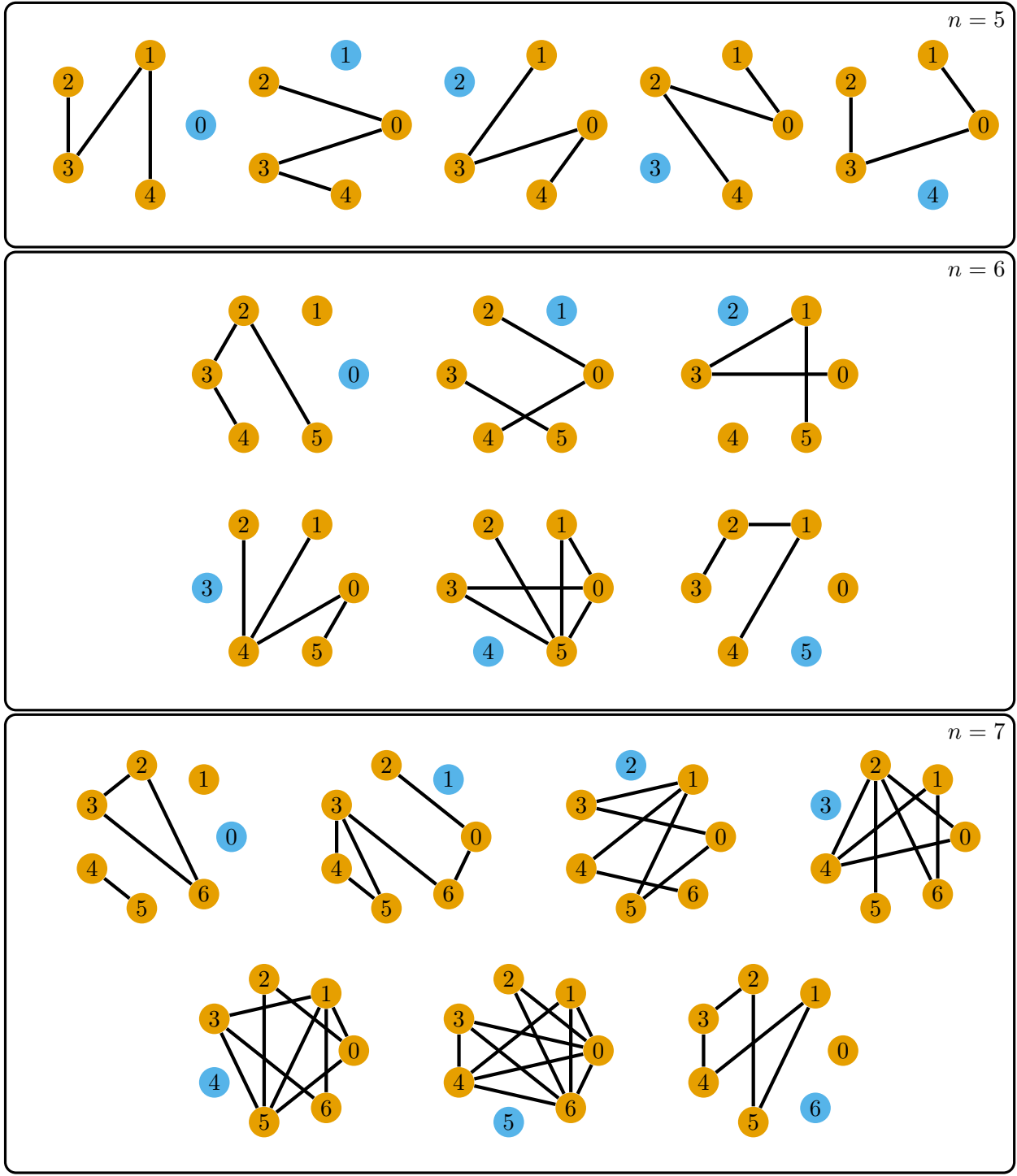
\begin{figure}[tbp]
\centering
\begin{tikzpicture}
\draw[very thick,rounded corners=5pt] (-2,-2) rectangle (14.5,2) node[below left] {$n=5$};
\begin{scope}
\foreach \i/\f in {0/color1,1/color0,2/color0,3/color0,4/color0}{\draw (\i*360/5:1.2cm) node[circle,inner sep=0, minimum size=0.5cm,fill=\f](\i){\i};}
\draw[-,draw,ultra thick] (1) edge (3);
\draw[-,draw,ultra thick] (1) edge (4);
\draw[-,draw,ultra thick] (2) edge (3);
\end{scope}
\begin{scope}[xshift=3.2cm]
\foreach \i/\f in {0/color0,1/color1,2/color0,3/color0,4/color0}{\draw (\i*360/5:1.2cm) node[circle,inner sep=0, minimum size=0.5cm,fill=\f](\i){\i};}
\draw[-,draw,ultra thick] (0) edge (2);
\draw[-,draw,ultra thick] (0) edge (3);
\draw[-,draw,ultra thick] (3) edge (4);
\end{scope}
\begin{scope}[xshift=6.4cm]
\foreach \i/\f in {0/color0,1/color0,2/color1,3/color0,4/color0}{\draw (\i*360/5:1.2cm) node[circle,inner sep=0, minimum size=0.5cm,fill=\f](\i){\i};}
\draw[-,draw,ultra thick] (0) edge (3);
\draw[-,draw,ultra thick] (0) edge (4);
\draw[-,draw,ultra thick] (1) edge (3);
\end{scope}
\begin{scope}[xshift=9.6cm]
\foreach \i/\f in {0/color0,1/color0,2/color0,3/color1,4/color0}{\draw (\i*360/5:1.2cm) node[circle,inner sep=0, minimum size=0.5cm,fill=\f](\i){\i};}
\draw[-,draw,ultra thick] (0) edge (1);
\draw[-,draw,ultra thick] (0) edge (2);
\draw[-,draw,ultra thick] (2) edge (4);
\end{scope}
\begin{scope}[xshift=12.8cm]
\foreach \i/\f in {0/color0,1/color0,2/color0,3/color0,4/color1}{\draw (\i*360/5:1.2cm) node[circle,inner sep=0, minimum size=0.5cm,fill=\f](\i){\i};}
\draw[-,draw,ultra thick] (0) edge (1);
\draw[-,draw,ultra thick] (0) edge (3);
\draw[-,draw,ultra thick] (2) edge (3);
\end{scope}
\end{tikzpicture}
\begin{tikzpicture}
\draw[very thick,rounded corners=5pt] (-2,-5.5) rectangle (14.5,2) node[below left] {$n=6$};
\begin{scope}[xshift=2.5cm]
\foreach \i/\f in {0/color1,1/color0,2/color0,3/color0,4/color0,5/color0}{\draw (\i*360/6:1.2cm) node[circle,inner sep=0, minimum size=0.5cm,fill=\f](\i){\i};}
\draw[-,draw,ultra thick] (2) edge (3);
\draw[-,draw,ultra thick] (2) edge (5);
\draw[-,draw,ultra thick] (3) edge (4);
\end{scope}
\begin{scope}[xshift=6.5cm]
\foreach \i/\f in {0/color0,1/color1,2/color0,3/color0,4/color0,5/color0}{\draw (\i*360/6:1.2cm) node[circle,inner sep=0, minimum size=0.5cm,fill=\f](\i){\i};}
\draw[-,draw,ultra thick] (0) edge (2);
\draw[-,draw,ultra thick] (0) edge (4);
\draw[-,draw,ultra thick] (3) edge (5);
\end{scope}
\begin{scope}[xshift=10.5cm]
\foreach \i/\f in {0/color0,1/color0,2/color1,3/color0,4/color0,5/color0}{\draw (\i*360/6:1.2cm) node[circle,inner sep=0, minimum size=0.5cm,fill=\f](\i){\i};}
\draw[-,draw,ultra thick] (0) edge (3);
\draw[-,draw,ultra thick] (1) edge (3);
\draw[-,draw,ultra thick] (1) edge (5);
\end{scope}
\begin{scope}[xshift=2.5cm,yshift=-3.5cm]
\foreach \i/\f in {0/color0,1/color0,2/color0,3/color1,4/color0,5/color0}{\draw (\i*360/6:1.2cm) node[circle,inner sep=0, minimum size=0.5cm,fill=\f](\i){\i};}
\draw[-,draw,ultra thick] (0) edge (4);
\draw[-,draw,ultra thick] (0) edge (5);
\draw[-,draw,ultra thick] (1) edge (4);
\draw[-,draw,ultra thick] (2) edge (4);
\end{scope}
\begin{scope}[yshift=-3.5cm,xshift=6.5cm]
\foreach \i/\f in {0/color0,1/color0,2/color0,3/color0,4/color1,5/color0}{\draw (\i*360/6:1.2cm) node[circle,inner sep=0, minimum size=0.5cm,fill=\f](\i){\i};}
\draw[-,draw,ultra thick] (0) edge (1);
\draw[-,draw,ultra thick] (0) edge (3);
\draw[-,draw,ultra thick] (0) edge (5);
\draw[-,draw,ultra thick] (1) edge (5);
\draw[-,draw,ultra thick] (2) edge (5);
\draw[-,draw,ultra thick] (3) edge (5);
\end{scope}
\begin{scope}[yshift=-3.5cm,xshift=10.5cm]
\foreach \i/\f in {0/color0,1/color0,2/color0,3/color0,4/color0,5/color1}{\draw (\i*360/6:1.2cm) node[circle,inner sep=0, minimum size=0.5cm,fill=\f](\i){\i};}
\draw[-,draw,ultra thick] (1) edge (2);
\draw[-,draw,ultra thick] (1) edge (4);
\draw[-,draw,ultra thick] (2) edge (3);
\end{scope}
\end{tikzpicture}
\begin{tikzpicture}
\draw[very thick,rounded corners=5pt] (-2,-5.5) rectangle (14.5,2) node[below left] {$n=7$};
\begin{scope}[xshift=0.5cm]
\foreach \i/\f in {0/color1,1/color0,2/color0,3/color0,4/color0,5/color0,6/color0}{\draw (\i*360/7:1.2cm) node[circle,inner sep=0, minimum size=0.5cm,fill=\f](\i){\i};}
\draw[-,draw,ultra thick] (2) edge (3);
\draw[-,draw,ultra thick] (2) edge (6);
\draw[-,draw,ultra thick] (3) edge (6);
\draw[-,draw,ultra thick] (4) edge (5);
\end{scope}
\begin{scope}[xshift=4.5cm]
\foreach \i/\f in {0/color0,1/color1,2/color0,3/color0,4/color0,5/color0,6/color0}{\draw (\i*360/7:1.2cm) node[circle,inner sep=0, minimum size=0.5cm,fill=\f](\i){\i};}
\draw[-,draw,ultra thick] (0) edge (2);
\draw[-,draw,ultra thick] (0) edge (6);
\draw[-,draw,ultra thick] (3) edge (4);
\draw[-,draw,ultra thick] (3) edge (5);
\draw[-,draw,ultra thick] (3) edge (6);
\draw[-,draw,ultra thick] (4) edge (5);
\end{scope}
\begin{scope}[xshift=8.5cm]
\foreach \i/\f in {0/color0,1/color0,2/color1,3/color0,4/color0,5/color0,6/color0}{\draw (\i*360/7:1.2cm) node[circle,inner sep=0, minimum size=0.5cm,fill=\f](\i){\i};}
\draw[-,draw,ultra thick] (0) edge (3);
\draw[-,draw,ultra thick] (0) edge (5);
\draw[-,draw,ultra thick] (1) edge (3);
\draw[-,draw,ultra thick] (1) edge (4);
\draw[-,draw,ultra thick] (1) edge (5);
\draw[-,draw,ultra thick] (4) edge (6);
\end{scope}
\begin{scope}[xshift=12.5cm]
\foreach \i/\f in {0/color0,1/color0,2/color0,3/color1,4/color0,5/color0,6/color0}{\draw (\i*360/7:1.2cm) node[circle,inner sep=0, minimum size=0.5cm,fill=\f](\i){\i};}
\draw[-,draw,ultra thick] (0) edge (2);
\draw[-,draw,ultra thick] (0) edge (4);
\draw[-,draw,ultra thick] (1) edge (4);
\draw[-,draw,ultra thick] (1) edge (6);
\draw[-,draw,ultra thick] (2) edge (4);
\draw[-,draw,ultra thick] (2) edge (5);
\draw[-,draw,ultra thick] (2) edge (6);
\end{scope}
\begin{scope}[xshift=2.5cm,yshift=-3.5cm]
\foreach \i/\f in {0/color0,1/color0,2/color0,3/color0,4/color1,5/color0,6/color0}{\draw (\i*360/7:1.2cm) node[circle,inner sep=0, minimum size=0.5cm,fill=\f](\i){\i};}
\draw[-,draw,ultra thick] (0) edge (1);
\draw[-,draw,ultra thick] (0) edge (2);
\draw[-,draw,ultra thick] (0) edge (5);
\draw[-,draw,ultra thick] (1) edge (3);
\draw[-,draw,ultra thick] (1) edge (5);
\draw[-,draw,ultra thick] (1) edge (6);
\draw[-,draw,ultra thick] (2) edge (5);
\draw[-,draw,ultra thick] (3) edge (5);
\draw[-,draw,ultra thick] (3) edge (6);
\end{scope}
\begin{scope}[xshift=6.5cm,yshift=-3.5cm]
\foreach \i/\f in {0/color0,1/color0,2/color0,3/color0,4/color0,5/color1,6/color0}{\draw (\i*360/7:1.2cm) node[circle,inner sep=0, minimum size=0.5cm,fill=\f](\i){\i};}
\draw[-,draw,ultra thick] (0) edge (1);
\draw[-,draw,ultra thick] (0) edge (2);
\draw[-,draw,ultra thick] (0) edge (3);
\draw[-,draw,ultra thick] (0) edge (4);
\draw[-,draw,ultra thick] (0) edge (6);
\draw[-,draw,ultra thick] (1) edge (4);
\draw[-,draw,ultra thick] (1) edge (6);
\draw[-,draw,ultra thick] (2) edge (6);
\draw[-,draw,ultra thick] (3) edge (4);
\draw[-,draw,ultra thick] (3) edge (6);
\draw[-,draw,ultra thick] (4) edge (6);
\end{scope}
\begin{scope}[xshift=10.5cm,yshift=-3.5cm]
\foreach \i/\f in {0/color0,1/color0,2/color0,3/color0,4/color0,5/color0,6/color1}{\draw (\i*360/7:1.2cm) node[circle,inner sep=0, minimum size=0.5cm,fill=\f](\i){\i};}
\draw[-,draw,ultra thick] (1) edge (4);
\draw[-,draw,ultra thick] (1) edge (5);
\draw[-,draw,ultra thick] (2) edge (3);
\draw[-,draw,ultra thick] (2) edge (5);
\draw[-,draw,ultra thick] (3) edge (4);
\end{scope}
\end{tikzpicture}
\caption{Multigraphs satisfying the conditions stated in \autoref{lem:edge-coloring} for $n\in \NN$ with $n\in \{5,6,7\}$. For each $n$ in this range, the multigraph $G=(N,E)\in \calG_n$ is represented through $n$ graphs, each of them containing the edges $E_i$ for each $i\in[n]$, where $i$ is colored in light blue. Condition~\ref{item:multigraph-monotonicity} follows from checking that the neighbors of $\rho_i$ in the $i$th copy of each graph are exactly the vertices $j\in [n]\setminus \{i\}$ with $j<\rho_i$. Condition~\ref{item:multigraph-paths} is guaranteed from the fact that for every $n$ and every triple $(j,k,\ell)$ of distinct vertices in $[n]$, there is a copy for which exactly one edge of the path $(j,k,\ell)$ is drawn. For $n=5,~ \bfrho=(3,2,3,1,1)$; for $n\in \{6,7\},~ \rho_i=i+_n 1$ for each $i\in [n]$.}
\label{fig:multigraphs-n-5-to-7}
\end{figure}
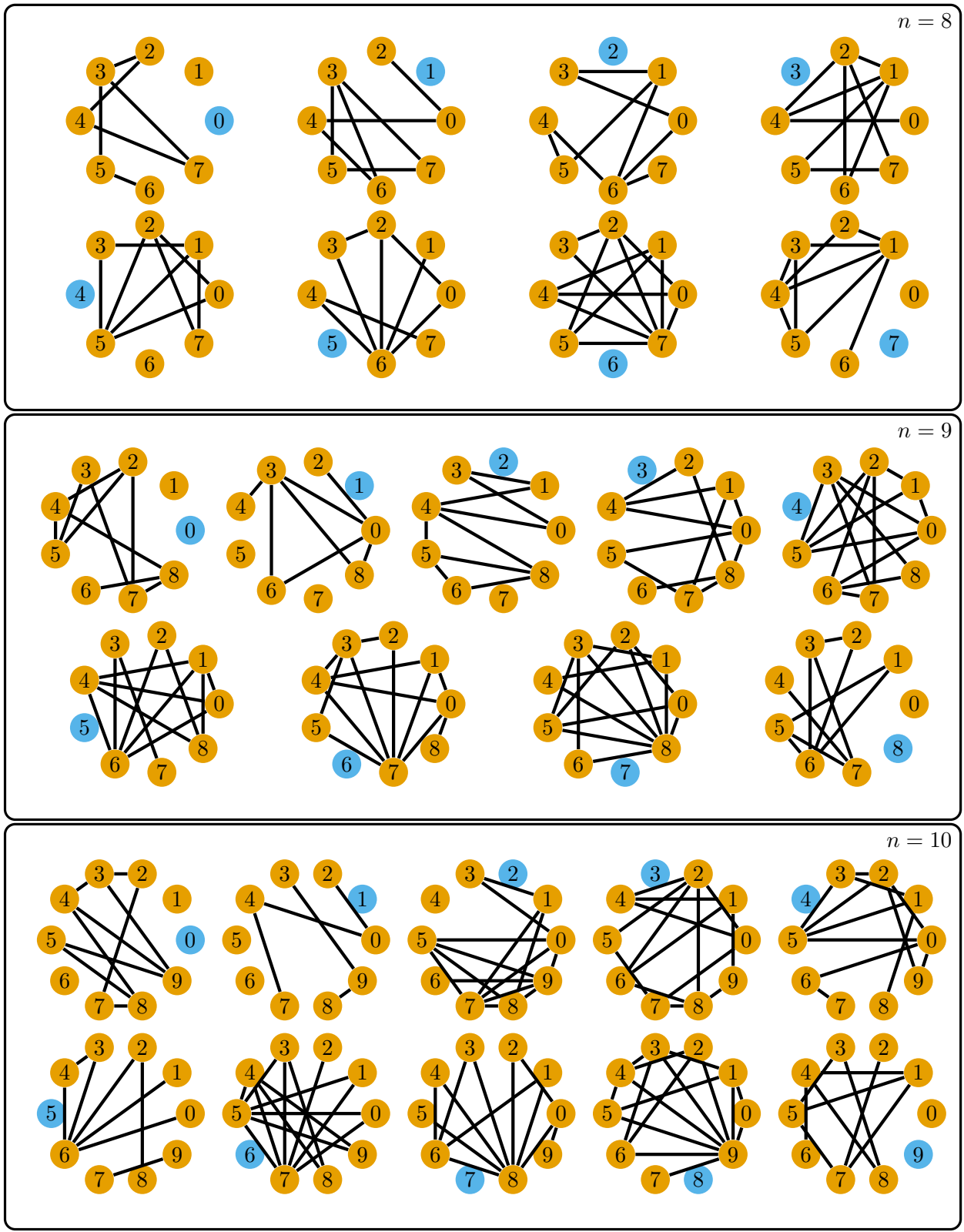
\begin{figure}[tbp]
\centering
\begin{tikzpicture}
\draw[very thick,rounded corners=5pt] (-2,-5) rectangle (14.5,2) node[below left] {$n=8$};
\begin{scope}[xshift=0.5cm]
\foreach \i/\f in {0/color1,1/color0,2/color0,3/color0,4/color0,5/color0,6/color0,7/color0}{\draw (\i*360/8:1.2cm) node[circle,inner sep=0, minimum size=0.5cm,fill=\f](\i){\i};}
\draw[-,draw,ultra thick] (2) edge (3);
\draw[-,draw,ultra thick] (2) edge (4);
\draw[-,draw,ultra thick] (3) edge (5);
\draw[-,draw,ultra thick] (3) edge (7);
\draw[-,draw,ultra thick] (4) edge (7);
\draw[-,draw,ultra thick] (5) edge (6);
\end{scope}
\begin{scope}[xshift=4.5cm]
\foreach \i/\f in {0/color0,1/color1,2/color0,3/color0,4/color0,5/color0,6/color0,7/color0}{\draw (\i*360/8:1.2cm) node[circle,inner sep=0, minimum size=0.5cm,fill=\f](\i){\i};}
\draw[-,draw,ultra thick] (0) edge (2);
\draw[-,draw,ultra thick] (0) edge (4);
\draw[-,draw,ultra thick] (3) edge (5);
\draw[-,draw,ultra thick] (3) edge (6);
\draw[-,draw,ultra thick] (3) edge (7);
\draw[-,draw,ultra thick] (4) edge (6);
\draw[-,draw,ultra thick] (5) edge (7);
\end{scope}
\begin{scope}[xshift=8.5cm]
\foreach \i/\f in {0/color0,1/color0,2/color1,3/color0,4/color0,5/color0,6/color0,7/color0}{\draw (\i*360/8:1.2cm) node[circle,inner sep=0, minimum size=0.5cm,fill=\f](\i){\i};}
\draw[-,draw,ultra thick] (0) edge (3);
\draw[-,draw,ultra thick] (0) edge (6);
\draw[-,draw,ultra thick] (1) edge (3);
\draw[-,draw,ultra thick] (1) edge (5);
\draw[-,draw,ultra thick] (1) edge (6);
\draw[-,draw,ultra thick] (4) edge (5);
\draw[-,draw,ultra thick] (4) edge (6);
\draw[-,draw,ultra thick] (6) edge (7);
\end{scope}
\begin{scope}[xshift=12.5cm]
\foreach \i/\f in {0/color0,1/color0,2/color0,3/color1,4/color0,5/color0,6/color0,7/color0}{\draw (\i*360/8:1.2cm) node[circle,inner sep=0, minimum size=0.5cm,fill=\f](\i){\i};}
\draw[-,draw,ultra thick] (0) edge (4);
\draw[-,draw,ultra thick] (1) edge (2);
\draw[-,draw,ultra thick] (1) edge (4);
\draw[-,draw,ultra thick] (1) edge (5);
\draw[-,draw,ultra thick] (1) edge (6);
\draw[-,draw,ultra thick] (2) edge (4);
\draw[-,draw,ultra thick] (2) edge (6);
\draw[-,draw,ultra thick] (2) edge (7);
\draw[-,draw,ultra thick] (5) edge (7);
\end{scope}
\begin{scope}[xshift=0.5cm,yshift=-3cm]
\foreach \i/\f in {0/color0,1/color0,2/color0,3/color0,4/color1,5/color0,6/color0,7/color0}{\draw (\i*360/8:1.2cm) node[circle,inner sep=0, minimum size=0.5cm,fill=\f](\i){\i};}
\draw[-,draw,ultra thick] (0) edge (2);
\draw[-,draw,ultra thick] (0) edge (5);
\draw[-,draw,ultra thick] (1) edge (3);
\draw[-,draw,ultra thick] (1) edge (5);
\draw[-,draw,ultra thick] (1) edge (7);
\draw[-,draw,ultra thick] (2) edge (5);
\draw[-,draw,ultra thick] (2) edge (7);
\draw[-,draw,ultra thick] (3) edge (5);
\end{scope}
\begin{scope}[xshift=4.5cm,yshift=-3cm]
\foreach \i/\f in {0/color0,1/color0,2/color0,3/color0,4/color0,5/color1,6/color0,7/color0}{\draw (\i*360/8:1.2cm) node[circle,inner sep=0, minimum size=0.5cm,fill=\f](\i){\i};}
\draw[-,draw,ultra thick] (0) edge (2);
\draw[-,draw,ultra thick] (0) edge (6);
\draw[-,draw,ultra thick] (1) edge (6);
\draw[-,draw,ultra thick] (2) edge (3);
\draw[-,draw,ultra thick] (2) edge (6);
\draw[-,draw,ultra thick] (3) edge (6);
\draw[-,draw,ultra thick] (4) edge (6);
\draw[-,draw,ultra thick] (4) edge (7);
\end{scope}
\begin{scope}[xshift=8.5cm,yshift=-3cm]
\foreach \i/\f in {0/color0,1/color0,2/color0,3/color0,4/color0,5/color0,6/color1,7/color0}{\draw (\i*360/8:1.2cm) node[circle,inner sep=0, minimum size=0.5cm,fill=\f](\i){\i};}
\draw[-,draw,ultra thick] (0) edge (2);
\draw[-,draw,ultra thick] (0) edge (4);
\draw[-,draw,ultra thick] (0) edge (7);
\draw[-,draw,ultra thick] (1) edge (4);
\draw[-,draw,ultra thick] (1) edge (5);
\draw[-,draw,ultra thick] (1) edge (7);
\draw[-,draw,ultra thick] (2) edge (3);
\draw[-,draw,ultra thick] (2) edge (5);
\draw[-,draw,ultra thick] (2) edge (7);
\draw[-,draw,ultra thick] (3) edge (7);
\draw[-,draw,ultra thick] (4) edge (7);
\draw[-,draw,ultra thick] (5) edge (7);
\end{scope}
\begin{scope}[xshift=12.5cm,yshift=-3cm]
\foreach \i/\f in {0/color0,1/color0,2/color0,3/color0,4/color0,5/color0,6/color0,7/color1}{\draw (\i*360/8:1.2cm) node[circle,inner sep=0, minimum size=0.5cm,fill=\f](\i){\i};}
\draw[-,draw,ultra thick] (1) edge (2);
\draw[-,draw,ultra thick] (1) edge (3);
\draw[-,draw,ultra thick] (1) edge (4);
\draw[-,draw,ultra thick] (1) edge (5);
\draw[-,draw,ultra thick] (1) edge (6);
\draw[-,draw,ultra thick] (2) edge (4);
\draw[-,draw,ultra thick] (3) edge (4);
\draw[-,draw,ultra thick] (3) edge (5);
\draw[-,draw,ultra thick] (4) edge (5);
\end{scope}
\end{tikzpicture}
\begin{tikzpicture}
\draw[very thick,rounded corners=5pt] (-2,-5) rectangle (14.5,2) node[below left] {$n=9$};
\begin{scope}
\foreach \i/\f in {0/color1,1/color0,2/color0,3/color0,4/color0,5/color0,6/color0,7/color0,8/color0}{\draw (\i*360/9:1.2cm) node[circle,inner sep=0, minimum size=0.5cm,fill=\f](\i){\i};}
\draw[-,draw,ultra thick] (2) edge (4);
\draw[-,draw,ultra thick] (2) edge (5);
\draw[-,draw,ultra thick] (2) edge (7);
\draw[-,draw,ultra thick] (3) edge (5);
\draw[-,draw,ultra thick] (3) edge (7);
\draw[-,draw,ultra thick] (4) edge (5);
\draw[-,draw,ultra thick] (4) edge (8);
\draw[-,draw,ultra thick] (6) edge (8);
\draw[-,draw,ultra thick] (7) edge (8);
\end{scope}
\begin{scope}[xshift=3.2cm]
\foreach \i/\f in {0/color0,1/color1,2/color0,3/color0,4/color0,5/color0,6/color0,7/color0,8/color0}{\draw (\i*360/9:1.2cm) node[circle,inner sep=0, minimum size=0.5cm,fill=\f](\i){\i};}
\draw[-,draw,ultra thick] (0) edge (2);
\draw[-,draw,ultra thick] (0) edge (3);
\draw[-,draw,ultra thick] (0) edge (6);
\draw[-,draw,ultra thick] (0) edge (8);
\draw[-,draw,ultra thick] (3) edge (4);
\draw[-,draw,ultra thick] (3) edge (6);
\draw[-,draw,ultra thick] (3) edge (8);
\end{scope}
\begin{scope}[xshift=6.4cm]
\foreach \i/\f in {0/color0,1/color0,2/color1,3/color0,4/color0,5/color0,6/color0,7/color0,8/color0}{\draw (\i*360/9:1.2cm) node[circle,inner sep=0, minimum size=0.5cm,fill=\f](\i){\i};}
\draw[-,draw,ultra thick] (0) edge (3);
\draw[-,draw,ultra thick] (0) edge (4);
\draw[-,draw,ultra thick] (1) edge (3);
\draw[-,draw,ultra thick] (1) edge (4);
\draw[-,draw,ultra thick] (4) edge (5);
\draw[-,draw,ultra thick] (4) edge (8);
\draw[-,draw,ultra thick] (5) edge (6);
\draw[-,draw,ultra thick] (5) edge (8);
\draw[-,draw,ultra thick] (6) edge (8);
\end{scope}
\begin{scope}[xshift=9.6cm]
\foreach \i/\f in {0/color0,1/color0,2/color0,3/color1,4/color0,5/color0,6/color0,7/color0,8/color0}{\draw (\i*360/9:1.2cm) node[circle,inner sep=0, minimum size=0.5cm,fill=\f](\i){\i};}
\draw[-,draw,ultra thick] (0) edge (1);
\draw[-,draw,ultra thick] (0) edge (4);
\draw[-,draw,ultra thick] (0) edge (5);
\draw[-,draw,ultra thick] (0) edge (8);
\draw[-,draw,ultra thick] (1) edge (4);
\draw[-,draw,ultra thick] (1) edge (7);
\draw[-,draw,ultra thick] (2) edge (4);
\draw[-,draw,ultra thick] (2) edge (8);
\draw[-,draw,ultra thick] (5) edge (7);
\draw[-,draw,ultra thick] (6) edge (8);
\draw[-,draw,ultra thick] (7) edge (8);
\end{scope}
\begin{scope}[xshift=12.8cm]
\foreach \i/\f in {0/color0,1/color0,2/color0,3/color0,4/color1,5/color0,6/color0,7/color0,8/color0}{\draw (\i*360/9:1.2cm) node[circle,inner sep=0, minimum size=0.5cm,fill=\f](\i){\i};}
\draw[-,draw,ultra thick] (0) edge (1);
\draw[-,draw,ultra thick] (0) edge (3);
\draw[-,draw,ultra thick] (0) edge (5);
\draw[-,draw,ultra thick] (0) edge (6);
\draw[-,draw,ultra thick] (1) edge (2);
\draw[-,draw,ultra thick] (1) edge (5);
\draw[-,draw,ultra thick] (2) edge (5);
\draw[-,draw,ultra thick] (2) edge (6);
\draw[-,draw,ultra thick] (2) edge (7);
\draw[-,draw,ultra thick] (3) edge (5);
\draw[-,draw,ultra thick] (3) edge (7);
\draw[-,draw,ultra thick] (3) edge (8);
\draw[-,draw,ultra thick] (6) edge (7);
\draw[-,draw,ultra thick] (6) edge (8);
\end{scope}
\begin{scope}[xshift=0.5cm,yshift=-3cm]
\foreach \i/\f in {0/color0,1/color0,2/color0,3/color0,4/color0,5/color1,6/color0,7/color0,8/color0}{\draw (\i*360/9:1.2cm) node[circle,inner sep=0, minimum size=0.5cm,fill=\f](\i){\i};}
\draw[-,draw,ultra thick] (0) edge (1);
\draw[-,draw,ultra thick] (0) edge (4);
\draw[-,draw,ultra thick] (0) edge (6);
\draw[-,draw,ultra thick] (1) edge (4);
\draw[-,draw,ultra thick] (1) edge (6);
\draw[-,draw,ultra thick] (1) edge (8);
\draw[-,draw,ultra thick] (2) edge (6);
\draw[-,draw,ultra thick] (2) edge (8);
\draw[-,draw,ultra thick] (3) edge (6);
\draw[-,draw,ultra thick] (3) edge (7);
\draw[-,draw,ultra thick] (4) edge (6);
\draw[-,draw,ultra thick] (4) edge (8);
\end{scope}
\begin{scope}[xshift=4.5cm,yshift=-3cm]
\foreach \i/\f in {0/color0,1/color0,2/color0,3/color0,4/color0,5/color0,6/color1,7/color0,8/color0}{\draw (\i*360/9:1.2cm) node[circle,inner sep=0, minimum size=0.5cm,fill=\f](\i){\i};}
\draw[-,draw,ultra thick] (0) edge (1);
\draw[-,draw,ultra thick] (0) edge (4);
\draw[-,draw,ultra thick] (0) edge (7);
\draw[-,draw,ultra thick] (0) edge (8);
\draw[-,draw,ultra thick] (1) edge (4);
\draw[-,draw,ultra thick] (1) edge (7);
\draw[-,draw,ultra thick] (2) edge (3);
\draw[-,draw,ultra thick] (2) edge (7);
\draw[-,draw,ultra thick] (3) edge (4);
\draw[-,draw,ultra thick] (3) edge (5);
\draw[-,draw,ultra thick] (3) edge (7);
\draw[-,draw,ultra thick] (4) edge (7);
\draw[-,draw,ultra thick] (5) edge (7);
\end{scope}
\begin{scope}[xshift=8.5cm,yshift=-3cm]
\foreach \i/\f in {0/color0,1/color0,2/color0,3/color0,4/color0,5/color0,6/color0,7/color1,8/color0}{\draw (\i*360/9:1.2cm) node[circle,inner sep=0, minimum size=0.5cm,fill=\f](\i){\i};}
\draw[-,draw,ultra thick] (0) edge (2);
\draw[-,draw,ultra thick] (0) edge (5);
\draw[-,draw,ultra thick] (0) edge (8);
\draw[-,draw,ultra thick] (1) edge (3);
\draw[-,draw,ultra thick] (1) edge (4);
\draw[-,draw,ultra thick] (1) edge (8);
\draw[-,draw,ultra thick] (2) edge (5);
\draw[-,draw,ultra thick] (2) edge (8);
\draw[-,draw,ultra thick] (3) edge (5);
\draw[-,draw,ultra thick] (3) edge (6);
\draw[-,draw,ultra thick] (3) edge (8);
\draw[-,draw,ultra thick] (4) edge (8);
\draw[-,draw,ultra thick] (5) edge (8);
\draw[-,draw,ultra thick] (6) edge (8);
\end{scope}
\begin{scope}[xshift=12.5cm,yshift=-3cm]
\foreach \i/\f in {0/color0,1/color0,2/color0,3/color0,4/color0,5/color0,6/color0,7/color0,8/color1}{\draw (\i*360/9:1.2cm) node[circle,inner sep=0, minimum size=0.5cm,fill=\f](\i){\i};}
\draw[-,draw,ultra thick] (1) edge (5);
\draw[-,draw,ultra thick] (1) edge (6);
\draw[-,draw,ultra thick] (2) edge (3);
\draw[-,draw,ultra thick] (2) edge (6);
\draw[-,draw,ultra thick] (3) edge (6);
\draw[-,draw,ultra thick] (3) edge (7);
\draw[-,draw,ultra thick] (4) edge (7);
\draw[-,draw,ultra thick] (5) edge (6);
\draw[-,draw,ultra thick] (5) edge (7);
\end{scope}
\end{tikzpicture}
\begin{tikzpicture}
\draw[very thick,rounded corners=5pt] (-2,-5) rectangle (14.5,2) node[below left] {$n=10$};
\begin{scope}
\foreach \i/\f in {0/color1,1/color0,2/color0,3/color0,4/color0,5/color0,6/color0,7/color0,8/color0,9/color0}{\draw (\i*360/10:1.2cm) node[circle,inner sep=0, minimum size=0.5cm,fill=\f](\i){\i};}
\draw[-,draw,ultra thick] (2) edge (3);
\draw[-,draw,ultra thick] (2) edge (7);
\draw[-,draw,ultra thick] (3) edge (4);
\draw[-,draw,ultra thick] (3) edge (9);
\draw[-,draw,ultra thick] (4) edge (8);
\draw[-,draw,ultra thick] (4) edge (9);
\draw[-,draw,ultra thick] (5) edge (8);
\draw[-,draw,ultra thick] (5) edge (9);
\draw[-,draw,ultra thick] (7) edge (8);
\end{scope}
\begin{scope}[xshift=3.2cm]
\foreach \i/\f in {0/color0,1/color1,2/color0,3/color0,4/color0,5/color0,6/color0,7/color0,8/color0,9/color0}{\draw (\i*360/10:1.2cm) node[circle,inner sep=0, minimum size=0.5cm,fill=\f](\i){\i};}
\draw[-,draw,ultra thick] (0) edge (2);
\draw[-,draw,ultra thick] (0) edge (4);
\draw[-,draw,ultra thick] (3) edge (9);
\draw[-,draw,ultra thick] (4) edge (7);
\draw[-,draw,ultra thick] (8) edge (9);
\end{scope}
\begin{scope}[xshift=6.4cm]
\foreach \i/\f in {0/color0,1/color0,2/color1,3/color0,4/color0,5/color0,6/color0,7/color0,8/color0,9/color0}{\draw (\i*360/10:1.2cm) node[circle,inner sep=0, minimum size=0.5cm,fill=\f](\i){\i};}
\draw[-,draw,ultra thick] (0) edge (3);
\draw[-,draw,ultra thick] (0) edge (5);
\draw[-,draw,ultra thick] (0) edge (7);
\draw[-,draw,ultra thick] (0) edge (9);
\draw[-,draw,ultra thick] (1) edge (3);
\draw[-,draw,ultra thick] (1) edge (7);
\draw[-,draw,ultra thick] (1) edge (8);
\draw[-,draw,ultra thick] (5) edge (7);
\draw[-,draw,ultra thick] (5) edge (8);
\draw[-,draw,ultra thick] (5) edge (9);
\draw[-,draw,ultra thick] (6) edge (9);
\draw[-,draw,ultra thick] (7) edge (8);
\draw[-,draw,ultra thick] (7) edge (9);
\draw[-,draw,ultra thick] (8) edge (9);
\end{scope}
\begin{scope}[xshift=9.6cm]
\foreach \i/\f in {0/color0,1/color0,2/color0,3/color1,4/color0,5/color0,6/color0,7/color0,8/color0,9/color0}{\draw (\i*360/10:1.2cm) node[circle,inner sep=0, minimum size=0.5cm,fill=\f](\i){\i};}
\draw[-,draw,ultra thick] (0) edge (2);
\draw[-,draw,ultra thick] (0) edge (4);
\draw[-,draw,ultra thick] (0) edge (7);
\draw[-,draw,ultra thick] (1) edge (4);
\draw[-,draw,ultra thick] (1) edge (6);
\draw[-,draw,ultra thick] (1) edge (9);
\draw[-,draw,ultra thick] (2) edge (4);
\draw[-,draw,ultra thick] (2) edge (5);
\draw[-,draw,ultra thick] (2) edge (6);
\draw[-,draw,ultra thick] (2) edge (8);
\draw[-,draw,ultra thick] (5) edge (7);
\draw[-,draw,ultra thick] (6) edge (8);
\draw[-,draw,ultra thick] (7) edge (8);
\draw[-,draw,ultra thick] (8) edge (9);
\end{scope}
\begin{scope}[xshift=12.8cm]
\foreach \i/\f in {0/color0,1/color0,2/color0,3/color0,4/color1,5/color0,6/color0,7/color0,8/color0,9/color0}{\draw (\i*360/10:1.2cm) node[circle,inner sep=0, minimum size=0.5cm,fill=\f](\i){\i};}
\draw[-,draw,ultra thick] (0) edge (2);
\draw[-,draw,ultra thick] (0) edge (5);
\draw[-,draw,ultra thick] (0) edge (6);
\draw[-,draw,ultra thick] (0) edge (9);
\draw[-,draw,ultra thick] (1) edge (3);
\draw[-,draw,ultra thick] (1) edge (5);
\draw[-,draw,ultra thick] (1) edge (8);
\draw[-,draw,ultra thick] (2) edge (3);
\draw[-,draw,ultra thick] (2) edge (5);
\draw[-,draw,ultra thick] (2) edge (9);
\draw[-,draw,ultra thick] (3) edge (5);
\draw[-,draw,ultra thick] (6) edge (7);
\end{scope}
\begin{scope}[yshift=-3cm]
\foreach \i/\f in {0/color0,1/color0,2/color0,3/color0,4/color0,5/color1,6/color0,7/color0,8/color0,9/color0}{\draw (\i*360/10:1.2cm) node[circle,inner sep=0, minimum size=0.5cm,fill=\f](\i){\i};}
\draw[-,draw,ultra thick] (0) edge (6);
\draw[-,draw,ultra thick] (1) edge (6);
\draw[-,draw,ultra thick] (2) edge (6);
\draw[-,draw,ultra thick] (2) edge (8);
\draw[-,draw,ultra thick] (3) edge (4);
\draw[-,draw,ultra thick] (3) edge (6);
\draw[-,draw,ultra thick] (4) edge (6);
\draw[-,draw,ultra thick] (7) edge (9);
\end{scope}
\begin{scope}[xshift=3.2cm,yshift=-3cm]
\foreach \i/\f in {0/color0,1/color0,2/color0,3/color0,4/color0,5/color0,6/color1,7/color0,8/color0,9/color0}{\draw (\i*360/10:1.2cm) node[circle,inner sep=0, minimum size=0.5cm,fill=\f](\i){\i};}
\draw[-,draw,ultra thick] (0) edge (5);
\draw[-,draw,ultra thick] (0) edge (7);
\draw[-,draw,ultra thick] (1) edge (5);
\draw[-,draw,ultra thick] (1) edge (7);
\draw[-,draw,ultra thick] (2) edge (7);
\draw[-,draw,ultra thick] (3) edge (5);
\draw[-,draw,ultra thick] (3) edge (7);
\draw[-,draw,ultra thick] (3) edge (8);
\draw[-,draw,ultra thick] (4) edge (5);
\draw[-,draw,ultra thick] (4) edge (7);
\draw[-,draw,ultra thick] (4) edge (8);
\draw[-,draw,ultra thick] (4) edge (9);
\draw[-,draw,ultra thick] (5) edge (7);
\draw[-,draw,ultra thick] (5) edge (9);
\end{scope}
\begin{scope}[xshift=6.4cm,yshift=-3cm]
\foreach \i/\f in {0/color0,1/color0,2/color0,3/color0,4/color0,5/color0,6/color0,7/color1,8/color0,9/color0}{\draw (\i*360/10:1.2cm) node[circle,inner sep=0, minimum size=0.5cm,fill=\f](\i){\i};}
\draw[-,draw,ultra thick] (0) edge (2);
\draw[-,draw,ultra thick] (0) edge (8);
\draw[-,draw,ultra thick] (0) edge (9);
\draw[-,draw,ultra thick] (1) edge (6);
\draw[-,draw,ultra thick] (1) edge (8);
\draw[-,draw,ultra thick] (2) edge (8);
\draw[-,draw,ultra thick] (3) edge (6);
\draw[-,draw,ultra thick] (3) edge (8);
\draw[-,draw,ultra thick] (4) edge (6);
\draw[-,draw,ultra thick] (4) edge (8);
\draw[-,draw,ultra thick] (5) edge (8);
\draw[-,draw,ultra thick] (6) edge (8);
\end{scope}
\begin{scope}[xshift=9.6cm,yshift=-3cm]
\foreach \i/\f in {0/color0,1/color0,2/color0,3/color0,4/color0,5/color0,6/color0,7/color0,8/color1,9/color0}{\draw (\i*360/10:1.2cm) node[circle,inner sep=0, minimum size=0.5cm,fill=\f](\i){\i};}
\draw[-,draw,ultra thick] (0) edge (1);
\draw[-,draw,ultra thick] (0) edge (9);
\draw[-,draw,ultra thick] (1) edge (3);
\draw[-,draw,ultra thick] (1) edge (5);
\draw[-,draw,ultra thick] (1) edge (9);
\draw[-,draw,ultra thick] (2) edge (4);
\draw[-,draw,ultra thick] (2) edge (6);
\draw[-,draw,ultra thick] (2) edge (9);
\draw[-,draw,ultra thick] (3) edge (4);
\draw[-,draw,ultra thick] (3) edge (6);
\draw[-,draw,ultra thick] (3) edge (9);
\draw[-,draw,ultra thick] (4) edge (5);
\draw[-,draw,ultra thick] (4) edge (9);
\draw[-,draw,ultra thick] (5) edge (9);
\draw[-,draw,ultra thick] (6) edge (9);
\draw[-,draw,ultra thick] (7) edge (9);
\end{scope}
\begin{scope}[xshift=12.8cm,yshift=-3cm]
\foreach \i/\f in {0/color0,1/color0,2/color0,3/color0,4/color0,5/color0,6/color0,7/color0,8/color0,9/color1}{\draw (\i*360/10:1.2cm) node[circle,inner sep=0, minimum size=0.5cm,fill=\f](\i){\i};}
\draw[-,draw,ultra thick] (1) edge (4);
\draw[-,draw,ultra thick] (1) edge (5);
\draw[-,draw,ultra thick] (1) edge (7);
\draw[-,draw,ultra thick] (2) edge (7);
\draw[-,draw,ultra thick] (3) edge (5);
\draw[-,draw,ultra thick] (3) edge (8);
\draw[-,draw,ultra thick] (4) edge (6);
\draw[-,draw,ultra thick] (4) edge (8);
\draw[-,draw,ultra thick] (5) edge (7);
\end{scope}
\end{tikzpicture}
\caption{Multigraphs satisfying the conditions stated in \autoref{lem:edge-coloring} for $n\in \NN$ with $n\in \{8,9,10\}$. The notation is the same as that used for \autoref{fig:multigraphs-n-5-to-7}. For each $n\in \{8,9,10\},~\bfrho\in R_n$ is given by $\rho_i=i+_n 1$ for each $i\in [n]$.}
\label{fig:multigraphs-n-8-to-10}
\end{figure}

A multigraph satisfying the conditions of \autoref{lem:edge-coloring} can be found in polynomial time. This follows from the work by \citet{chandrasekaran2013deterministic} because the total number of events, the number of variables determining each event, and the size of the domain of each random variable are all polynomial, and because~\eqref{eq:lll-condition} can be strengthened to
\[
    \frac{1}{2^{n-4}}(e(4n-9))^{1+\varepsilon} \leq 1
\]
for some constant $\varepsilon>0$. Indeed, the number of events is equal to the number $n(n-1)(n-2)/2$ of triples of vertices, each event is determined by the $2(n-3)$ colored edges that may or not be drawn for a given triple, and all random variables have domain $\{0,1\}$. When $\epsilon=1/100$, the stronger inequality still holds for $n\geq 11$.

We are now ready to prove \autoref{lem:existence-blocking-sets}.
\begin{proof}[Proof of \autoref{lem:existence-blocking-sets}]

Consider $n\in \NN$ with $n\geq 5$, and let $\bfrho\in R_n$ and $G=(N,E)\in \calG(n)$ be such that Conditions~\ref{item:multigraph-monotonicity}-\ref{item:multigraph-paths} in the statement of \autoref{lem:edge-coloring} are satisfied, whose existence is guaranteed by this lemma. We claim the result for $\bfrho$ and sets $S^b_{ij}$ for each $i,j\in [n]$ with $i\not=j$ and $b\in \{0,1\}$ defined as
\begin{equation}
    S^0_{ij} = N_i(j,G) \text{ and } S^1_{ij} = ([n]\setminus \{i,j\})\setminus  N_i(j,G) \quad  \text{for every } i,j\in [n] \text{ with }i\not=j.\label{eq:def-blocking-sets}
\end{equation}
We show in the following Properties~\ref{item:ifr}-\ref{item:mon} in the statement of the lemma.

Property~\ref{item:ifr} follows directly from the definition of the blocking sets in~\eqref{eq:def-blocking-sets}, together with the fact that $N_i(j,G) \cap \{i,j\} = \emptyset$ for every $i,j\in [n]$ with $i\not=j$ from the definition of $\calG_n$.

To see Property~\ref{item:one-pos-per-agent}, we suppose for the sake of contradiction that there exists $j\in [n]$ and $\bfb\in \{0,1\}^n$ such that the set $\bigcup_{i\in [n]\setminus \{j\}} S^{b_i}_{ij}$ has size different from $n-2$ and $n-1$. Since Property~\ref{item:ifr} implies that $j$ does not belong to this set, we must have that $\big| \bigcup_{i\in [n]\setminus \{j\}} S^{b_i}_{ij} \big| \leq n-3$. Therefore, there are values $k,\ell\in [n]\setminus \{j\}$ with $k\not=\ell$ such that
\begin{equation}
    S^{b_i}_{ij} \cap \{k,\ell\} = \emptyset \text{ for every } i\in [n]\setminus \{j\}.\label{eq:implication-more-pos-per-agent}
\end{equation}
However, by Condition~\ref{item:multigraph-paths} in \autoref{lem:edge-coloring} we have that there is $i\in [n]\setminus \{j,k,\ell\}$ such that $|N_i(j,G) \cap \{k,\ell\}| = 1$. Using the definition of the blocking sets in~\eqref{eq:def-blocking-sets}, this yields $|S^0_{ij} \cap \{k,\ell\}| = |S^1_{ij} \cap \{k,\ell\}|=1$, a contradiction to~\eqref{eq:implication-more-pos-per-agent}.

To see Property \ref{item:one-agent-per-pos-def}, we first note that from the definition of the blocking sets in~\eqref{eq:def-blocking-sets} we have that
\begin{equation}
    \text{ for every } i\in [n], b\in \{0,1\}, \text{ and }j,k\in [n]\setminus \{i\} \text{ with }j\not = k\colon k\in S^b_{ij} \Longleftrightarrow j\in S^b_{ik}.\label{eq:symmetry-blocking-sets}
\end{equation}
Let $j\in [n]$ and $\bfb\in \{0,1\}^n$ be such that $\bigcup_{i\in [n]\setminus \{j\}} S^{b_i}_{ij} = [n]\setminus \{j\}$, and fix $j'\in [n]\setminus \{j\}$ arbitrarily. Since $j'\in \bigcup_{i\in [n]\setminus \{j\}} S^{b_i}_{ij}$ and $j'\not\in S^{b_{j'}}_{j'j}$ from the definition of $\calG_n$, there exists $i'\in [n]\setminus \{j,j'\}$ such that $j'\in  S^{b_{i'}}_{i'j}$. By~\eqref{eq:symmetry-blocking-sets}, this implies that $j\in  S^{b_{i'}}_{i'j'}$, thus $j\in \bigcup_{i\in [n]\setminus \{j'\}} S^{b_{i}}_{ij'}$.

To prove Property~\ref{item:one-agent-per-pos-non-def}, let $j, k\in [n]$ with $j\not= k$ and $\bfb\in \{0,1\}^n$ be such that $\bigcup_{i\in [n]\setminus \{j\}} S^{b_i}_{ij} = [n]\setminus \{j,k\}$. This implies that $k\not\in S^{b_i}_{ij}$ for every $i\in [n]\setminus \{j,k\}$. Take now $j'\in [n]\setminus \{j,k\}$ arbitrarily and suppose, for the sake of contradiction, that for every $i\in [n]\setminus \{j',k\}$ we also have that $k\not\in S^{b_i}_{ij'}$. These two properties imply that, for every $i\in [n]\setminus \{j,j',k\}$ we have that $k\not\in S^{b_i}_{ij} \cup S^{b_i}_{ij'}$. Using the definition of the blocking sets in~\eqref{eq:def-blocking-sets}, this implies that either $k\in N_i(j,G) \cap N_i(j',G)$ or that $k\not\in N_i(j,G) \cup N_i(j',G)$ for every $i\in [n]\setminus \{j,j',k\}$. This is equivalent to
\begin{equation}
    |N_i(k,G) \cap \{j,j'\}| \in \{0,2\} \text{ for every } i\in [n]\setminus \{j,j',k\},
\end{equation}
a contradiction to Condition~\ref{item:multigraph-paths} in \autoref{lem:edge-coloring}. We conclude that $k\in \bigcup_{i\in [n]\setminus \{j'\}} S^{b_{i}}_{ij'}$.

Finally, Property~\ref{item:mon} follows directly from the definition of the blocking sets in~\eqref{eq:def-blocking-sets} and Condition~\ref{item:multigraph-monotonicity} in \autoref{lem:edge-coloring}.

\end{proof}

\subsection{Proof of \autoref{thm:monotonicity}}
\label{app:pf-thm-monotonicity}

Let $n\in\NN$ with $n\geq 4$. If $n=4$ the result follows directly from \autoref{lem:suf-conds-g} and \autoref{lem:existence-g-n4}, so we assume $n\geq 5$ henceforth.
    
Consider $\bfrho\in R_n$ and sets $S^b_{ij}$, for each $i,j\in[n]$ with $i\not=j$ and $b\in \{0,1\}$, satisfying \autoref{lem:existence-blocking-sets}. For a given message profile $\bfb\in\{0,1\}^n$, let $A_j(\bfb_{-j})$ be the set of positions available for agent~$j$, in addition to its default position $j$, \ie 
\begin{equation}
    \textstyle A_j(\bfb_{-j}) = ([n]\setminus \{j\}) \setminus \bigcup_{i\in [n]\setminus \{j\}} S^{b_i}_{ij}.\label{eq:available-sets}
\end{equation}
Then, by Conditions~\ref{item:ifr} and~\ref{item:one-pos-per-agent} of \autoref{lem:existence-blocking-sets}, for every $j\in[n]$ and $\bfb\in \{0,1\}^n$, either $A_j(\bfb_{-j})=\{k\}$ for some $k\neq j$, or $A_j(\bfb_{-j})=\emptyset$. Define $g\colon\{0,1\}^n\to\calP_n$ such that for all $j\in[n]$ and $\bfb\in \{0,1\}^n$,
\begin{equation}
    (g(\bfb))^{-1}(j) = \begin{cases} 
        k & \text{ if } A_j(\bfb_{-j}) = \{k\},\\
        j & \text{ if } A_j(\bfb_{-j}) = \emptyset.
    \end{cases} \label{eq:def-g}
\end{equation}
We now claim the result for $f_{g,\bfrho}$. By \autoref{lem:suf-conds-g}, it suffices to show that~$g$ is well-defined, in the sense that $(g(\bfb))^{-1}$ is a bijection from $[n]$ to $[n]$ for every $\bfb\in \{0,1\}^n$, %
and that~$g$ and~$\bfrho$ satisfy Conditions~\ref{item:g-impartial}-\ref{item:g-monotone} of the lemma. We show each of these in turn.

We start by showing that~\eqref{eq:def-g} defines a bijection between agents and positions, \ie that $\{(g(\bfb))^{-1}(j)\colon j\in [n]\} = [n]$ for every $\bfb\in \{0,1\}^n$. 
It is clear from~\eqref{eq:def-g} that $(g(\bfb))^{-1}(j)\in[n]$ for every $j\in [n]$ and $\bfb\in \{0,1\}^n$. 
Then, it remains to show that for every $j,k\in[n]$ and $\bfb\in \{0,1\}^n$ such that $(g(\bfb))^{-1}(j)=k$, it holds that $(g(\bfb))^{-1}(j') \neq k$ for all $j'\in [n]\setminus\{j\}$. Let then $j,k\in[n]$ and $\bfb\in \{0,1\}^n$ be such that $(g(\bfb))^{-1}(j)=k$.
We first consider the case where $k=j$. Then, by~\eqref{eq:available-sets} and~\eqref{eq:def-g}, $\bigcup_{i\in [n]\setminus \{j\}} S^{b_i}_{ij} = [n]\setminus \{j\}$, and Condition~\ref{item:one-agent-per-pos-def} of \autoref{lem:existence-blocking-sets} implies that $k\in \bigcup_{i\in [n]\setminus\{j'\}} S^{b_i}_{ij'}$ for every $j'\in [n]\setminus \{j\}$. Thus, by~\eqref{eq:available-sets} and~\eqref{eq:def-g}, $(g(\bfb))^{-1}(j')\not=k$. Now consider the case where $k\neq j$. Then, by~\eqref{eq:available-sets} and~\eqref{eq:def-g}, $\bigcup_{i\in [n]\setminus \{j\}} S^{b_i}_{ij} = [n]\setminus \{j,k\}$, and Condition~\ref{item:one-agent-per-pos-non-def} of \autoref{lem:existence-blocking-sets} implies that $k\in\bigcup_{i\in [n]\setminus \{j'\}} S^{b_i}_{ij'}$ for every $j'\in [n]\setminus \{j,k\}$. Thus, by~\eqref{eq:available-sets} and~\eqref{eq:def-g}, $(g(\bfb))^{-1}(j') \not=k$.

That~$g$ satisfies Condition~\ref{item:g-impartial} of \autoref{lem:suf-conds-g} follows directly from~\eqref{eq:def-g}, since $(g(\bfb))^{-1}(j)$ does not depend on $b_j$ for any $j\in[n]$ and $\bfb\in\{0,1\}^n$.

To see that $g$ satisfies Condition~\ref{item:g-ifr} of \autoref{lem:suf-conds-g}, let $j,k\in [n]$ be arbitrary. We show that that there exists $\bfb\in\{0,1\}^n$ with $(g(\bfb))(k) = j$, distinguishing the cases where $j=k$ and $j\neq k$.

For $j=k$, by~\eqref{eq:available-sets} and~\eqref{eq:def-g}, it suffices to show the existence of $\bfb\in \{0,1\}^{n}$ such that $\bigcup_{i\in [n]\setminus \{j\}} S^{b_i}_{ij}=[n]\setminus \{j\}$. Let $[n]\setminus\{j\}=\{i_{p}\}_{p=0}^{n-2}$ such that $i_0<i_1<\dots <i_{n-2}$. Fix $b_j$ arbitrarily and, for each $i\in [n]\setminus \{j\}$, define $b_i$ inductively as follows: take $b_{i_0}\in\{0,1\}$ arbitrarily such that $|S^{b_{i_0}}_{i_0j}|\geq 2$, which exists by Condition~\ref{item:ifr} of \autoref{lem:existence-blocking-sets}; for $q\in \{1,2,\ldots, n-2\}$, take $b_{i_{q}}\in\{0,1\}$ arbitrarily such that
\begin{equation}
    \Bigg| \bigcup_{p=0}^{q} S^{b_{i_{p}}}_{i_{p}j}\Bigg| \geq \Bigg| \bigcup_{p=0}^{q-1} S^{b_{i_{p}}}_{i_{p}j}\Bigg| + 1 \label{eq:growing-blocking-sets}
\end{equation}
if such a value of $b_{i_{q}}$ exists, and $b_{i_{q}}=0$ otherwise. 
We claim that $\bigcup_{i\in [n]\setminus \{j\}} S^{b_i}_{ij} = [n]\setminus \{j\}$. Assume for contradiction that this was not the case; then, for every $q\in [n-2]$, $\bigcup_{p=0}^{q-1} S^{b_{i_{p}}}_{i_{p}j}\not=[n]\setminus \{j\}$. 
Condition~\ref{item:ifr} of \autoref{lem:existence-blocking-sets} states that $S^0_{i_{q}j} \cup S^1_{i_{q}j} = [n]\setminus \{i_{q}, j\}$ for each $q\in \{1,2,\ldots, n-2\}$, thus $b_{i_{q}}$ satisfying~\eqref{eq:growing-blocking-sets} exists as long as $[n]\setminus \{i_{q}, j\} \not\subseteq \bigcup_{p=0}^{q-1} S^{b_{i_{p}}}_{i_{p}j}$. Moreover, denoting 
\[
    \bar{L} = \left\{ q\in \{1,2,n-2\}\colon [n]\setminus \{i_{q}, j\} \subseteq \bigcup_{p=0}^{q-1} S^{b_{i_{p}}}_{i_{p}j} \right\},
\]
we have that $|\bar{L}|\leq 1$. Indeed, existence of $q,q'\in\bar{L}$ with $q<q'$ would imply that $i_{q'}\in \bigcup_{p=0}^{q-1} S^{b_{i_{p}}}_{i_{p}j}$ but $i_{q'}\not\in \bigcup_{p=0}^{q'-1} S^{b_{i_{p}}}_{i_{p}j}$, a contradiction. Thus~\eqref{eq:growing-blocking-sets} holds for every $q\in \{1,2,\ldots,n-2\} \setminus \bar{L}$, which together with the fact that $\big|S^{b_{i_0}}_{i_0j}\big|\geq 2$ implies that
\[
    \Bigg|\bigcup_{p=0}^{q} S^{b_{i_{p}}}_{i_{p}j}\Bigg| \geq 2 + \big|\{1,2\ldots,q\} \setminus \bar{L}\big| \quad \text{for every $q\in \{1,2,\ldots,n-2\}$}.
\]
Since $|\bar{L}|\leq 1$ we conclude that  $\big| \bigcup_{i\in [n]\setminus \{j\}} S^{b_i}_{ij}\big| = \big|\bigcup_{p=0}^{n-2} S^{b_{i_{p}}}_{i_{p}j}\big| \geq n-1$, a contradiction.

For $j\neq k$, by~\eqref{eq:available-sets} and~\eqref{eq:def-g}, it suffices to show the existence of $\bfb\in \{0,1\}^n$ such that $k\not\in \bigcup_{i\in [n]\setminus \{j\}} S^{b_i}_{ij}$. By Condition~\ref{item:ifr} of \autoref{lem:existence-blocking-sets} it follows that $k\not\in S^0_{kj}\cup S^1_{kj}$ and that, for every $i\in [n]\setminus \{j,k\}$, either $k\not\in S^{0}_{ij}$ or $k\not\in S^{1}_{ij}$. Taking $b_k\in \{0,1\}$ arbitrarily and, for each $i\in [n]\setminus\{j,k\}$, $b_i$ such that $k\not\in S^{b_i}_{ij}$, we have that $k\not\in \bigcup_{i\in [n]\setminus \{j\}} S^{b_i}_{ij}$.

We finally prove that $g$ and $\bfrho$ satisfy Condition~\ref{item:g-monotone} of \autoref{lem:suf-conds-g}. Let $j\in [n]$ and $\bfb\in\{0,1\}^{n}$. Denote by $\bfb_{-\{j,\rho_{j}\}}$ the profile of messages for all agents except $j$ and $\rho_{j}$, and let $\bar{A}_{\rho_{j}}(\bfb_{-\{j,\rho_{j}\}})=([n]\setminus \{\rho_{j}\}) \setminus \bigcup_{i\in [n]\setminus \{j,\rho_{j}\}} S^{b_i}_{i\rho_{j}}$ be the set of positions \emph{not} blocked for agent $\rho_{j}$ by any agent except $j$. Then, by~\eqref{eq:available-sets} and Condition~\ref{item:mon} of \autoref{lem:existence-blocking-sets},
\begin{align}
    A_{\rho_{j}}\big(0,\bfb_{-\{j,\rho_{j}\}}\big) & = \left\{ k\in \bar{A}_{\rho_{j}}\big(\bfb_{-\{j,\rho_{j}\}}\big)\colon k>\rho_{j} \text{ or } k=j\right\}, \label{eq:available-pos-0} \\
   A_{\rho_{j}}\big(1,\bfb_{-\{j,\rho_{j}\}}\big) & = \left\{ k\in \bar{A}_{\rho_{j}}\big(\bfb_{-\{j,\rho_{j}\}}\big)\colon k<\rho_{j} \text{ or } k=j\right\},\label{eq:available-pos-1}
\end{align}
where the first argument of $A_{\rho_{j}}$ is the message of agent~$j$.
We claim that
\begin{equation}
    (g(1,\bfb_{-j}))^{-1}(\rho_{j}) \leq (g(0,\bfb_{-j}))^{-1}(\rho_{j}).  \label{eq:g-blocking-sets-mon}
\end{equation}
If $(g(0,\bfb_{-j}))^{-1}(\rho_{j})=j$, then $j\in A_{\rho_{j}}\big(0,\bfb_{-\{j,\rho_{j}\}}\big)$ by~\eqref{eq:def-g}, $j\in \bar{A}_{\rho_{j}}(\bfb_{-\{j,\rho_{j}\}})$ by~\eqref{eq:available-pos-0}, $j\in A_{\rho_{j}}\big(1,\bfb_{-\{j,\rho_{j}\}}\big)$ by~\eqref{eq:available-pos-1}, and $(g(1,\bfb_{-j}))^{-1}(\rho_{j}) = j$ by~\eqref{eq:def-g}, so that~\eqref{eq:g-blocking-sets-mon} holds with equality. An analogous argument shows that the same is true if $(g(1,\bfb_{-j}))^{-1}(\rho_{j})=j$.

If $(g(0,\bfb_{-j}))^{-1}(\rho_{j}) = \rho_{j}$, then $A_{\rho_{j}}\big(0,\bfb_{-\{j,\rho_{j}\}}\big)=\emptyset$ by~\eqref{eq:def-g}, $j\not\in \bar{A}_{\rho_{j}}(\bfb_{-\{j,\rho_{j}\}})$ by~\eqref{eq:available-pos-0}, $A_{\rho_{j}}\big(1,\bfb_{-\{j,\rho_{j}\}}\big)\subseteq\left\{ k\in [n]\colon k<\rho_{j}\right\}$ by~\eqref{eq:available-pos-1}, and either $(g(1,\bfb_{-j}))^{-1}(\rho_{j}) = \rho_{j}$ or $(g(1,\bfb_{-j}))^{-1}(\rho_{j})<\rho_{j}$ by~\eqref{eq:def-g};~\eqref{eq:g-blocking-sets-mon} follows in both cases. Analogously, if $(g(1,\bfb_{-j}))^{-1}(\rho_{j}) = \rho_{j}$, then $A_{\rho_{j}}\big(1,\bfb_{-\{j,\rho_{j}\}}\big) = \emptyset$ by~\eqref{eq:def-g}, $j \not\in \bar{A}_{\rho_{j}}(\bfb_{-\{j,\rho_{j}\}})$ by~\eqref{eq:available-pos-1}, $A_{\rho_{j}}\big(0,\bfb_{-\{j,\rho_{j}\}}\big) \subseteq \left\{ k\in [n]\colon k>\rho_{j}\right\}$ by~\eqref{eq:available-pos-0}, and either $(g(0,\bfb_{-j}))^{-1}(\rho_{j}) = \rho_{j}$ or $(g(0,\bfb_{-j}))^{-1}(\rho_{j}) > \rho_{j}$ by~\eqref{eq:def-g};~\eqref{eq:g-blocking-sets-mon} again follows in both cases.
    
Finally, if $(g(0,\bfb_{-j}))^{-1}(\rho_{j}) \not\in \{j,\rho_{j}\}$ and $(g(1,\bfb_{-j}))^{-1}(\rho_{j}) \not\in \{j,\rho_{j}\}$, it follows from~\eqref{eq:def-g} that $(g(0,\bfb_{-j}))^{-1}(\rho_{j}) \in A_{\rho_{j}}(0,\bfb_{-\{j,\rho_{j}\}}) \setminus \{j\}$ and $(g(1,\bfb_{-j}))^{-1}(\rho_{j}) \in A_{\rho_{j}}(1,\bfb_{-\{j,\rho_{j}\}}) \setminus \{j\}$, and from~\eqref{eq:available-pos-0} and~\eqref{eq:available-pos-1} that
\[
    (g(1,\bfb_{-j}))^{-1}(\rho_{j}) < \rho_{j} < (g(0,\bfb_{-j}))^{-1}(\rho_{j}),
\]
and~\eqref{eq:g-blocking-sets-mon} again holds.

\section{Deferred Proofs from \autoref{sec:weak-unanimity}}

\subsection{Proof of \autoref{lem:existence-vertex-coloring}}
\label{app:pf-lem-existence-vertex-coloring}

In order to prove \autoref{lem:existence-vertex-coloring}, we first study the following related problem.
Given $n\in \NN$, we ask for the existence of values $L_i\in \NN$ for $i\in \{0,1,2\}$ and a family $\{S^i_{\ell}\}_{i\in \{0,1,2\}, \ell\in [L_i]}$, with $S^i_\ell \subseteq [n]$ for every $i\in \{0,1,2\}$ and $\ell\in [L_i]$, such that
\begin{align}
    \text{for every } i\in \{0,1,2\} \text{ and } u,v\in [n] \text{ with } u\not= v, \text{ it exists } \ell\in [L_i] \text{ s.t. } S^i_\ell \cap \{u,v\} = \{u\}, \label{eq:cutting-sets-1}\\
    \text{for every } i\in \{0,1,2\} \text{ and } \ell\in [L_i], \ell'\in [L_{i+_3 1}], \text{ it holds that } S^{i+_3 1}_{\ell'} \not\subseteq S^i_\ell. \label{eq:cutting-sets-2}
\end{align}
For a given $n\in \NN$, we denote the set of all families $\{S^i_\ell\}_{i\in \{0,1,2\}, \ell\in [L_i]}$ satisfying~\eqref{eq:cutting-sets-1}-\eqref{eq:cutting-sets-2} for some values of $L_i$ as $\calS(n)$.
We obtain the following lemma.

\begin{lemma}
\label{lem:cutting-sets}
    For every $n\in \NN$ with $n\geq 5$, it holds that $\calS(n) \not= \emptyset$.
\end{lemma}

\begin{proof}   
    In order to show the lemma, we let $n\in \NN$ with $n\geq 5$ and distinguish the cases $n=5$ and $n\geq 6$.

    If $n=5$, we define $L_0=L_1=L_2=4$ and
    \begin{align*}
        S^0_0=\{0,1,2\},\ S^0_1=\{0,3,4\},\ S^0_2=\{1,3\},\ S^0_3=\{2,4\},\\
        S^1_0=\{0,1,3\},\ S^1_1=\{0,2,4\},\ S^1_2=\{1,4\},\ S^1_3=\{2,3\},\\
        S^2_0=\{0,1,4\},\ S^2_1=\{0,2,3\},\ S^2_2=\{1,2\},\ S^2_3=\{3,4\}.
    \end{align*}
    To see that~\eqref{eq:cutting-sets-1} holds for these sets, observe that for every $i\in \{0,1,2\}$ and $u\in [n]$, there are indices $k(i,u),~k'(i,u)\in [L_i]$ such that $S^i_{k(i,u)}\cap S^i_{k'(i,u)}=\{u\}$, as shown in \autoref{tab:indices-cutting-sets-n5}.
    Therefore, for each $i,u,v$ as in~\eqref{eq:cutting-sets-1}, it holds either $S^i_{k(i,u)} \cap \{u,v\} = \{u\}$ or $S^i_{k'(i,u)} \cap \{u,v\} = \{u\}$.
    \begin{table}[]
    \centering
    \begin{tabular}{ccccccccccccccccccc}
    \toprule
    $i$        & \quad & \multicolumn{5}{c}{$0$}     & \quad & \multicolumn{5}{c}{$1$}     & \multicolumn{5}{c}{$2$}     \\
    $u$        & \quad & $0$ & $1$ & $2$ & $3$ & $4$ & \quad & $0$ & $1$ & $2$ & $3$ & $4$ & \quad & $0$ & $1$ & $2$ & $3$ & $4$ \\ 
    \midrule
    $k(i,u)$ & \quad & $0$ & $0$ & $0$ & $1$ & $1$ & \quad & $0$ & $0$ & $1$ & $0$ & $1$ & \quad & $0$ & $0$ & $1$ & $1$ & $0$ \\
    $k'(i,u)$ & \quad & $1$ & $2$ & $3$ & $2$ & $3$ & \quad & $1$ & $2$ & $3$ & $3$ & $2$ & \quad & $1$ & $2$ & $2$ & $3$ & $3$ \\ \bottomrule
    \end{tabular}
    \caption{Values of $k(i,u)$ and $k'(i,u)$ as described in the proof of \autoref{lem:cutting-sets} for the case $n=5$. The column with $i=0$ and $u=3$, for example, states that $S^0_1 \cap S^0_2 = \{3\}$, thus for every $v\in [5]$ with $v\not=3$ we have either that $S^0_1 \cap \{3,v\}=3$ or that $S^0_2\cap \{3,v\}=3$.}
    \label{tab:indices-cutting-sets-n5}
    \end{table}
    It is easy to check that~\eqref{eq:cutting-sets-2} holds as well.
    For example, when $i=0$, we have that $\{0,1,3\}\not\subseteq S^0_\ell$ for every $\ell \in [4], ~\{0,2,4\}\not\subseteq S^0_\ell$ for every $\ell\in [4], ~\{1,4\}\not\subseteq S^0_\ell$ for every $\ell\in [4]$, and $\{2,3\}\not\subseteq S^0_\ell$ for every $\ell\in [4]$. 
    For $i\in \{1,2\}$ this condition follows analogously. 

    If $n\geq 6$, we define $L_1=L_2=L_3=n$ and
    \begin{align*}
        S^0_\ell=\{\ell,\ell+_n 1,\ell+_n 2\} \text{ for every } \ell\in [n],\\
        S^1_\ell=\{\ell,\ell+_n 2,\ell+_n 3\} \text{ for every } \ell\in [n],\\
        S^2_\ell=\{\ell,\ell+_n 3,\ell+_n 4\} \text{ for every } \ell\in [n].
    \end{align*}
    In this case,~\eqref{eq:cutting-sets-2} follows from the fact that all sets have the same size and $S^i_\ell\not= S^j_{\ell'}$ for every $i,j\in \{0,1,2\}$ with $i\not=j$ and every $\ell,\ell'\in [n]$, because for every $n\geq 6$ and every $\ell\in [n]$ we have that $\ell+_n 5 \not = \ell$.
    
    To see that~\eqref{eq:cutting-sets-1} holds as well, we let $i\in \{0,1,2\}$ and $u,v\in [n]$ with $u\not= v$ be arbitrary values. We prove that there exists $\ell\in [n]$ such that $S^i_\ell \cap \{u,v\}=\{u\}$ distinguishing on the value of $i$.
    
    If $i=0$ and $v\not\in \{u+_n 1, u+_n 2\}$, then $v\not\in S^0_u=\{u,u+_n 1,u+_n 2\}$ and thus $S^0_u \cap \{u,v\} = \{u\}$.
    If $i=0$ and $v\in \{u+_n 1, u+_n 2\}$, then since $n\geq 6$ we have that $v\not\in S^0_{u-_n 2}=\{u-_n 2, u-_n 1, u\}$ and thus $S^0_{u-_n 2} \cap \{u,v\} = \{u\}$.

    If $i=1$ and $v\not\in \{u+_n 2, u+_n 3\}$, then $v\not\in S^1_u=\{u,u+_n 2,u+_n 3\}$ and thus $S^1_u \cap \{u,v\} = \{u\}$.
    If $i=1$ and $v\in \{u+_n 2, u+_n 3\}$, then since $n\geq 6$ we have $v\not\in S^1_{u-_n 2}=\{u-_n 2, u, u+_n 1\}$ and thus $S^2_{u-_n 2} \cap \{u,v\} = \{u\}$.

    We finally consider the case $i=2$. If $v\not\in \{u+_n 3, u+_n 4\}$, then $v\not\in S^2_u=\{u,u+_n 3,u+_n 4\}$ and thus $S^2_u \cap \{u,v\} = \{u\}$.
    If $v = u+_n 3$ and $n\not=6$, then $v\not\in S^2_{u-_n 3}=\{u-_n 3, u, u+_n 1\}$ and thus $S^2_{u-_n 3} \cap \{u,v\} = \{u\}$.
    If $v = u+_n 3$ and $n=6$, then $v\not\in S^2_{u-_n 4}=\{u-_n 4, u-_n 1, u\}$ and thus $S^2_{u-_n 4} \cap \{u,v\} = \{u\}$.
    If $v = u+_n 4$ and $n\not=7$, then $v\not\in S^2_{u-_n 3}=\{u-_n 3, u, u+_n 1\}$ and thus $S^2_{u-_n 3} \cap \{u,v\} = \{u\}$.
    Finally, if $v = u+_n 4$ and $n=7$, then $v\not\in S^2_{u-_n 4}=\{u-_n 4, u-_n 1, u\}$ and thus $S^2_{u-_n 4} \cap \{u,v\} = \{u\}$.
\end{proof}

\autoref{fig:matrix-coloring-n6} illustrates the solutions constructed in the proof of \autoref{lem:cutting-sets} for $n=6$, complementing that for $n=5$ shown in Part (b) of \autoref{fig:matrix-coloring-n5}.
\begin{figure}[t]
    \centering
    \begin{tikzpicture}[scale=0.8]
    \useasboundingbox (-2.7,-2.7) rectangle (2.7,2.7);

\tikzset{
    contour/.style={
    black, 
      double=lightgray,
      opacity=0.4,
      double distance=#1,
      cap=round,
      rounded corners=0.5mm,
    }
  }
\pgfdeclarelayer{bg}
\pgfsetlayers{bg,main}
    \foreach \i/\c in {0/color0,1/color1,2/color2,3/color3,4/color4,5/color5}{\draw (\i*360/6:2cm) node[circle,fill=\c](\i){\i};}
    \begin{pgfonlayer}{bg}
    \draw[component,draw=Brown,line width=1.5cm](0.center) to[bend left=15] (1.center) node[above right=-3mm] {\textcolor{Brown}{$S^0_0$}} to[bend left=15]  (2.center) to[bend left=15] cycle;
    \draw[component,draw=Violet,line width=1.4cm](1.center) to[bend left=15] (2.center) node[above left=-3mm] {\textcolor{Violet}{$S^0_1$}} to[bend left=15] (3.center) to[bend left=15] cycle;
    \draw[component,draw=MidnightBlue,line width=1.4cm] (2.center) to[bend left=15] (3.center) node[left] {\textcolor{MidnightBlue}{$S^0_2$}} to[bend left=15](4.center) to[bend left=15] cycle;
    \draw[component,draw=OliveGreen,line width=1.2cm](3.center) to[bend left=15] (4.center) node[below left=-2mm] {\textcolor{OliveGreen}{$S^0_3$}} to[bend left=15] (5.center) to[bend left=15] cycle ;
    \draw[component,draw=Thistle,line width=1.1cm] (4.center) to[bend left=15] (5.center) node[below right=-2mm] {\textcolor{Thistle}{$S_4^0$}} to[bend left=15] (0.center) to[bend left=15] cycle ;
    \draw[component,draw=Emerald,line width=1.0cm] (5.center) to[bend left=15] (0.center) node[right=3mm] {\textcolor{Emerald}{$S_5^0$}} to[bend left=15] (1.center)  to[bend left=15] cycle ;
    \end{pgfonlayer}

\end{tikzpicture}
\hfill
\begin{tikzpicture}[scale=0.8]
\useasboundingbox (-2.7,-2.7) rectangle (2.7,2.7);

\tikzset{
    contour/.style={
      black, 
      double=lightgray,
      opacity=0.4,
      double distance=#1,
      cap=round,
      rounded corners=0.5mm,
    }
  }
\pgfdeclarelayer{bg}
\pgfsetlayers{bg,main}

    \foreach \i/\c in {0/color0,1/color1,2/color2,3/color3,4/color4,5/color5}{\draw (\i*360/6:2cm) node[circle,fill=\c](\i){\i};}
    \begin{pgfonlayer}{bg}
    \draw[component,draw=Brown,line width=1.5cm] (0.center) to[bend left=15] (2.center) to[bend left=20] node[above left] {\textcolor{Brown}{$S^1_0$}} (3.center)  to[bend left=15] cycle;
    \draw[component,draw=Violet,line width=1.4cm](1.center) to[bend left=15] (3.center)  to[bend left=15] node[below left=-2mm] {\textcolor{Violet}{$S^1_1$}} (4.center) to[bend left=15] cycle;
    \draw[component,draw=MidnightBlue,line width=1.3cm](2.center) to[bend left=15] (4.center) to[bend left=15] node[below] {\textcolor{MidnightBlue}{$S^1_2$}} (5.center)  to[bend left=15] cycle;
    \draw[component,draw=OliveGreen,line width=1.2cm](3.center) to[bend left=15] (5.center) to[bend left=15] node[below right=-2mm] {\textcolor{OliveGreen}{$S^1_3$}} (0.center) to[bend left=15] cycle;
    \draw[component,draw=Thistle,line width=1.1cm](4.center) to[bend left=15] (0.center) to[bend left=15] node[above right] {\textcolor{Thistle}{$S^1_4$}} (1.center) to[bend left=15] cycle;
    \draw[component,draw=Emerald,line width=1.0cm](5.center) to[bend left=15] (1.center) to[bend left=15] node[above] {\textcolor{Emerald}{$S^1_5$}} (2.center) to[bend left=15] cycle ;
    \end{pgfonlayer}

\end{tikzpicture}
\hfill
\begin{tikzpicture}[scale=0.8]
\useasboundingbox (-2.7,-2.7) rectangle (2.7,2.7);
\tikzset{
    contour/.style={
      black, 
      double=lightgray,
      opacity=0.4,
      double distance=#1,
      cap=round,
      rounded corners=0.5mm,
    }
  }
\pgfdeclarelayer{bg}
\pgfsetlayers{bg,main}
   \foreach \i/\c in {0/color0,1/color1,2/color2,3/color3,4/color4,5/color5}{\draw (\i*360/6:2cm) node[circle,fill=\c](\i){\i};}
    \begin{pgfonlayer}{bg}
    \draw[component,draw=Brown,line width=1.5cm](0.center) to[bend right=15] (4.center) to[bend right=15] node[below left] {\textcolor{Brown}{$S^2_0$}} (3.center) to[bend right=15] cycle;
    \draw[component,draw=Violet,line width=1.4cm](1.center) to[bend right=15] (5.center) to[bend right=15] node[below] {\textcolor{Violet}{$S^2_1$}} (4.center) to[bend right=15] cycle;
    \draw[component,draw=MidnightBlue,line width=1.3cm](2.center) to[bend right=15] (0.center) to[bend right=15] node[below right=-2mm] {\textcolor{MidnightBlue}{$S^2_2$}} (5.center) to[bend right=15] cycle;
    \draw[component,draw=OliveGreen,line width=1.2cm](3.center) to[bend right=15] (1.center) to[bend right=15] node[above right] {\textcolor{OliveGreen}{$S^2_3$}} (0.center) to[bend right=15] cycle;
    \draw[component,draw=Thistle,line width=1.1cm](4.center) to[bend right=15] (2.center) to[bend right=15] node[above] {\textcolor{Thistle}{$S^2_4$}} (1.center)  to[bend right=15] cycle;
    \draw[component,draw=Emerald,line width=1.0cm](5.center)to[bend right=15] (3.center) to[bend right=15] node[above left] {\textcolor{Emerald}{$S_5^2$}} (2.center) to[bend right=15] cycle;
    \end{pgfonlayer}
\end{tikzpicture}

\vspace{.5cm}

        \caption{Illustration of the family $\{S^i_\ell\}_{i\in \{0,1,2\}, \ell\in [6]}$ in $\calS(6)$ constructed in the proof of \autoref{lem:cutting-sets}. As in Part (b) of \autoref{fig:matrix-coloring-n5}, sets $\{S^i_\ell\}_{\ell\in [6]}$ for each $i\in \{0,1,2\}$ are represented as hyperedges of the corresponding hypergraph $G_i$.}
        \label{fig:matrix-coloring-n6}
    \end{figure}
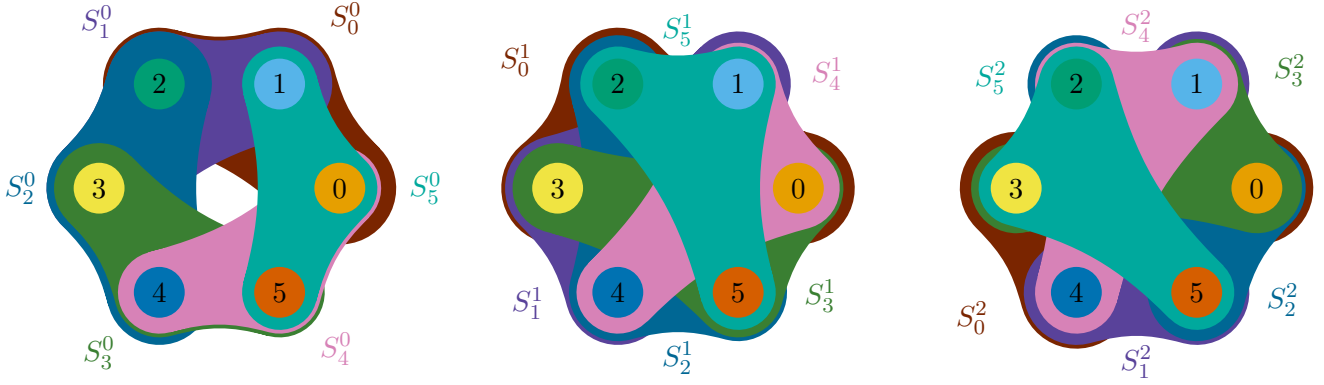
We are now ready to prove \autoref{lem:existence-vertex-coloring}.

\begin{proof}[Proof of \autoref{lem:existence-vertex-coloring}]

    Let $n,m,\bfd^1,\bfd^2,\bfd^3$ be as in the statement of the lemma. We consider values $L_0,~L_1,~L_2$ and the corresponding family $\{S^i_\ell\}_{i\in \{0,1,2\}, \ell\in [L_i]} \in \calS(n)$, whose existence is guaranteed due to \autoref{lem:cutting-sets}.
    We define matrices $\bfA^0,\bfA^1,\bfA^2\in [n]^{m\times m}$ as follows:
    \begin{enumerate}[label=(\arabic*)]
        \item For every $i\in \{0,1,2\}$ and every $p\in [m]$, fix $A^i_{pp} = d^i_p$.\label{item:def-matrices-1}
        \item For every $i\in \{0,1,2\}$ and every $p\in [m]$, let $\ell \in [L_i]$ be such that $S^i_{\ell} \cap \left\{d^i_p, d^{i+_3 1}_p\right\} = d^i_p$ and define $\ell(i,p)=\ell$.\label{item:def-matrices-2}
        \item For every $i\in \{0,1,2\}$ and $p,q\in [m]$ with $p\not=q$, fix $A^{i}_{pq} \in S^{i}_{\ell(i,q)} \cap (S^{i-_3 1}_{\ell(i-_3 1,p)})^C$ arbitrarily.\label{item:def-matrices-3}
    \end{enumerate}

    Intuitively, we start fixing the diagonals according to the vectors $\bfd^i$ for each $i\in \{0,1,2\}$, and for the remaining entries, we do the following. We let $S^{i}_{\ell(i,q)}$ be the feasible set of values for column $q$ of matrix $\bfA^i$, where $\ell(i,q)$ is taken such that $d^i_q$ belongs to this set but $d^{i+_3 1}_q$ does not. We take its complement as the feasible set of values for row $q$ of matrix $\bfA^{i+_3 1}$. We finally pick, for each non-diagonal entry, any value that is feasible both for the corresponding row and the corresponding column.
    
    We first show the correctness of this procedure, in the sense that $\ell\in [L_i]$ as described in Step~\ref{item:def-matrices-2} always exists for each $i\in \{0,1,2\}$, and that the intersection computed in Step~\ref{item:def-matrices-3} is non-empty.
    To see the former, note that from the statement of the lemma, we know that for every $i\in \{0,1,2\}$ and $p\in [m]$ it holds that $d^i_p \not= d^{i+_3 1}_p$. Since the family $\{S^i_\ell\}_{i\in \{0,1,2\}, \ell\in [L_i]}$ satisfies~\eqref{eq:cutting-sets-1}, we obtain that for each $i\in \{0,1,2\}$, there is $\ell\in [L_i]$ such that $S^i_{\ell} \cap \left\{d^i_p, d^{i+_3 1}_p\right\} = d^i_p$.
    To see that the intersection computed in Step~\ref{item:def-matrices-3} is non-empty, let $i\in \{0,1,2\}$ and $p,q\in [m]$ with $p\not= q$ be arbitrary and observe that, from~\eqref{eq:cutting-sets-2}, $S^{i+_3 1}_{\ell(i+_3 1,q)} \not\subseteq S^i_{\ell(i,p)}$, thus $S^{i+_3 1}_{\ell(i+_3 1,q)} \cap (S^i_{\ell(i,p)})^C \not= \emptyset$.
    We conclude that the matrices $\bfA^1,\bfA^2,\bfA^3\in [n]^{m\times m}$ constructed via Steps~\ref{item:def-matrices-1}-\ref{item:def-matrices-3} are well-defined.  

    We now prove that $(\bfA^0,\bfA^1,\bfA^2)\in \calA(n,m,\bfd)$.
    Indeed, Condition~\eqref{eq:vertex-coloring-1} in the definition of this set follows directly from Step~\ref{item:def-matrices-1}.
    To see Condition~\eqref{eq:vertex-coloring-2}, we need to check that $A^i_{pq} \not= A^{i+_3 1}_{qr}$ for every $i\in \{0,1,2\}$ and $p,q,r\in [m]$.
    Let $i\in \{0,1,2\}$ and $p,q,r\in [m]$ be arbitrary.
    If $p=q=r$, the fact that $A^i_{pq}\not= A^{i+_3 1}_{qr}$ follows directly from the condition $d^i_q\not= d^{i+_3 1}_q$ in the statement of the lemma, together with Step~\ref{item:def-matrices-1}.
    If $p=q\not= r$,~\eqref{eq:vertex-coloring-2} is equivalent to $A^{i+_3 1}_{qr} \not= d^i_q$ due to Step~\ref{item:def-matrices-1}.
    Step~\ref{item:def-matrices-3} implies $A^{i+_3 1}_{qr} \not\in S^i_{\ell(i,q)}$ and Step~\ref{item:def-matrices-2} implies $d^i_q \in S^i_{\ell(i,q)}$, thus the condition follows.
    Similarly, if $p\not=q= r$,~\eqref{eq:vertex-coloring-2} is equivalent to $A^i_{pq} \not= d^{i+_3 1}_q$ due to Step~\ref{item:def-matrices-1}.
    Step~\ref{item:def-matrices-3} implies $A^i_{pq} \in S^i_{\ell(i,q)}$ and Step~\ref{item:def-matrices-2} implies $d^{i+_3 1}_q \not\in S^i_{\ell(i,q)}$, thus the condition follows.
    Finally, if $p\not= q \not= r$, Step~\ref{item:def-matrices-3} implies both $A^i_{pq} \in S^i_{\ell(i,q)}$ and $A^{i+_3 1}_{qr} \not\in S^i_{\ell(i,q)}$, thus the condition follows once again.
    This concludes the proof of the lemma.
\end{proof}

\subsection{Proof of \autoref{thm:weak-unanimity}}
\label{app:pf-thm-weak-unanimity}

    We start by introducing some notation that allows us to map a permutation in $\calP_n$ to an index in $[n!]$ and vice versa. For $n\in \NN$, we define $\kappa_n\colon \calP_n\to [n!]$ such that $\kappa_n(\pi)$ is the lexicographical index of the permutation $\pi\in \calP_n$, \ie the position of $\pi$ in the list of permutations of $n$ elements arranged in lexicographical order. Note that, for each $n\in \NN$, this function is a bijection and, for a value $p\in [n!]$, $\kappa^{-1}_n(p)$ corresponds to the permutation of $n$ elements in the $p$th position of this list.
    
    Let $n\geq 5$ be arbitrary. We omit the subindex $n$ in the function $\kappa_n$ in the remainder of the proof for ease of notation. For $i\in \{0,1,2\}$, we let $\bfd^i\in [n]^{n!}$ be the vector defined as
    \begin{equation}
        d^i_p = (\kappa^{-1}(p))^{-1}(i) \quad \text{for every } p\in [n!],
        \label{eq:definition-d}
    \end{equation}
    \ie $\bfd^i$ is the vector whose $p$th entry is the position of agent $i$ in the ranking $\kappa^{-1}(p)$. Since this definition implies that $d^i_p\not=d^j_p$ for every $i,j\in \{0,1,2\}$ with $i\not=j$ and every $p\in [n!]$, we have by \autoref{lem:existence-vertex-coloring} that $\calA(n,n!,\bfd)\not=\emptyset$.
    Let then $(\bfA^0,\bfA^1,\bfA^2) \in \calA(n,n!,\bfd)$ be arbitrary and let, for $i\in [n]\setminus \{0,1,2\}$, $g_i\colon[n!]^3 \to [n]$ be defined as follows:
    \begin{equation}
        g_i(p,q,r) = \begin{cases}
            (\kappa^{-1}(p))^{-1}(i) & \text{ if } p=q=r,\\[.3cm]
            ([n]\setminus \{A^0_{qr},A^1_{rp},A^2_{pq}\})_{i-3} & \text{ otherwise,}
        \end{cases} \quad \text{ for every } (p,q,r)\in [n!]^3,\label{eq:definition-g}
    \end{equation}
    In a slight abuse of notation, we let here $([n]\setminus \{A^0_{qr},A^1_{rp},A^2_{pq}\})_{i-3}$ denote the $(i-3)$th element of the tuple obtained by considering the elements of $[n]\setminus \{A^0_{qr},A^1_{rp},A^2_{pq}\}$ in increasing order.
    
    We can now properly define the mechanism. Let $f\colon \calP_n^n\to \calP_n$ be such that, for every $\bfpi\in \calP^n_n$ and $i\in [n]$,
    \begin{equation}
        (f(\bfpi))^{-1}(i) = \begin{cases}
            A^{i}_{\kappa(\pi_{i+_3 1}) \kappa(\pi_{i+_3 2})} & \text{ if } i\in \{0,1,2\},\\[.3cm]
            g_i(\kappa(\pi_0),\kappa(\pi_1),\kappa(\pi_2)) & \text{ otherwise.}
        \end{cases} \label{eq:definition-f}
    \end{equation}
    We first prove that $f$ is well-defined, \ie that $(f(\bfpi))^{-1}$ as defined in~\eqref{eq:definition-f} is a bijection from~$[n]$ to~$[n]$ for every $\bfpi \in \calP^n_n$.
    
    \begin{claim}
    \label{claim:f-well-defined}
        For every $\bfpi\in \calP^n_n$, it holds that $\{(f(\bfpi))^{-1}(i)\colon i\in [n]\} = [n]$.
    \end{claim}

    \begin{proof}

    To verify the claim, we first note that, for each $\bfpi\in \calP_n^n$ and $i\in [n]$, it holds that $(f(\bfpi))^{-1}(i) \in [n]$. If $i\in \{0,1,2\}$ this follows from the fact that $\bfA^{i}$ has entries in $[n]$ due to the definition of the set $\calA(n,n!,\bfd)$. For $i\in [n]\setminus \{0,1,2\}$, it follows since the image of $g_i$ is a subset of $[n]$ from the definition of this function in~\eqref{eq:definition-g}.
    
    We prove in the following that for every $\bfpi\in \calP_n^n$ and $i,j\in [n]$ with $i\not= j$, we have that $(f(\bfpi))^{-1}(i) \not= (f(\bfpi))^{-1}(j)$. Together with the previous observation, this allows us to conclude the claim. 
    
    Let $\bfpi\in \calP_n^n$ and $i,j\in [n]$ with $i\not= j$ be arbitrary.
    If $i$ and $j$ are both in $\{0,1,2\}$, we have either $j=i+_3 1$ or $i=j+_3 1$; we assume the former w.l.o.g. By Condition~\eqref{eq:vertex-coloring-2} with $p=\kappa(\pi_{i+_3 1}),~ q=\kappa(\pi_{i+_3 2})$, and $r=\kappa(\pi_{i})$, we have that 
    \[
        A^i_{\kappa(\pi_{i+_3 1})\kappa(\pi_{i+_3 2})} \not= A^{i+_3 1}_{\kappa(\pi_{i+_3 2})\kappa(\pi_{i})} \Longleftrightarrow A^{i}_{\kappa(\pi_{i+_3 1}) \kappa(\pi_{i+_3 2})} \not= A^{j}_{\kappa(\pi_{j+_3 1}) \kappa(\pi_{j+_3 2})}.
    \]
    This condition, together with the definition of~$f$ in~\eqref{eq:definition-f}, yields $(f(\bfpi))^{-1}(i) \not= (f(\bfpi))^{-1}(j)$.
    
    If $|\{i,j\} \cap \{0,1,2\}|=1$, say w.l.o.g.\ $i\in \{0,1,2\}$ and $j\in [n]\setminus \{0,1,2\}$, we distinguish two cases.
    
    If $\pi_1=\pi_2=\pi_3$, then denoting $\hat{\pi}=\pi_1$ we have that 
        \[
            (f(\bfpi))^{-1}(i) = A^{i}_{\kappa(\hat{\pi})\kappa(\hat{\pi})} = d^i_{\kappa(\hat{\pi})} = (\kappa^{-1}(\kappa(\hat{\pi})))^{-1}(i) = \hat{\pi}^{-1}(i),
        \]
        where the first equality follows from the definition of $f$ in~\eqref{eq:definition-f}, the second one from the fact that $(\bfA^0,\bfA^1,\bfA^2)\in \calA(n,n!,\bfd)$ and Condition~\eqref{eq:vertex-coloring-1} in the definition of this set, and the third one from the definition of $\bfd^{i}$ in~\eqref{eq:definition-d}.
        On the other hand, 
        \[
            (f(\bfpi))^{-1}(j) = g_j(\kappa(\hat{\pi}),\kappa(\hat{\pi}),\kappa(\hat{\pi}))=(\kappa^{-1}(\kappa(\hat{\pi})))^{-1}(j) = \hat{\pi}^{-1}(j),
        \]
        where the first equality follows from the definition of $f$ in~\eqref{eq:definition-f} and the second one from the definition of $g_j$ in~\eqref{eq:definition-g}. Since $\hat{\pi}^{-1}(i) \not= \hat{\pi}^{-1}(j)$, we obtain that $(f(\bfpi))^{-1}(i) \not= (f(\bfpi))^{-1}(j)$, as claimed.
        
    Otherwise, if we do not have $\pi_1=\pi_2=\pi_3$, then $(f(\bfpi))^{-1}(i)=A^{i}_{\kappa(\pi_{i+_3 1}) \kappa(\pi_{i+_3 2})}$ and, from the definitions of $f$ and $g_j$ in~\eqref{eq:definition-f} and~\eqref{eq:definition-g}, respectively, we know that
        \[
            (f(\bfpi))^{-1}(j) = g_j(\kappa(\pi_0),\kappa(\pi_1),\kappa(\pi_2)) \not= A^{i}_{\kappa(\pi_{i+_3 1}) \kappa(\pi_{i+_3 2})}.
        \]
        The inequality $(f(\bfpi))^{-1}(i) \not= (f(\bfpi))^{-1}(j)$ follows as well.
    
    Finally, if $i,~ j\in [n]\setminus \{0,1,2\}$, we again distinguish two cases.
    
    If $\pi_1=\pi_2=\pi_3$, then denoting $\hat{\pi}=\pi_1$ we have that, for $\ell\in \{i,j\}$,
        \[
            (f(\bfpi))^{-1}(\ell) = g_\ell(\kappa(\hat{\pi}),\kappa(\hat{\pi}),\kappa(\hat{\pi}))=(\kappa^{-1}(\kappa(\hat{\pi})))^{-1}(\ell) = \hat{\pi}^{-1}(\ell),
        \]
        where the first equality follows from the definition of $f$ in~\eqref{eq:definition-f} and the second one from the definition of $g_j$ in~\eqref{eq:definition-g}. Since $\hat{\pi}^{-1}(i) \not= \hat{\pi}^{-1}(j)$, we obtain that $(f(\bfpi))^{-1}(i) \not= (f(\bfpi))^{-1}(j)$.
        
        Otherwise, if we do not have $\pi_1=\pi_2=\pi_3$,
        we have from \autoref{eq:definition-g} and~\eqref{eq:definition-f} that, for $\ell\in \{i,j\}$,
        \[
            (f(\bfpi))^{-1}(\ell) = \left([n] \setminus \big\{A^{0}_{\kappa(\pi_{1}) \kappa(\pi_{2})}, A^{1}_{\kappa(\pi_{2}) \kappa(\pi_{0})}, A^{2}_{\kappa(\pi_{0}) \kappa(\pi_{1})} \big\}\right)_{\ell-3}.
        \]
        Since the set on the right-hand side contains each value in $[n]$ at most once and $i\not= j$, we conclude that $(f(\bfpi))^{-1}(i) \not= (f(\bfpi))^{-1}(j)$.

    \end{proof}

    The fact that $f$ is impartial follows easily from its definition; we check it by fixing $\bfpi\in \calP^n_n$ and $i\in [n]$ arbitrarily. 
    If $i\in \{0,1,2\}$, we have from the definition of $f$ in~\eqref{eq:definition-f} that $(f(\bfpi))^{-1}(i)$ only depends on $\pi_{i+_3 1}$ and $\pi_{i+_3 2}$. Since $i\not \in \{i+_3 1, i+_3 2\}$, we conclude that $(f(\bfpi))^{-1}(i)=(f(\tilde{\pi},\bfpi_{-i}))^{-1}(i)$ for every $\tilde{\pi}\in \calP_n$.
    If $i\in [n]\setminus \{0,1,2\}$, we have from the definition of $f$ in~\eqref{eq:definition-f} that $(f(\bfpi))^{-1}(i)$ only depends on $\pi_0$, $\pi_1$, and $\pi_2$, so the fact that $(f(\bfpi))^{-1}(i) = (f(\tilde{\pi},\bfpi_{-i}))^{-1}(i)$ for every $\tilde{\pi}\in \calP_n$ directly follows.

    We finally show that $f$ satisfies weak unanimity. Let $\hat{\pi}\in \calP_n$ be such that $\pi_i=\hat{\pi}$ for every $i\in [n]$.
    For each $i\in \{0,1,2\}$, we have that
    \begin{equation}
        (f(\bfpi))^{-1}(i) = A^{i}_{\kappa(\hat{\pi})\kappa(\hat{\pi})} = d^{i}_{\kappa(\hat{\pi})} = (\kappa^{-1}(\kappa(\hat{\pi})))^{-1}(i) = \hat{\pi}^{-1}(i).\label{eq:weak-unanimity-1}
    \end{equation}
    Indeed, the first equality follows from the definition of $f$ in~\eqref{eq:definition-f}, the second one from the fact that $(\bfA^0,\bfA^1,\bfA^2)\in \calA(n,n!,\bfd)$ and Condition~\eqref{eq:vertex-coloring-1} in the definition of this set, and the third one from the definition of $\bfd^{i}$ in~\eqref{eq:definition-d}.
    For each $i\in [n]\setminus \{0,1,2\}$, we have that
    \begin{equation}
        (f(\bfpi))^{-1}(i) = g_i(\kappa(\hat{\pi}),\kappa(\hat{\pi}),\kappa(\hat{\pi})) = (\kappa^{-1}(\kappa(\hat{\pi})))^{-1}(i) = \hat{\pi}^{-1}(i),\label{eq:weak-unanimity-2}
    \end{equation}
    where the first equality follows from the definition of $f$ in~\eqref{eq:definition-f} and the second one from the definition of $g_i$ in~\eqref{eq:definition-g}.
    Equations~\eqref{eq:weak-unanimity-1} and~\eqref{eq:weak-unanimity-2} yield $f(\bfpi)=\hat{\pi}$, which shows weak unanimity and completes the proof of the theorem.

\section{Deferred Proofs from \autoref{sec:impossibilities}}

\subsection{Proof of \autoref{thm:imp-ifr-n23}}
\label{app:pf-thm-imp-ifr-n23}

    We first consider $n=2$. Suppose that $f$ is an impartial $2$-ranking mechanism, and note that $\calP_2 = \{(0\ \ 1), (1\ \ 0)\}$. Let $\hat{\pi} \in \calP_2$ be such that $f((0\ \ 1), (0\ \ 1)) = \hat{\pi}$. Then, due to impartiality we have
    \begin{align*}
        (f((1\ \ 0), (0\ \ 1)))^{-1}(0) & = (f((0\ \ 1), (0\ \ 1)))^{-1}(0) = \hat{\pi}^{-1}(0),\\
        (f((0\ \ 1), (1\ \ 0)))^{-1}(1) & = (f((0\ \ 1), (0\ \ 1)))^{-1}(1) = \hat{\pi}^{-1}(1),\\
        (f((1\ \ 0), (1\ \ 0)))^{-1}(0) & = (f((0\ \ 1), (1\ \ 0)))^{-1}(0) = \hat{\pi}^{-1}(0),\\
    \end{align*}
    and therefore
    \[
        f((1\ \ 0), (0\ \ 1)) = f((0\ \ 1), (1\ \ 0)) = f((1\ \ 0), (1\ \ 0)) =\hat{\pi}.
    \]
    We conclude that $f$ has a single outcome $\hat{\pi}$, so it does not satisfy individual full rank.
    
    For the case $n=3$, we let $f$ be an impartial $3$-ranking mechanism and make use of the following claim, stating that $f$ cannot output two rankings that are cyclic shifts of each other.

    \begin{claim}
    \label{claim:rotations-n3}
        Let $\bfpi\in \calP_3^3$ be such that $f(\bfpi) = (i_0\ \ i_1\ \ i_2)$. Then, for every $\tilde{\bfpi}\in \calP_3^3$ we have that $f(\tilde{\bfpi}) \not= (i_2 \ \ i_0\ \ i_1)$.
    \end{claim}

    \begin{proof}
        Suppose for the sake of contradiction that $\bfpi, \tilde{\bfpi} \in \calP^3_3$ are such that $f(\bfpi) = (i_0\ \ i_1\ \ i_2)$ and $f(\tilde{\bfpi}) = (i_2 \ \ i_0\ \ i_1)$. Due to impartiality, we know that $(f(\tilde{\pi}_{i_0},\bfpi_{-i_0}))^{-1}(i_0) = (f(\bfpi))^{-1}(i_0) = 0$, thus we have either $f(\tilde{\pi}_{i_0},\bfpi_{-i_0}) = (i_0\ \ i_1\ \ i_2)$ or $f(\tilde{\pi}_{i_0},\bfpi_{-i_0}) = (i_0\ \ i_2\ \ i_1)$. We will show that a contradiction is reached in either case.

        If $f(\tilde{\pi}_{i_0},\bfpi_{-i_0}) = (i_0\ \ i_1\ \ i_2)$, we have by impartiality that 
        \[
            (f(\pi_{i_1},\tilde{\bfpi}_{-i_1}))^{-1}(i_2) = (f(\tilde{\pi}_{i_0},\bfpi_{-i_0}))^{-1}(i_2) = 2.
        \]
        However, impartiality also implies $(f(\pi_{i_1},\tilde{\bfpi}_{-i_1}))^{-1}(i_1) = (f(\tilde{\bfpi}))^{-1}(i_1) = 2$, a contradiction. This is illustrated in Part (a) of \autoref{fig:imp-ifr-n3}.
        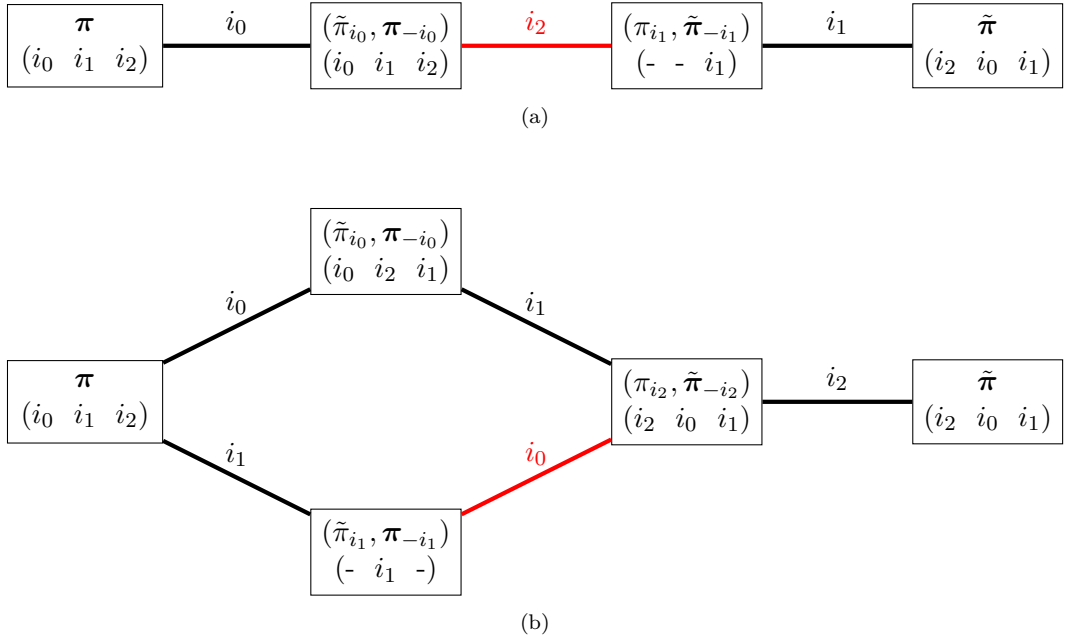
\begin{figure}[t]
        \centering
        \begin{subfigure}[c]{\textwidth}
        \centering
        \begin{tikzpicture}
        \draw (0,0) node[x=-2, rectangle,draw=black, align=center, inner sep=5pt](1){$\bfpi$\\ $(i_0\ \ i_1\ \ i_2)$};
        \draw (4,0) node[x=0, rectangle,draw=black, align=center, inner sep=4pt](2){$(\tilde{\pi}_{i_0},\bfpi_{-i_0})$\\ $(i_0\ \ i_1\ \ i_2)$};
        \draw (8,0) node[x=0, rectangle,draw=black, align=center, inner sep=4pt](3){$(\pi_{i_1},\tilde{\bfpi}_{-i_1})$\\ $(\text{-}\ \ \text{-}\ \ i_1)$};
        \draw (12,0) node[x=0, rectangle,draw=black, align=center, inner sep=4pt](4){$\tilde{\bfpi}$\\ $(i_2\ \ i_0\ \ i_1)$};
        \draw[-,draw,ultra thick] (1) edge node [midway,above] {$i_0$} (2);
        \draw[-,red,ultra thick] (2) edge node [midway,above] {$i_2$} (3);
        \draw[-,draw,ultra thick] (3) edge node [midway,above] {$i_1$} (4);
        \end{tikzpicture}
        \caption{}
        \end{subfigure}

        \vspace{1cm}
        \begin{subfigure}[c]{\textwidth}
        \centering
        \begin{tikzpicture}
        \draw (0,0) node[x=-2, rectangle,draw=black, align=center, inner sep=5pt](1){$\bfpi$\\ $(i_0\ \ i_1\ \ i_2)$};
        \draw (4,2) node[x=0, rectangle,draw=black, align=center, inner sep=4pt](2){$(\tilde{\pi}_{i_0},\bfpi_{-i_0})$\\ $(i_0\ \ i_2\ \ i_1)$};
        \draw (4,-2) node[x=0, rectangle,draw=black, align=center, inner sep=4pt](2a){$(\tilde{\pi}_{i_1},\bfpi_{-i_1})$\\ $(\text{-}\ \ i_1\ \ \text{-})$};
        \draw (8,0) node[x=0, rectangle,draw=black, align=center, inner sep=4pt](3){$(\pi_{i_2},\tilde{\bfpi}_{-i_2})$\\ $(i_2\ \ i_0\ \ i_1)$};
        \draw (12,0) node[x=0, rectangle,draw=black, align=center, inner sep=4pt](4){$\tilde{\bfpi}$\\ $(i_2\ \ i_0\ \ i_1)$};
        \draw[-,draw,ultra thick] (1) edge node [midway,above] {$i_0$} (2);
        \draw[-,draw,ultra thick] (1) edge node [midway,above] {$i_1$} (2a);
        \draw[-,draw,ultra thick] (2) edge node [midway,above] {$i_1$} (3);
        \draw[-,red,ultra thick] (2a) edge node [midway,above] {$i_0$} (3);
        \draw[-,draw,ultra thick] (3) edge node [midway,above] {$i_2$} (4);
        \end{tikzpicture}
        \caption{}
        \end{subfigure}
        \caption{Illustration of the the proof of \autoref{claim:rotations-n3}. Each graph contains ranking profiles with the corresponding output of the mechanism $f$ as vertices and impartiality relations as edges, so that profiles $\bfpi^1$ and $\bfpi^2$ are connected by an edge labeled with $i\in [3]$ if one profile can be obtained from the other through a change on the ranking cast by agent $i$. Outputs follow either from the assumptions in the proof or from impartiality; dashes represent that any unassigned agent may be assigned to the corresponding position. In both cases we suppose that $f$ is a $3$-ranking mechanism satisfying impartiality and individual full rank and that $\bfpi,~ \tilde{\bfpi}\in \calP^3_3$ are profiles such that $f(\bfpi) = (i_0\ \ i_1\ \ i_2)$ and $f(\tilde{\bfpi}) = (i_2 \ \ i_0\ \ i_1)$. Red edges show conflicting profiles, in the sense that the agent able to switch from one profile to another through this edge cannot be in the same position in both profiles. (a) shows the contradiction reached when we further assume $f(\tilde{\pi}_{i_0},\bfpi_{-i_0}) = (i_0\ \ i_1\ \ i_2)$; (b) shows the contradiction reached when we further assume $f(\tilde{\pi}_{i_0},\bfpi_{-i_0}) = (i_0\ \ i_2\ \ i_1)$.} 
        \label{fig:imp-ifr-n3}
        \end{figure}
        
        On the other hand, if $f(\tilde{\pi}_{i_0},\bfpi_{-i_0}) = (i_0\ \ i_2\ \ i_1)$, impartiality implies both
        \[
            (f(\pi_{i_2},\tilde{\bfpi}_{-i_2}))^{-1}(i_1) = (f(\tilde{\pi}_{i_0},\bfpi_{-i_0}))^{-1}(i_1) = 2 \quad \text{and} \quad (f(\pi_{i_2},\tilde{\bfpi}_{-i_2}))^{-1}(i_2) = (f(\tilde{\bfpi}))^{-1}(i_2) = 0.
        \]
        These two equalities yield $f(\pi_{i_2},\tilde{\bfpi}_{-i_2}) = (i_2\ \ i_0\ \ i_1)$. Using impartiality once again, we obtain that $(f(\tilde{\pi}_{i_1},\bfpi_{-i_1}))^{-1}(i_0) = (f(\pi_{i_2},\tilde{\bfpi}_{-i_2}))^{-1}(i_0) = 1$. However, impartiality also implies $(f(\tilde{\pi}_{i_1},\bfpi_{-i_1}))^{-1}(i_1) = (f(\bfpi))^{-1}(i_1) = 1$, a contradiction. This is illustrated in Part (b) of \autoref{fig:imp-ifr-n3}.
    \end{proof}

    \autoref{claim:rotations-n3} implies that $f$ has at most two different outcomes: one among $\{(0\ \ 1\ \ 2),\ (2\ \ 0\ \ 1),\ (1\ \ 2\ \ 0)\}\}$ and one among $\{(0\ \ 2\ \ 1),\ (1\ \ 0\ \ 2),\ (2\ \ 1\ \ 0)\}\}$. Then, for every $i\in \{0,1,2\}$ it holds
    \[
        \big|\left\{(f(\bfpi))^{-1}(i)\colon \bfpi \in \calP^3_3\right\}\big| \leq 2,
    \]
    contradicting the fact that for every $k\in \{0,1,2\}$ there exists $\bfpi\in \calP^3_3$ such that $(f(\bfpi))(k) = i$.

\subsection{Proof of \autoref{thm:imp-unanimity}}
\label{app:pf-thm-imp-unanimity}

    For $n\in \{2,3\}$, the theorem follows directly from \autoref{lem:relation-axioms} and \autoref{thm:imp-ifr-n23}. Let now $n\in \NN$ with $n\geq 4$ and suppose that $f$ is an impartial $n$-ranking mechanism satisfying unanimity. We define profiles $\bfpi^0,\ldots, \bfpi^{n-1}\in \calP_n^n$ as follows:
    \begin{equation}
        \pi^{\ell}_i = \begin{cases}
            (1\ \ 2\ \ \cdots\ \  n-1\ \ 0) & \text{ if } i\in [n-\ell],\\
            (0\ \ 1\ \ \cdots\ \ n-2\ \ n-1) & \text{ if } i\in \{n-\ell,\ldots,n-1\},
        \end{cases} \quad \text{ for every } i, \ell \in [n].\label{eq:def-profiles-R-ell}
    \end{equation}
    Note that, in particular, $\pi^0_i = (1\ \ 2\ \ \cdots\ \  n-1\ \ 0)$ for every $i\in [n]$ and $\pi^{n-1}_i = (0\ \ 1\ \ \cdots\ \ n-2\ \ n-1)$ for every $i\in \{1,2,\ldots,n-1\}$. The following claim, which we prove by induction over $\ell\in \{1,2,\ldots,n-1\}$, allows us to reach a contradiction.
    
    \begin{claim}
    \label{claim:un}
        For every $\ell\in \{1,2,\ldots,n-1\}$ and every $i\in \{1,\ldots,n-\ell\}$, it holds that $(f(\bfpi^{\ell}))^{-1}(i) = i-1$.
    \end{claim}

    \begin{proof}
    We start by observing that, for every $\ell\in [n]$ and $i\in [n]$, we have that
    \[
        (\pi^\ell_i)^{-1}(1) < (\pi^\ell_i)^{-1}(2) < \cdots < (\pi^{\ell}_i)^{-1}(n-1),
    \]
    so unanimity of $f$ implies 
    \begin{equation}
        (f(\bfpi^{\ell}))^{-1}(1) < (f(\bfpi^{\ell}))^{-1}(2) < \cdots < (f(\bfpi^{\ell}))^{-1}(n-1) \quad \text{for every } \ell\in [n].\label{eq:unanimity-ell}
    \end{equation}
    
    We prove the claim by induction over $\ell$. 
    For the base case $\ell=1$, we observe that
    \[
        (f(\bfpi^1))^{-1}(n-1) = (f(\pi^1_{n-1},\bfpi^0_{-(n-1)}))^{-1}(n-1) = (f(\bfpi^0))^{-1}(n-1) = n-2,
    \]
    where the first equality follows from the definition of the profiles in~\eqref{eq:def-profiles-R-ell}, the second one from impartiality, and the last one from the fact that $f(\bfpi^0)=(1\ \ 2\ \ \cdots\ \  n-1\ \ 0)$ by unanimity (recall that unanimity implies weak unanimity by \autoref{lem:relation-axioms}).
    Together with \eqref{eq:unanimity-ell}, this implies that $(f(\bfpi^1))^{-1}(i) = i-1$ for every $i\in \{1,2,\ldots,n-1\}$.

    We now fix $\ell \in \{1,2,\ldots,n-2\}$ and assume that for every $i\in \{1,\ldots,n-\ell\}$ it holds that $(f(\bfpi^{\ell}))^{-1}(i) = i-1$. Then, we obtain
    \begin{align}
        (f(\bfpi^{\ell+1}))^{-1}(n-(\ell+1)) & = \big(f\big(\pi^{\ell+1}_{n-(\ell+1)}, \bfpi^{\ell}_{-(n-(\ell+1))}\big)\big)^{-1}(n-(\ell+1)) \nonumber \\
        & = (f(\bfpi^{\ell}))^{-1}(n-(\ell+1)) \nonumber \\
        & = n-(\ell+1)-1.\label{eq:pos-imp-agent-un}
    \end{align}
    Indeed, the first equality follows from the definition of the profiles in~\eqref{eq:def-profiles-R-ell}, the second one from impartiality, and the last one from the inductive hypothesis. %
    Together with \eqref{eq:unanimity-ell}, this implies that $(f(\bfpi^{\ell+1}))^{-1}(i) = i-1$ for every $i\in \{1,2,\ldots,n-(\ell+1)\}$, concluding the inductive step and thus the proof of the claim.
    \end{proof}

    \autoref{claim:un} implies, in particular, that $(f(\bfpi^{n-1}))^{-1}(1) = 0$, and thus $(f(\bfpi^{n-1}))^{-1}(0) \not= 0$. Defining $\tilde{\pi}=(0\ \ 1\ \ \cdots\ \ n-2\ \ n-1)$, we have by impartiality that
    \[
        (f(\tilde{\pi}, \bfpi^{n-1}_{-0}))^{-1}(0) = (f(\bfpi^{n-1}))^{-1}(0) \not = 0.
    \]
    However, unanimity implies that $f(\tilde{\pi}, \bfpi^{n-1}_{-0}) = (0\ \ 1\ \ \cdots\ \ n-2\ \ n-1)$, so in particular $(f(\tilde{\pi}, \bfpi^{n-1}_{-0}))^{-1}(0)=0$, a contradiction. We conclude that impartiality and unanimity are not compatible. This proof is illustrated in \autoref{fig:imp-unanimity} for the case of $n=4$.

\begin{figure}[t]
  \centering
  \begin{tikzpicture}
        \Text[x=-7,y=3.4]{Agent}
        \Text[x=-7,y=2.8]{$0$}
        \Text[x=-7,y=2.2]{$1$}
        \Text[x=-7,y=1.6]{$2$}
        \Text[x=-7,y=1]{$3$}
        
        \Text[x=-5,y=2.8]{$(1\ \ 2\ \ 3\ \ 0)$}
        \Text[x=-5,y=2.2]{$(1\ \ 2\ \ 3\ \ 0)$}
        \Text[x=-5,y=1.6]{$(1\ \ 2\ \ 3\ \ 0)$}
        \Text[x=-5,y=1]{$(1\ \ 2\ \ 3\ \ 0)$}
        \Text[x=-5]{$(1\ \ 2\ \ 3\ \ 0)$}

        \Text[x=-2,y=2.8]{$(1\ \ 2\ \ 3\ \ 0)$}
        \Text[x=-2,y=2.2]{$(1\ \ 2\ \ 3\ \ 0)$}
        \Text[x=-2,y=1.6]{$(1\ \ 2\ \ 3\ \ 0)$}
        \Text[x=-2,y=1]{$(0\ \ 1\ \ 2\ \ 3)$}
        \Text[x=-2]{$(1\ \ 2\ \ 3\ \ \text{-})$}

        \Text[x=1,y=2.8]{$(1\ \ 2\ \ 3\ \ 0)$}
        \Text[x=1,y=2.2]{$(1\ \ 2\ \ 3\ \ 0)$}
        \Text[x=1,y=1.6]{$(0\ \ 1\ \ 2\ \ 3)$}
        \Text[x=1,y=1]{$(0\ \ 1\ \ 2\ \ 3)$}
        \Text[x=1]{$(1\ \ 2\ \ \text{-}\ \ \text{-})$}

        \Text[x=4,y=2.8]{$(1\ \ 2\ \ 3\ \ 0)$}
        \Text[x=4,y=2.2]{$(0\ \ 1\ \ 2\ \ 3)$}
        \Text[x=4,y=1.6]{$(0\ \ 1\ \ 2\ \ 3)$}
        \Text[x=4,y=1]{$(0\ \ 1\ \ 2\ \ 3)$}
        \Text[x=4]{$(1\ \ \text{-}\ \ \text{-}\ \ \text{-})$}

        \Text[x=7,y=2.8]{$(0\ \ 1\ \ 2\ \ 3)$}
        \Text[x=7,y=2.2]{$(0\ \ 1\ \ 2\ \ 3)$}
        \Text[x=7,y=1.6]{$(0\ \ 1\ \ 2\ \ 3)$}
        \Text[x=7,y=1]{$(0\ \ 1\ \ 2\ \ 3)$}
  \end{tikzpicture}
  \caption{Illustration of the proof of \autoref{thm:imp-unanimity} for $n=4$. For each ranking profile, represented by enumerating the individual rankings cast by each agent, the constraints on the output ranking that impartiality (with respect to the previous profile) and unanimity impose are stated below the profile; dashes represent that any unassigned agent may be assigned to the corresponding position. For the right-most profile, impartiality implies that agent $0$ cannot be in position $0$, while unanimity implies that it must.} 
  \label{fig:imp-unanimity}
\end{figure}

\bibliographystyle{abbrvnat}
\bibliography{ranking-CFK-2023}
	
\end{document}